\documentclass{kais}
\usepackage{wrapfig,epsf}
\usepackage[dcucite]{harvard}
\usepackage{graphicx}
\usepackage{slashbox}
\usepackage{flushend}
\usepackage[caption=false]{subfig}
\usepackage{amssymb,amsfonts,amsmath}
\usepackage{epstopdf}

\newtheorem{defi}{Definition}[section]

\newtheorem{theorem}{Theorem}[section]

\newtheorem{corollary}[theorem]{Corollary}

\received{Jul 29, 2013}
\revised{May 22, 2014}
\accepted{Aug 04, 2014}

\pubyear{2014}
\pagerange{\pageref{firstpage}--\pageref{lastpage}}
\volume{xxx}

\begin{document}
\label{firstpage}

\title{Evaluating Link Prediction Methods}

\author[Y. Yang et al]{Yang Yang$^{1,2}$, Ryan N. Lichtenwalter$^{1,2}$, and Nitesh V. Chawla$^{1,2}$\\ $^1$Department of Computer Science and Engineering, \\
384 Fitzpatrick Hall, University of Notre Dame, Notre Dame, IN 46556, USA; \\
$^2$Interdisciplinary Center for Network Science \& Applications, \\
384 Nieuwland Hall of Science, University of Notre Dame, Notre Dame, IN 46556, USA \\
\{yyang1, rlichten, nchawla\}@nd.edu }
 
\maketitle

\begin{abstract}
Link prediction is a popular research area with important applications in a variety of disciplines, including biology, social science, security, and medicine. The fundamental requirement of link prediction is the accurate and effective prediction of new links in networks. While there are many different methods proposed for link prediction, we argue that the practical performance potential of these methods is often unknown because of challenges in the evaluation of link prediction, which impact the reliability and reproducibility of results. We describe these challenges, provide theoretical proofs and empirical examples demonstrating how current methods lead to questionable conclusions, show how the fallacy of these conclusions is illuminated by methods we propose, and develop recommendations for consistent, standard, and applicable evaluation metrics. We also recommend the use of precision-recall threshold curves and associated areas in lieu of receiver operating characteristic curves due to complications that arise from extreme imbalance in the link prediction classification problem.
\end{abstract}


\section{Introduction}
Link prediction generally stated is the task of predicting relationships in a network \cite{sarukkai:2000,getoor:2003,liben-nowell:2003,taskar:2003,hasan:2005,huang:2005}. Typically it is approached specifically as the task of predicting new links given some set of existing nodes and links. 
Existing nodes and links may be present from a prior time period, where general link prediction is useful to anticipate future behavior \cite{scripps:2008,leroy:2010,lichtenwalter:2010}. Alternatively, existing nodes and links may also represent some portion of the topology in a network whose exact topology is difficult to measure. In this case, link prediction can identify or substantially narrow possibilities that are difficult or expensive to determine through direct experimentation \cite{martinez:1999,sprinzak:2003,szilagyi:2005}. Thus, even in domains where link prediction seems impossibly difficult or offers a high ratio of false positives to true positives, it may be useful \cite{clauset:2008,clauset:2009}. Formally, we can state the link prediction problem as below (first implicitly defined in the work of \cite{liben-nowell:2007}):
\begin{defi}
In a network $G=(V,E)$, $V$ is the set of nodes and $E$ is the set of edges. The link prediction task is to predict whether there is or will be a link $e(u, v)$ between a pair of nodes $u$ and $v$, where $u, v \in V$ and $e(u, v) \notin E$.
\end{defi}
\indent {Generally the link prediction problem falls into two categories:}
\begin{itemize}
\item {
Predict the links that will be added to an observed network in the future. In this scenario, the link prediction task is applicable to predicting future friendship or future collaboration, for instance, and it is also informative for exploring mechanisms underlying network evolution.
}

\item {
Infer missing links from an observed network. The prediction of missing links is mostly used to identify lost or hidden links, such as inferring unobserved protein-protein interactions.
}

\end{itemize}
The prediction of future links considers network evolution while the inference of missing links considers a static network \cite{liben-nowell:2007}. In either of the two scenarios, instances in which the link forms or is shown to exist compose the positive class, and instances in which the link does not form or is shown not to exist compose the negative class. The negative class represents the vast majority of the instances, and the positive class is a small minority.


\subsection{The Evaluation Conundrum} 

Link prediction entails all the complexities of evaluating ordinary binary classification for imbalanced class distributions, but it also includes several new parameters and intricacies that make it fundamentally different. Real-world networks are often very large and sparse, involving many millions or billions of nodes and roughly the same number of edges. Due to the resulting computational burden, \textit{test set sampling} is common in link prediction evaluation \cite{liben-nowell:2003,hasan:2005,murata:2007}. Such sampling, when not properly reflective of the original distribution, can greatly increase the likelihood of biased evaluations that do not meaningfully indicate the true performance of link predictors. The selected \textit{evaluation metric} can have a tremendous bearing on the apparent quality and ranking of predictors even with proper testing distributions \cite{raeder:2010}. The \textit{directionality} of links also introduces issues that do not exist in typical classification tasks. Finally, for tasks involving network evolution, such as predicting the appearance of links in the future, the classification process involves \textit{temporal aspects}. Training and testing set constructs must appropriately address these nuances.

To better describe the intricacies of evaluation in the link prediction problem, we first depict the framework for evaluation \cite{wang:2007,omadadhain:2005a,omadadhain:2005b,lichtenwalter:2010} in Figure~\ref{fig:linkpred}. Computations occur within network snapshots based on particular segments of data. Comparisons among predictors require that evaluation encompasses precisely the same set of instances whether the predictor is unsupervised or supervised. We construct four network snapshots:

\begin{itemize}
	\item Training features: data from some period in the past, $G_{t-x}$ up to $G_{t-1}$, from which we derive feature vectors for training data.
	\item Training labels: data from $G_{t}$, the last training-observable period, from which we derive class labels, whether the link forms or not, for the training feature vectors.
	\item Testing features: data from some period in the past up to $G_{t}$, from which we derive feature vectors for testing data. Sometimes it may be ideal to maintain the window size that we use for the training feature vector, so we commence the snapshot at $G_{t-x+1}$. In other cases, we might want to be sure not to ignore effects of previously existing links, so we commence the snapshot at $G_{t-x}$.
	\item Testing labels: data from $G_{t+1}$, from which we derive class labels for the testing feature vector. This data is strictly excluded from inclusion in any training data.
\end{itemize}

A classifier is constructed from the training data and evaluated on the testing data. There are always strictly divided training and testing sets, because $G_{t+1}$ is never observable in training.

Note that for supervised methods, the division of the training data into a feature and label network is not strictly necessary. Edges from the training data may be used to calculate features for instances with a positive class label, and missing edges in the training data might be used to calculate features for instances with negative class labels (Figure~\ref{fig:toy_example}). Nonetheless, division into a feature and label network may increase both the freshness and power of the data being modeled, because the links that appear are recent and selected according to the underlying evolution process. Testing data must \emph{always} be divided in feature and label networks where labels are drawn from unobserved data in $G_{t+1}$.


\begin{figure}
	\centering
	\includegraphics[width=0.75\linewidth]{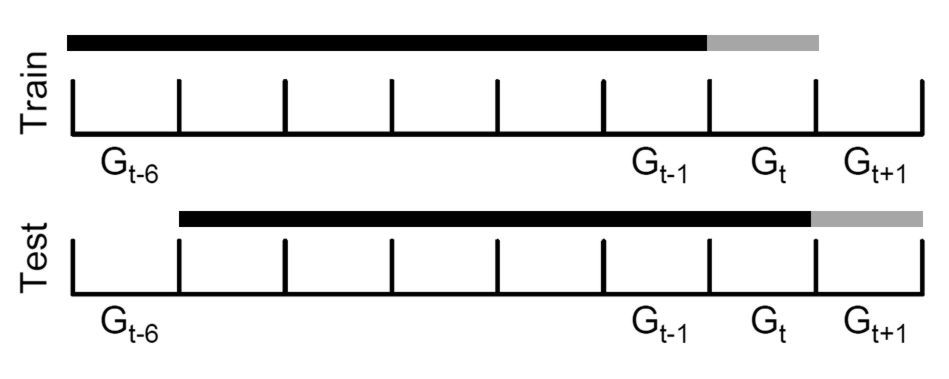}
	\caption{Link prediction and evaluation. The black color indicates snapshots of the network from which link prediction features are calculated (feature network). The gray color indicates snapshots of the network from which the link prediction instances are labeled (label network). We can observe all links at or before time $t$, and we aim to predict future links that will occur at time $t+1$.}
	\label{fig:linkpred}
\end{figure}

\begin{figure}
	\centering
	\includegraphics[width=3.8in]{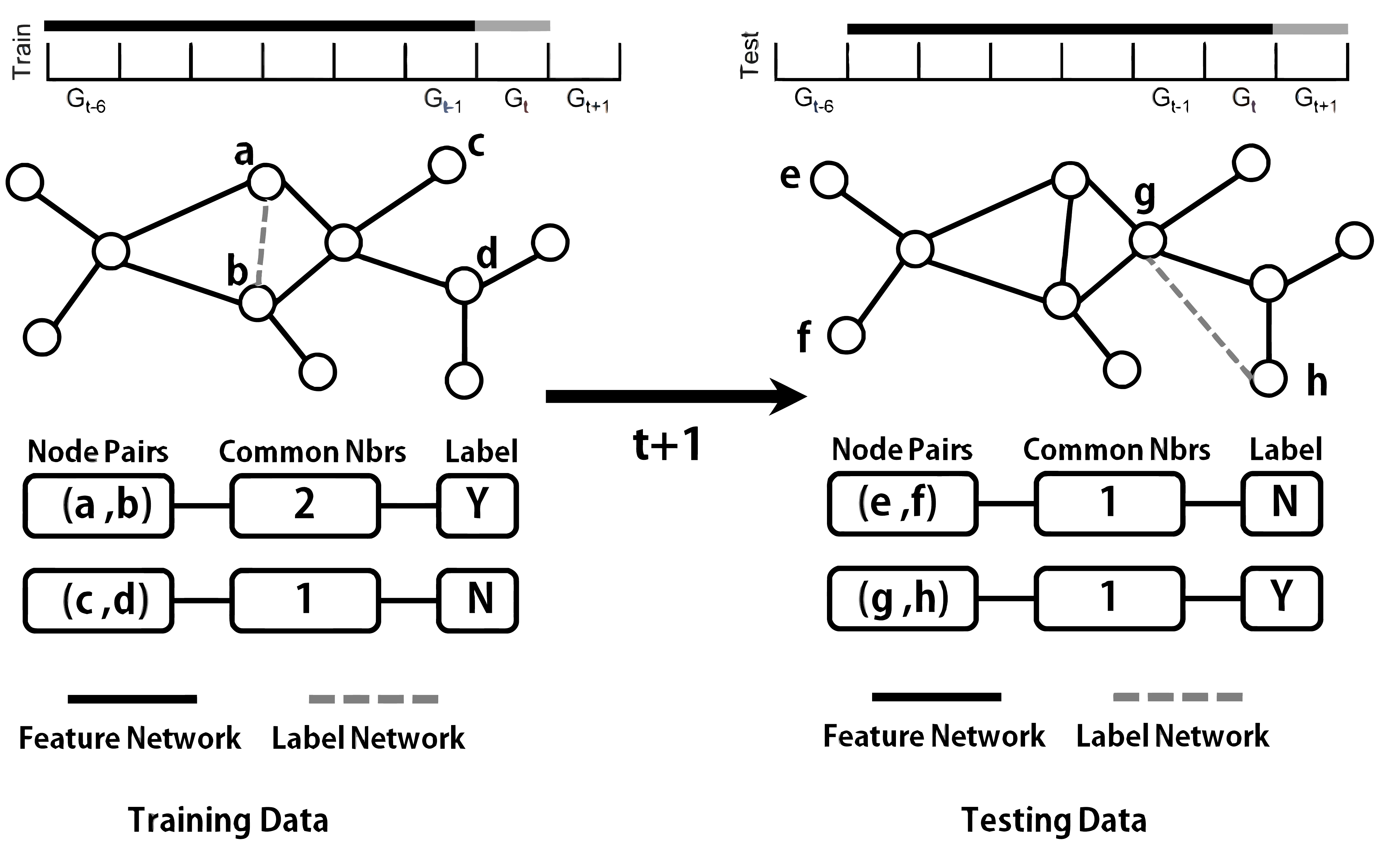}
	\caption{Link prediction and evaluation toy example. This graph provides a toy example for the construction of training and testing sets in the link prediction problem. Black links correspond to the black period in Figure~\ref{fig:linkpred}, from which link prediction features were extracted. Gray links correspond to the gray future period in Figure~\ref{fig:linkpred}, which is used to label link prediction instances. Based on the black observed network we find that the \textit{common neighbors} values for $\left(a,b\right)$ and $\left(c,d\right)$ are $2$ and $1$ respectively. Based on the gray label network we know that the link between $a$ and $b$ forms and the link between $c$ and $d$ does not form. The same methodology is applied to construct the testing data.}
	\label{fig:toy_example}
\end{figure}

\textbf{Test set sampling} is a common practice in link prediction evaluation \cite{wang:2007,lichtenwalter:2010,hasan:2005,leskovec:2010,murata:2007,liben-nowell:2003,narayanan:2011,scripps:2008,scellato:2011}. Link prediction should be evaluated with a complete, unsampled test set whenever possible. Randomly removing and subsequently predicting ``test edges'' should be a last resort in networks where multiple snapshots do not exist. Even in networks that do not evolve, such as protein-protein interaction networks, it is possible to use high-confidence and low-confidence edges to construct different networks to more effectively evaluate models. Removing and predicting edges can remove information from the original network in unpredictable ways \cite{stumpf:2005}, and the removed information has the potential to affect prediction methods differently. More significantly, in longitudinal networks, randomly sampling edges for testing from a single data set in a supervised approach reflects prediction performance with respect to a random process instead of the true underlying regulatory mechanism.

The reason sampling seems necessary, and a primary reason link prediction is such a challenging domain within which to evaluate and interpret performance, is \textbf{extreme class imbalance}. We extensively analyze issues related to sampling in Section \ref{sec:sampling} and cover the significance of class imbalance on evaluation in Section~\ref{sec:imbalance}. Fairly and effectively evaluating a link predictor in the face of imbalance requires determining which \textbf{evaluation metric} to use (Sections \ref{sec:evaluation}, \ref{sec:topk} and \ref{sec:prcurve}), whether to restrict the enormous set of potential predictions, and how best to restrict the set if so.

Another issue is \textbf{directionality}, for which there is no analog in typical classification tasks. In undirected networks, the same method may predict two different results for a link between $v_a$ and $v_b$ depending on the order in which vertices are presented. There must be one final judgment of whether the link will form, but that judgment differs depending on arbitrary assignment of source and target identities. We expand upon this in Section \ref{sec:directionality}.

All of these issues impede the production of fair, comparable results across published methods. Perhaps even more importantly, they interfere with rendering judgments of performance that indicate what we might really expect of our prediction methods in deployment scenarios. It is difficult to compare from one paper to the next, and many frequently employed evaluation methods produce results that are unfairly favorable to a particular method or otherwise unrepresentative of expected deployment performance. We seek to provide a reference for important issues to consider and a set of recommendations to those performing link prediction research.

\subsection{Contributions}
We explore a range of issues and challenges that are germane to the evaluation of link predictors. Some of them are discussed in our previous work \cite{lichtenwalter:2012}, of which this paper is a substantial expansion. We introduce several formalisms entirely absent from the previous work and provide much more principled coverage of underlying issues and challenges with evaluating link prediction methods. In addition to more complete coverage of previously published topics, the extension includes the following significant advances over existing work:

\begin{itemize}
\item Additional data sets for stronger empirical demonstration.
\item Theoretical proof of several statements surrounding evaluation.
\item Discussion of the advantages and disadvantages of another popular metric of link prediction evaluation, the top $K$ predictive rate.
\item Exploration of evaluation characteristics when considering link prediction according to temporal distance. This is related to the existing study of evaluation characteristics when considering link prediction according to geodesic distance.
\end{itemize}

The point of this work is not to illustrate the superiority of one method of link prediction over another, which distinguishes it from most previous link prediction publications. Our objective is to identify fair and effective methods to evaluate link prediction performance. Overall, \textbf{our contributions} are summarized as follows:
\begin{itemize}
	\item We discuss the challenges in evaluating link prediction algorithms.
	\item We empirically demonstrate the effects of test set sampling on link prediction evaluation and offer related proofs.
	\item We demonstrate that commonly used evaluation metrics lead to deceptive conclusions with respect to link prediction results. We additionally show that precision-recall curves are a fair and consistent way to view, understand, and compare link prediction results.
	\item We propose guidelines for a fair and effective framework for link prediction evaluation.
\end{itemize}

\section{Preliminaries}
\subsection{Data and Methods}
\label{sec:data}
We report all results on four publicly available longitudinal data sets. We will hence refer to these data sets as Condmat \cite{newman:2001}, DBLP \cite{deng:2011}, Enron \cite{leskovec:2009} and Facebook \cite{viswanath:2009}. They are constructed by moving through sequences of collaboration events (Condmat and DBLP) or communication events (Enron and Facebook). In Condmat and DBLP each collaboration of $k$ individuals forms an undirected $k$-clique with weights in inverse linear proportion to $k$. The detailed information about these data sets are presented in Table~\ref{tab:network_properties}. These networks are weighted and undirected.

\begin{table}
\caption{{Network Characteristics}}
\centering
\resizebox{4.2in}{!}{
\begin{tabular}{|l|l|l|l|l|l|l|l|} \hline
\textbf{Networks} & Condmat & DBLP & Enron & Facebook \\ \hline \hline
Nodes & 13,873 & 3,215 & 16,922 & 1,829\\ \hline
Edges & 55,269 & 9,816 & 34,825 & 13,838 \\ \hline
Density & 5.74e-4 & 1.90e-3  & 2.00e-4 & 8.27e-3 \\ \hline
\end{tabular}
}
\label{tab:network_properties}
\end{table}

We use three different link prediction methods, and each method represents a different modeling approach. The preferential attachment predictor \cite{barabasi:1999,barabasi:2002,newman:2001b} uses degree product and represents predictors based on node statistics. The Adamic/Adar predictor \cite{adamic:2001} represents common neighbors predictors. The PropFlow predictor \cite{lichtenwalter:2010} represents predictors based on paths and random walks.

\emph{We emphasize here that the point of this work is not to illustrate the superiority of one method of link prediction over another. It is instead to demonstrate that the described effects and arguments have real impacts on performance evaluation. If we show that the effects pertain in at least one network, it follows that they may exist in others and must be considered}.

\subsection{Definitions and Terminology}
\label{sec_glossary}
\textbf{Network:} A network is represented as $G=(V,E)$, where $V$ is the set of nodes and $E$ is the set of edges. For two nodes $u,v \in V$, $e(u, v) \in E$ if there is a link between nodes $u$ and $v$.\\
\textbf{Neighbors:} In a network $G=(V,E)$, for a node $u$, $\Gamma(u) = \{v | (u, v) \in E\}$ represents the set of neighbors of node $u$. \\
\textbf{Link Prediction:} The link prediction task in a network $G=(V,E)$ is to determine whether there is or will be a link $e(u, v)$ between a pair of nodes $u$ and $v$, where $u, v \in V$ and $e(u, v) \notin E$.\\
\textbf{Common Neighbors:} For two nodes, $u$ and $v$ with sets of neighbors $\Gamma(u)$ and $\Gamma(v)$ respectively, the set of their common neighbors is defined as $\Gamma(u) \cap \Gamma(v)$, and the cardinality of the set is $|\Gamma(u) \cap \Gamma(v)|$. Often as $|\Gamma(u) \cap \Gamma(v)|$ grows, the likelihood that $u$ and $v$ will be connected also increases \cite{newman:2001b}.\\
\textbf{Adamic/Adar:} In the link prediction problem, the Adamic/Adar \cite{adamic:2001} metric is defined as below, where $n \in \Gamma(u) \cap \Gamma(v)$ is the set of common neighbors of $u$ and $v$:
\begin{align*}
&AA(u, v) = \sum_{n \in \Gamma(u) \cap \Gamma(v)} \frac{1}{\log{|\Gamma(n)}|}
\end{align*}
\textbf{Preferential Attachment:} The Preferential Attachment \cite{barabasi:1999} metric is the multiplication of the degrees of nodes $u$ and $v$:
\begin{align*}
&PA(u, v) = |\Gamma(u)| |\Gamma(v)|
\end{align*}
\textbf{PropFlow:} The PropFlow \cite{lichtenwalter:2010} method corresponds to the probability that a restricted, outward-moving random walk starting at $u$ ends at $v$ using link weights as transition probabilities. It produces a score $\text{PropFlow}(u,v)$ that can serve as an estimation of the likelihood of new links.\\
\textbf{Geodesic Distance:} The shortest path length $\ell$ between two given nodes $u$ and $v$.\\
\textbf{Prediction Terminology:} \textbf{TP} stands for \textit{true positives}, \textbf{TN} stands for \textit{true negatives}, \textbf{FP} stands for \textit{false positives}, and \textbf{FN} stands for \textit{false negatives}. \textbf{P} stands for \textit{positive instances}, and \textbf{N} stands for \textit{negative instances}.\\
\textbf{Sensitivity/true positive rate:}
\begin{align*}
\text{Sensitivity} = \frac{|TP|}{|TP|+|FN|}
\end{align*}\\
\textbf{Specificity/true negative Rate:}
\begin{align*}
\text{Specificity} = \frac{|TN|}{|FP|+|TN|}
\end{align*}\\
\textbf{Precision:}
\begin{align*}
\text{Precision} = \frac{|TP|}{|TP|+|FP|}
\end{align*}\\
\textbf{Recall:}
\begin{align*}
\text{Recall} = \frac{|TP|}{|TP|+|FN|}
\end{align*}\\
\textbf{Fallout/false positive rate:} The false positive rate (fallout in information retrieval) is defined as below:
\begin{align*}
\text{Fallout} = \frac{|FP|}{|FP|+|TN|}
\end{align*}\\
\textbf{Accuracy:}
\begin{align*}
\text{Accuracy} = \frac{|TP|+|TN|}{|P|+|N|}
\end{align*}\\
\textbf{Top $K$ Predictive Rate/R-precision:} The top $K$ predictive rate is the percentage of correctly classified positive samples among the top $K$ instances in the ranking produced by a link predictor $\mathcal{P}$. We denote the top $K$ predictive rate as $TPR_{K}$, where $K$ is a definable threshold. $TPR_{K}$ is equivalent to R-precision in information retrieval.\\
\textbf{ROC:} The receiver operating characteristic (ROC) represents the performance trade-off between true positives and false positives at different decision boundary thresholds \cite{mason:2002,fawcett:2004}.\\
\textbf{AUROC:} Area under the ROC curve.\\
\textbf{Precision-recall Curve:} Precision-recall curves are also threshold curves. Each point corresponds to a different score threshold with a different precision and recall value \cite{davis:2006}. \\
\textbf{AUPR:} Area under the precision-recall curve.

\section{Evaluation Metrics and Existing Challenges}
\label{sec:evaluation}
Evaluation metrics typically used in link prediction overlap those used in any binary classification task. They can be divided into two broad categories: fixed-threshold metrics and threshold curves. Fixed-threshold metrics suffer from the limitation that some estimate of a reasonable threshold must be available in the score space. In research contexts, where we are curious about performance without necessarily being attached to any particular domain or deployment, such estimates are generally unavailable. Threshold curves, such as the receiver operating characteristic (ROC) curve \cite{mason:2002,fawcett:2004} and derived curves like cost curves \cite{drummond:2006} and precision-recall curves \cite{davis:2006}, provide alternatives in these cases.

\subsection{Fixed-threshold Metrics}
Fixed-threshold metrics rely on different types of thresholds: prediction score, percentage of instances, or number of instances. In link prediction specifically, there are additional constraints. Some link prediction methods produce poorly calibrated scores. For instance, it often may not hold that two vertices with a degree product of 10,000 are 10 times as likely to form a new link as two with a degree product of 1,000. This is not critical when the goal of the model is to rank; however, when the scores of link predictors are poorly calibrated, it is difficult to select an appropriate threshold for any fixed-threshold performance evaluation metric. A concrete and detailed example is provided in Section~\ref{sec:topk}, where we find that even a minor change of threshold value can lead to a completely different evaluation of link prediction models.

In Figure~\ref{fig:density} we provide the probability density functions and cumulative density functions for two prediction methods (Adamic/Adar and Preferential Attachment) on DBLP. For ease of interpretation, we use the inverse values of Preferential Attachment and Adamic/Adar, so smaller values of Preferential Attachment and Adamic/Adar indicate higher likelihood of occurrence. In Figure~\ref{fig:density} we observe that preferential Attachment and Adamic/Adar have different types of distributions. This makes it difficult to identify a meaningful threshold based on normalized prediction score. Any attempt to select a value-based threshold will produce an unfair comparison between these two prediction methods.

\begin{figure}
	\centering
	\vspace{-0.4cm}
	\subfloat[AA-KDF]{
		\includegraphics[width=0.4\linewidth]{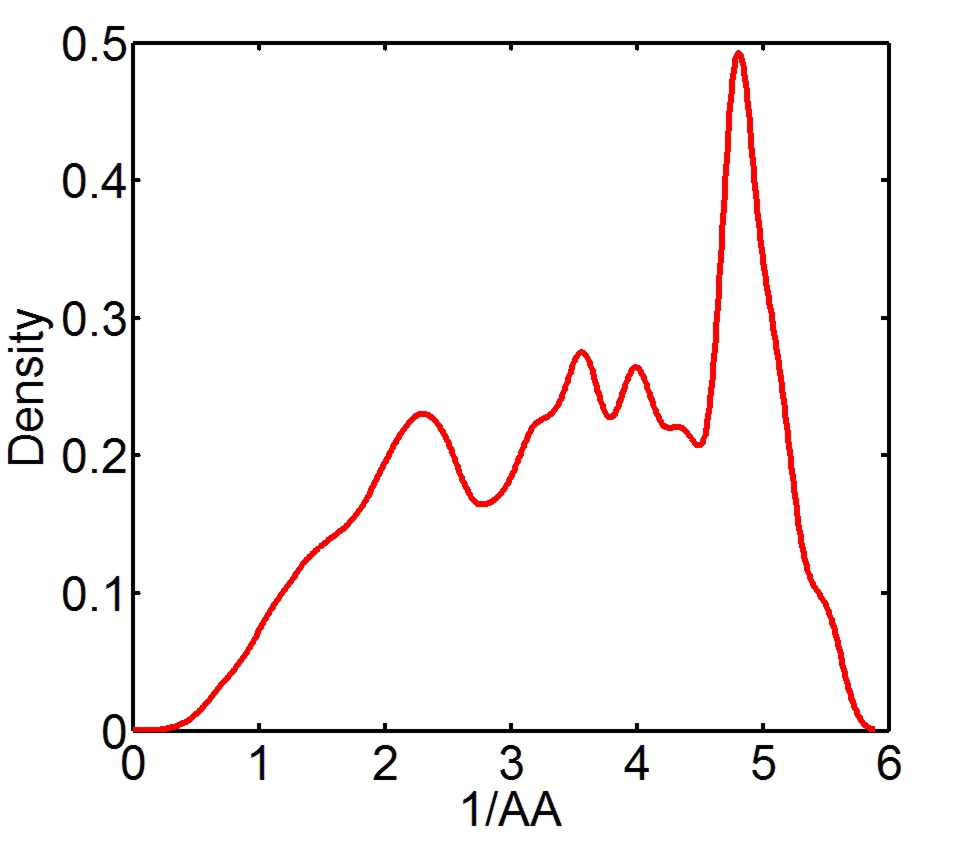}
	}
	\subfloat[PA-KDF]{
		\includegraphics[width=0.4\linewidth]{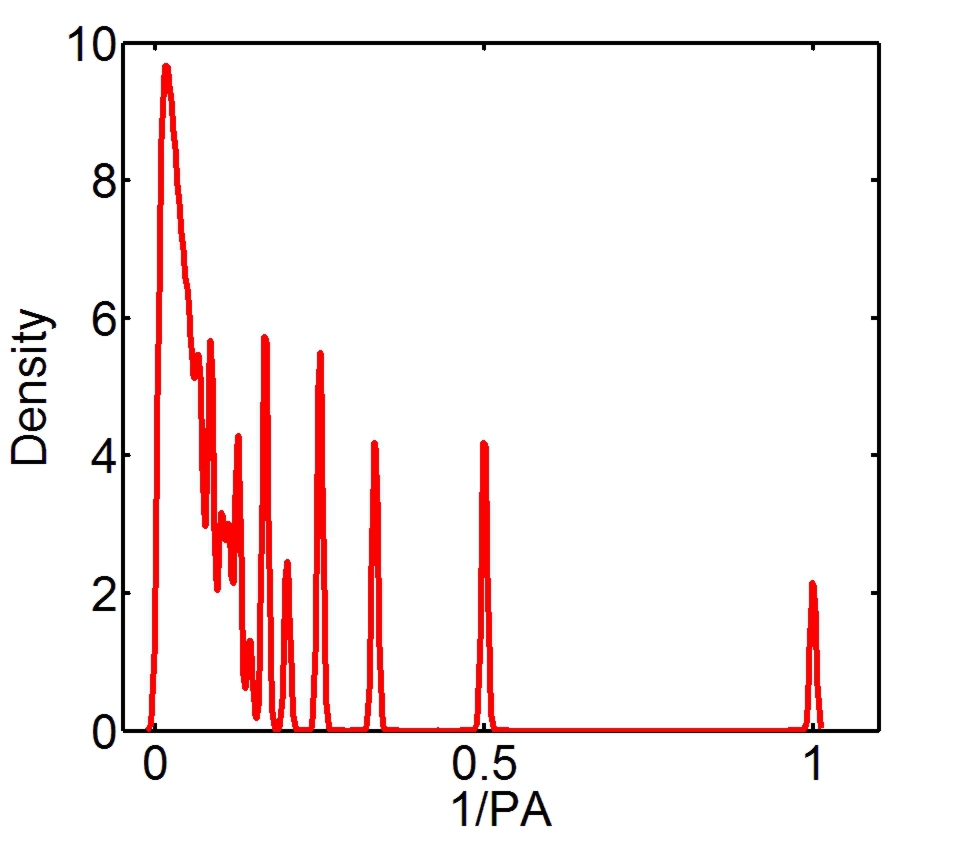}
	}\\
	\subfloat[AA-CDF]{
		\includegraphics[width=0.4\linewidth]{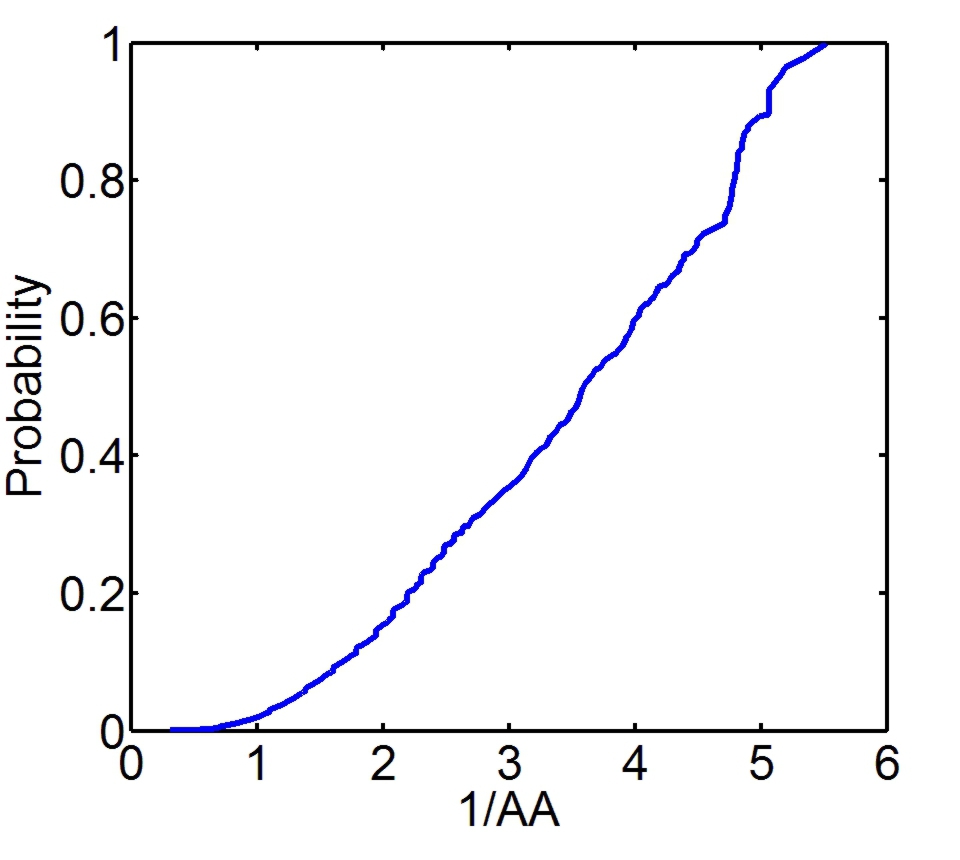}
	}
	\subfloat[PA-CDF]{
		\includegraphics[width=0.4\linewidth]{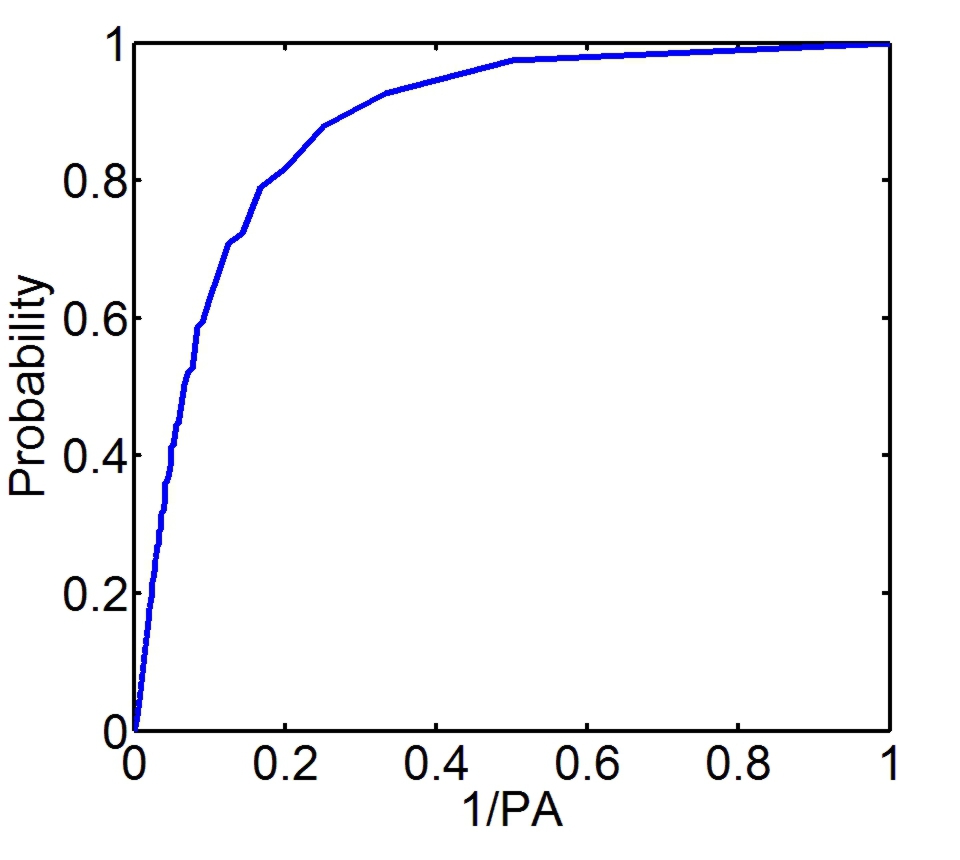}
	}	
	\caption{Kernel Density Function and Cumulative Density Function}
	\label{fig:density}
\end{figure}

A cardinality-based threshold is also problematic. We shall presently advocate exploring links by geodesic distance $\ell$. Within this paradigm, it makes little sense to speak about the $TPR_{K}$ results as a percentage of the total number of potential links in $\ell$. Resources to explore potential links in model deployment scenarios are unlikely to change because the potential positives happen to span a larger distance. It is appropriate to use percentages with $TPR_{K}$ only when $\ell=2$ or $\ell \leq \infty$ and only when there is a reasonable expectation that $K$ is logical within the domain. On the other hand, when using an absolute number of instances, if the data does not admit a trivially simple class boundary, we set classifiers up for unstable evaluation by presenting them with class ratios of millions to one and taking infinitesimal $K$. In Section~\ref{sec:topk} we discuss $TPR_{K}$ evaluation in greater detail.

While \textit{accuracy} \cite{hasan:2005,taskar:2003,fletcher:2011}, \textit{precision} \cite{wang:2007,hasan:2005,omadadhain:2005a,taskar:2003,huang:2005,Wang:2011}, \textit{recall} \cite{hasan:2005,yin:2010,taskar:2003,huang:2005,hopcroft:2011}, and top $K$ equivalents \cite{wang:2007,omadadhain:2005a,backstrom:2011,Wang:2011,dong:2012} are used commonly in link prediction literature, we need to be cautious when results come only in terms of fixed thresholds.

\begin{figure}
	\centering
	\subfloat[ROC Curve]{
		\includegraphics[width=0.40\linewidth]{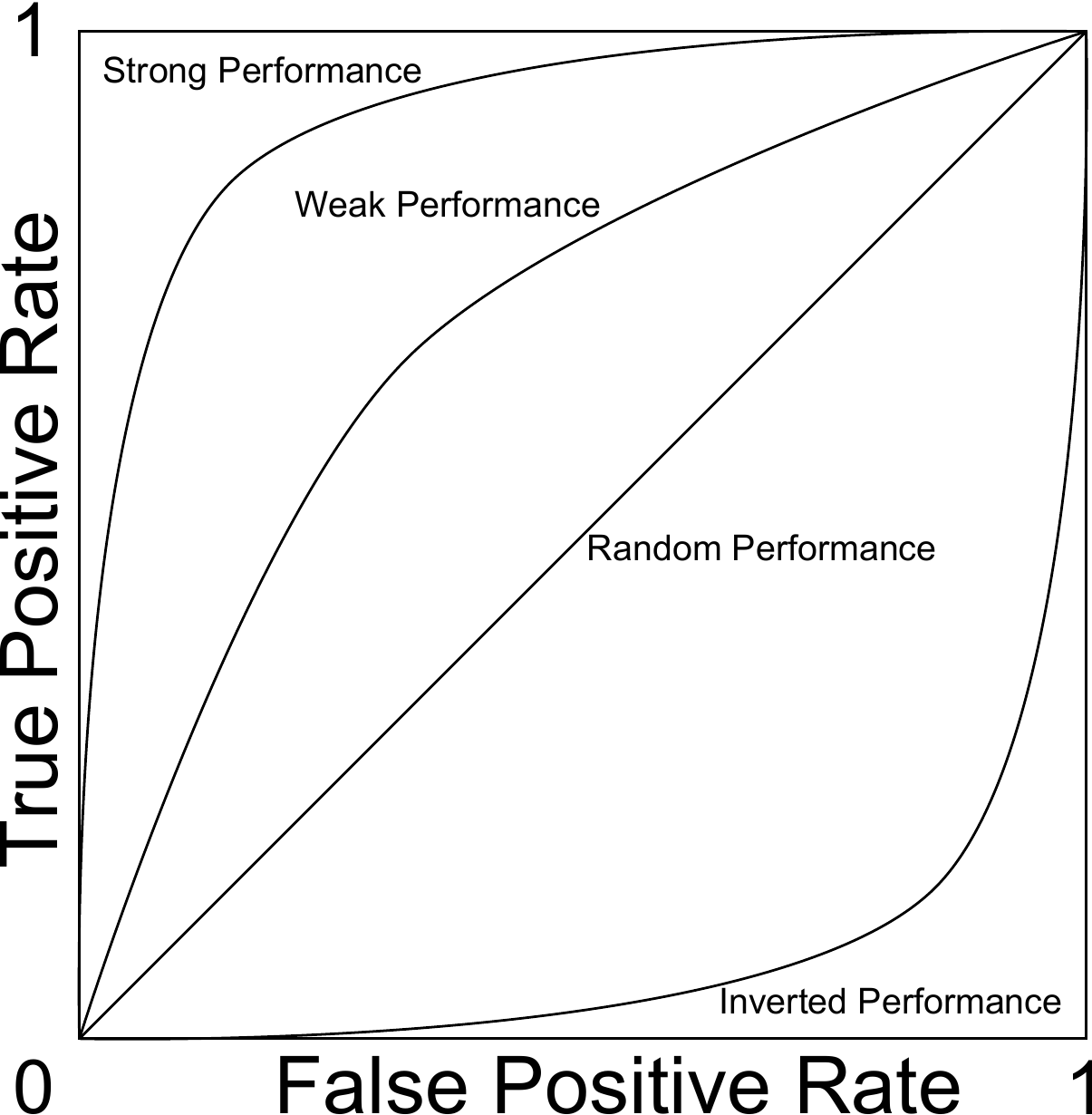}
	}
	\hspace{0.5cm}
	\subfloat[Precision-Recall Curve] {
		\includegraphics[width=0.40\linewidth]{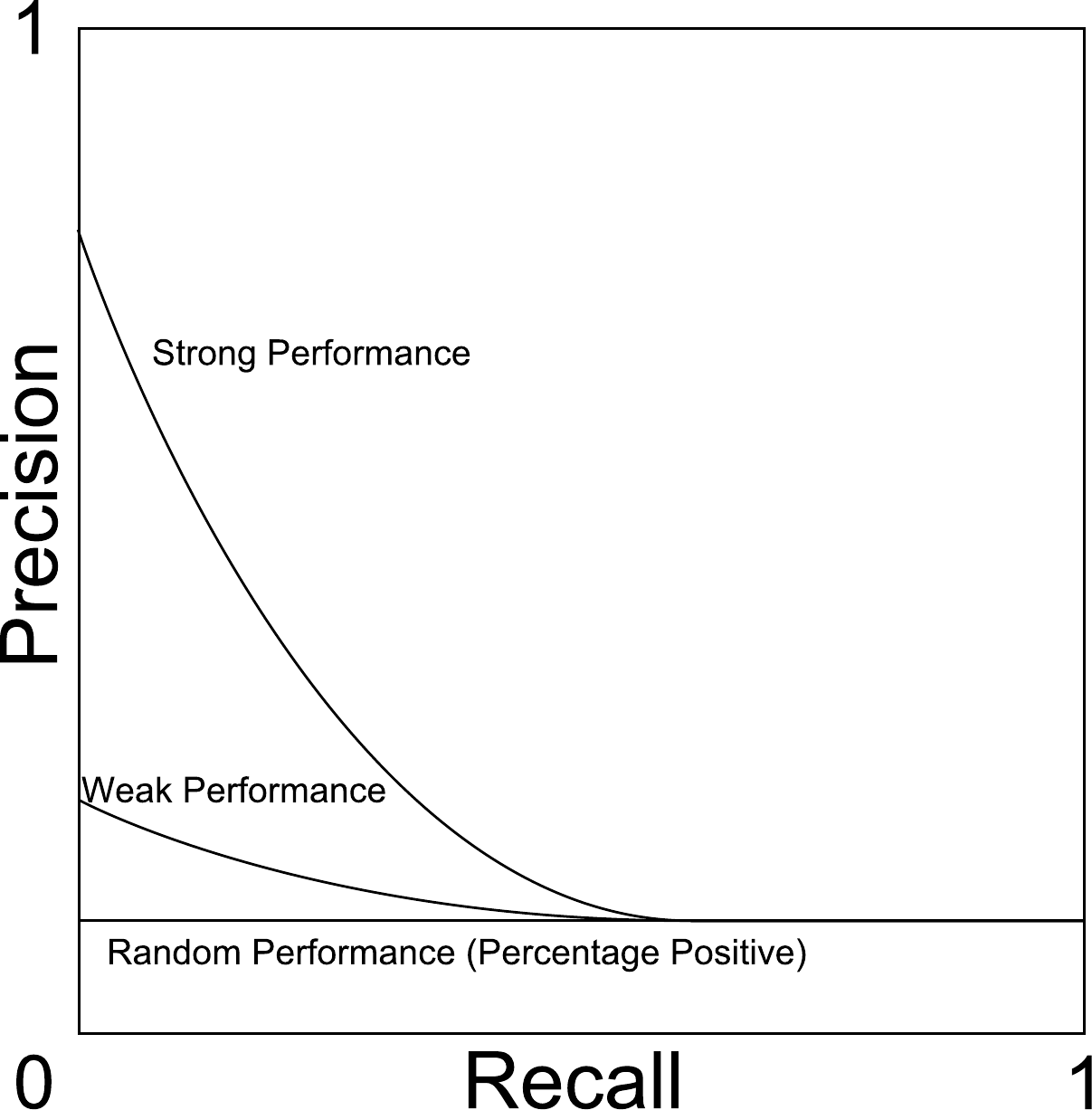}
	}
	\caption{A visualization of the two threshold curves used throughout this paper.}
	\label{fig:curves}
\end{figure}

\subsection{Threshold Curves}
Due to the rarity of cases when researchers are in possession of reasonable fixed thresholds, threshold curves are commonly used in the binary classification community to express results. They are especially popular when the class distribution is highly imbalanced and hence are used increasingly commonly in link prediction evaluation \cite{clauset:2008,wang:2007,lichtenwalter:2010,backstrom:2011,davis:2011,goldberg:2003,Sun:2012,yang:2012}. Threshold curves admit their own scalar measures, which serve as a single summary statistic of performance. The ROC curve shows the true positive rate with respect to the false positive rate at all classification thresholds, and its area (AUROC) is equivalent to the probability of a randomly selected positive instance appearing above a randomly selected negative instance in score space. The precision-recall (PR) curve shows precision with respect to recall at all classification thresholds. It is connected to ROC curves such that one precision-recall curve dominates another if and only if the corresponding ROC curves show the same domination \cite{davis:2006}. We will use ROC curves and precision-recall curves to illustrate our points and eventually argue for the use of precision-recall curves and areas. Figure \ref{fig:curves} illustrates a depiction of the two curve metrics.


\subsection{Class Imbalance}
\label{sec:imbalance}
In typical binary classification tasks, classes are approximately balanced. As a result, we can calculate expectations for baseline classifier performance. For instance, the expected accuracy $\left( \frac{TP+TN}{TP+FP+TN+FN} \right)$, precision $\left( \frac{TP}{TP+FP} \right)$, recall $\left( \frac{TP}{TP+FN} \right)$, AUROC, and AUPR of a random classifier, an all-positive classifier, and an all-negative classifier are 0.5.

Binary classification problems that exhibit class imbalance do not share this property, and the link prediction domain is an extreme example. The expectation for each of the classification metrics diverges for random and trivial classifiers. Accuracy is problematic because its value approaches unity for trivial predictors that always return false. Correct classification of rare positive instances is simultaneously more important since those instances represent exceptional cases of high relative interest. Classification is an exercise in optimizing some measure of performance, so we must not select a measure of performance that leads to a useless result. ROC curves offer a baseline random performance of 0.5 and penalize trivial predictors when appropriate. Optimizing ROC optimizes the production of class boundaries that maximize $TP$ while minimizing $FP$. Precision and precision-recall curves offer baseline performance calibrated to the class balance ratio, and this can present a soberer view of performance.

\subsection{Directionality}
\label{sec:directionality}
In undirected networks, an additional methodological parameter pertains in the task of evaluation, which is rarely reported in literature. In directed networks, an \textit{ordered} pair of vertices uniquely specifies a prediction, because order implies edge direction. In undirected networks, the lack of directionality renders the order ambiguous since two pairs map to one edge. For any given edge, there are two potentially different prediction outputs. For some prediction methods, such as those based on node properties or common neighbors, the prediction output remains the same irrespective of ordering, but this is not true in general. Most notably, many prediction methods based on paths and walks, such as PropFlow \cite{lichtenwalter:2010} and Hitting Time \cite{liben-nowell:2003}, implicitly depend on notions of source and target.
\begin{defi} In an undirected network $G=(V, E)$, for a link prediction method $\mathcal{P}$, if there exists a pair of nodes $u$ and $v$ such that $\mathcal{P} (u,v) \neq \mathcal{P} (v,u)$, then $\mathcal{P}$ is said to be \textit{directional}.
\end{defi}

Contemplate Figure \ref{fig:directionality} with the goal of predicting a link between $v_a$ and $v_b$. Consider the percentage of $\ell=2$ paths starting from $v_a$ that reach $v_b$ versus the percentage of $\ell=2$ paths starting from $v_b$ that reach $v_a$. Clearly, all $\ell=2$ paths originating at $v_b$ reach $v_a$ whereas only a third of the $\ell=2$ paths originating at $v_a$ reach $v_b$. In a related vein, consider the probability of reaching one vertex from another in random walks. Clearly all walks starting at $v_b$ that travel at least two hops must reach $v_a$ whereas the probability of reaching $v_b$ from $v_a$ in two hops is lower. Topological prediction outputs may diverge whenever $v_a$ and $v_b$ are in different automorphism orbits within a shared connected component. 

\begin{figure}
	\centering
	\includegraphics[width=0.3\linewidth]{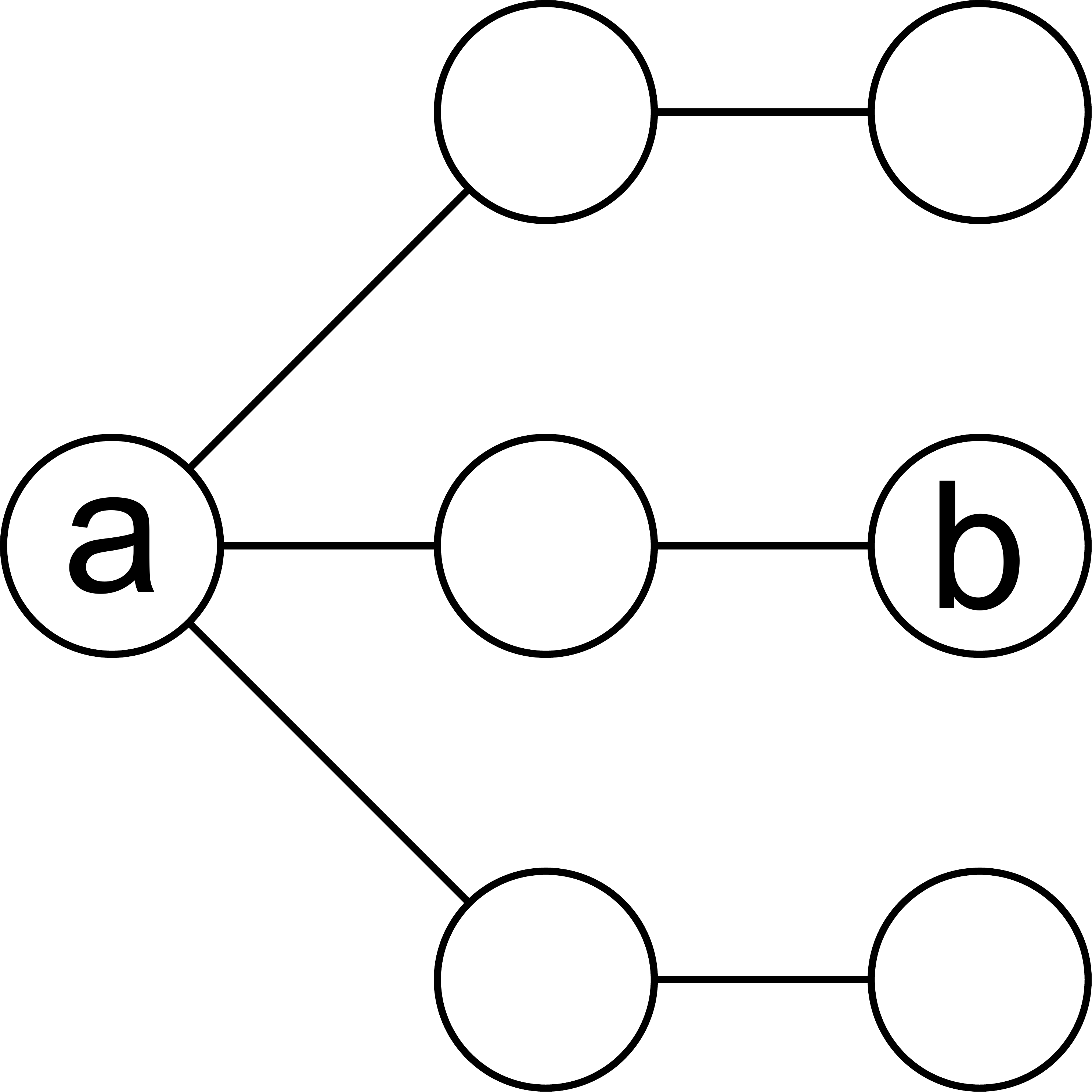}
	\caption{A simple undirected graph example to illustrate that predictions for a link from $v_a$ to $v_b$ may differ according to which vertex is considered as the source.}
	\label{fig:directionality}
\end{figure}

This raises the question of how to determine the final output of a method that produces two different outputs depending on the input. Any functional mapping from two values to a single value suffices. Selection of an optimal method depends on both the predictor and the scenario and is outside the scope of this paper. Nonetheless, it is important for reasons of reproducibility not to neglect this question when describing results for directional predictors.

The approach consistent with the process for directed networks would be to generate a ranked list of scores that includes predictions with each node alternately serving as source and destination. This approach is workable in a deployment scenario, since top ranked outputs may be selected as predicted links regardless of the underlying source and target. It is not feasible as a research method for presenting results, however, because the meaning of the resulting threshold curves is ambiguous. There is also no theoretical reason to suspect any sort of averaging effect in the construction of threshold curves.

To emphasize this empirically, we computed AUROC in the undirected Condmat data set for two methods using the PropFlow predictor (predicting links within 2-hop distance). The PropFlow predictor is directional, so for two nodes $u$ and $v$ it is possible that $\text{PropFlow}\left(u, v\right) \neq \text{PropFlow}\left(v, u\right)$. The first method includes a prediction in the output for both underlying orderings, and the resulting area is 0.610. The second method computes the arithmetic mean of the predictions from the two underlying orderings to produce a single final prediction for the rankings, and the resulting area is 0.625.
\begin{figure}
	\centering
	\includegraphics[width=0.5\linewidth]{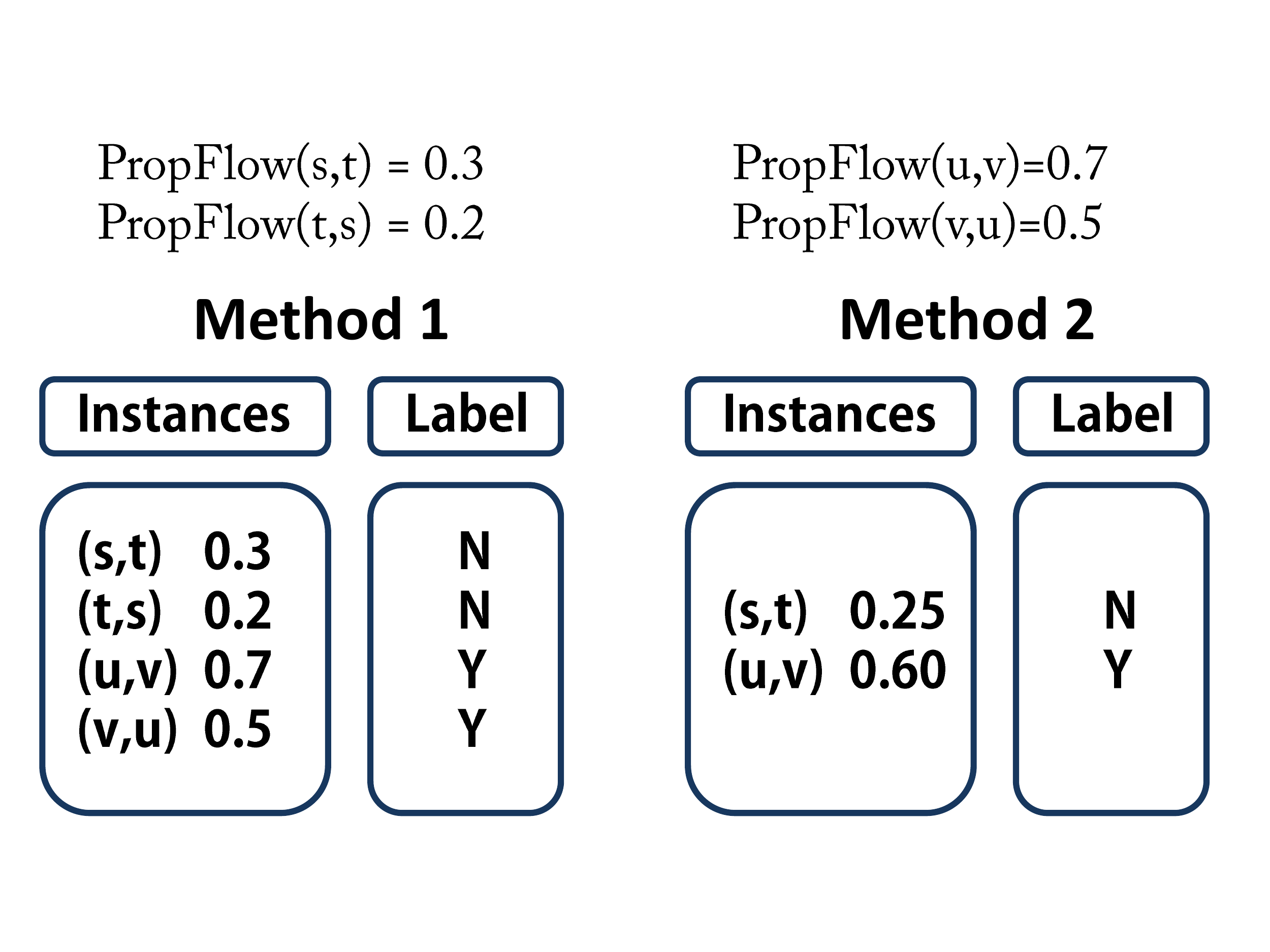}
	\caption{Directionality Toy Example. In this example, for node pairs $\left( s, t\right)$ and $\left( u, v\right)$, $\text{PropFlow(s,t)} \neq \text{PropFlow(t,s)}$ and $\text{PropFlow(u,v)} \neq \text{PropFlow(v,u)}$. In one evaluation method the curve is generated from the ranked list of predictions over both vertex orderings. In the other method the arithmetic mean of the prediction on the two underlying orderings is used to generate the curve. In the Condmat dataset these two methods have different AUROCs: 0.610 and 0.625 respectively.}
	\label{fig:directionality_example}
\end{figure}

\section{Test Set Sampling and Class Imbalance}
\label{sec:sampling}
Test set sampling is popular in link prediction domains, because sparse networks usually include only a tiny fraction of the $O(|V|^2)$ links supported by $V$. Each application of link prediction must provide outputs for what is essentially the entire set of $\binom{|V|}{2}$ links. For even moderately sized networks, this is an enormously large number that places unreasonable practical demands on processing and even storage resources. As a result, there are many sampling methods for link prediction testing sets. One common method is selecting a subset of edges at random from the original complete set \cite{wang:2007,hasan:2005,scripps:2008,Wang:2011,yin:2010,Sun:2012}. Another is to select only the edges that span a particular geodesic distance \cite{lichtenwalter:2010,scellato:2011,yang:2012,lu:2011,scellato:2010}. Yet another is to select edges so that the sub-distribution composed by a particular geodesic distance is approximately balanced \cite{wang:2007,narayanan:2011}. Finally any number of potential methods can select edges that present a sufficient amount of information along a particular dimension \cite{murata:2007,liben-nowell:2003}, for instance selecting only the edges where each member vertex has a degree of at least 2.

When working with threshold-based measures, any sampling method that removes negative class instances above the decision threshold can unpredictably raise most information retrieval measures. Precision is inflated by the removal of false positives. In top $K$ measures, recall is inflated by the opportunity for additional positives to appear above the threshold after the negatives are removed. This naturally affects the harmonic mean, $F$-measure. Accuracy is affected by any test set modification since the number of incorrect classifications may change. Clearly we cannot report meaningful results with these threshold-based measures when performing any type of sampling on the test set. The question is whether it is fair to sample the test set when evaluating with threshold curves.

At first it may seem that subsampling negatives from the test set has no negative effects on ROC curves and areas. There is a solid theoretical basis for this belief, but issues specific to link prediction relating to extreme imbalance cause problems in practice. We will first describe these problems, and why using evaluation methods involving extreme subsampling are problematic. Then we will show that test set sampling actually creates what we believe is a much more significant problem with the testing distribution.

\subsection{Impact of Sampling on ROC}
\label{sec_roc_robustness}
Theoretically, ROC curves and their associated areas are unaffected by changes in class distribution alone. This is a source of great appeal, since it renders consistent judgments even as imbalance becomes increasingly extreme. Consequently, it is theoretically possible to fairly sample negatives from the test set without affecting ROC results. The proper way to model fair random removals of test instances closest to the actual ROC curve construction step is as random selection without replacement from the unsorted or sorted list of output scores. As long as the distribution remains stable in the face of random removals, the ROC curve and area will remain unchanged.

In practice, we do not want to waste the effort necessary to generate lists of output scores only to actually examine a fractional percentage of them. We must instead find a way to transfer our fair model of random removals in the ranked list of output scores to a network sampling method while theoretically preserving all feature distributions. The solution is to randomly sample potential edges. Given a network in which our original evaluation strategy was to consider a test set with every potential edge based on the previously observed network, we generate an appropriately sized random list of edges that do not exist in the test period.

As suggested by Hoeffding's inequality \cite{hoeffding}, in machine learning the test set has finite-sample variance. When the test set is sampled, the performance is as likely to be pleasantly surprising as unpleasantly surprising, though likely not to be surprising at all \cite{learningdata}. We provide a concrete mathematical proof specifically in the link prediction domain. Theoretically we can prove that with a random sampling ratio $p$ of negative instances to final testing set size, the variance of measured performance increases as $p$ decreases. 
\begin{theorem}
\label{lp_stability}
For any link predictor $\mathcal{P}$ the variance of measured performance increases when the negative class sample percentage $p$ decreases.
\end{theorem}

\begin{proof}
For a specific predictor $\mathcal{P}$ we assume that among all $N$ negative instances there are $C$ instances that can be classified correctly by $\mathcal{P}$ while the other $N - C$ instances can not be classified correctly by $\mathcal{P}$.\\
Based on the fact that we randomly sample $Np$ negative instances for inclusion in the final test set, the number of negative instances that can be detected by predictor $\mathcal{P}$ among these $Np$ negative instances is a random variable $X$ that has probability mass function:
\begin{align*}
P(X = x) = \frac{\binom{C}{x}\binom{N-C}{Np-x}}{\binom{N}{Np}}
\end{align*}
Trivially $X$ follows a \textit{\textbf{Hypergeometric}} distribution and the variance of $X$ is:
\begin{align*}
\text{Var}(X) = \frac{C(N-C)p(1-p)}{(N - 1)}
\end{align*}
Since the performance $X/Np$ has variance following:
\begin{align}
\label{var_negative}
\text{Var}\left(\frac{X}{Np}\right) = \frac{C(N-C)(1-p)}{N^{2}(N - 1)p}
\end{align}
it follows that when $p$ decreases $\text{Var}\left(\frac{X}{Np}\right)$ increases.
\end{proof}

From \textbf{Equation}~\ref{var_negative} we can observe that the variance of the measured performance increases linearly with $\frac{1}{p}$.

We demonstrate the results of this random sampling strategy empirically in Figure \ref{fig:undersample1} (unsupervised learning) and Figure~\ref{fig:undersample2} (supervised learning). We conduct these experiments for AUROC but not for AUPR, because precision-recall curves respond to the modifications of the testing set if the class distribution changes \cite{davis:2006}. The experimental settings follow. We explore the effect of sampling of negative class instances in the testing set. We include in our experiments six sampling rates: $10^{-3}$\%, $10^{-2}$\%, $10^{-1}$\%, $10^{0}$\%, $10^{1}$\%, and $10^{2}$\%. For each sampling rate, we randomly sample the testing set 100 times and calculate the AUROC for each. Thus for each sampling rate, we have 100 AUROC scores. In Figure~\ref{fig:undersample1} and Figure~\ref{fig:undersample2} we report the mean, minimum, and maximum of these scores.

The AUROC remains stable down to 1\% of the negative class in Condmat and only down to 10\% of the negative class in Facebook for unsupervised learning. Below this, it destabilizes. While stability to 1\% sampling in Condmat may seem quite good, it is critical to note that the imbalance ratios of link prediction in large networks are such that 1\% of the complete original test set often still contains an unmanageably large number of instances. Similarly in Figure~\ref{fig:undersample2} the AUROC remains stable only down to 10\% negative class sampling for both DBLP and Enron in supervised learning circumstance.

\begin{figure}
	\centering
	\vspace{-0.4cm}
	\subfloat[Condmat]{
		\includegraphics[width=0.47\linewidth]{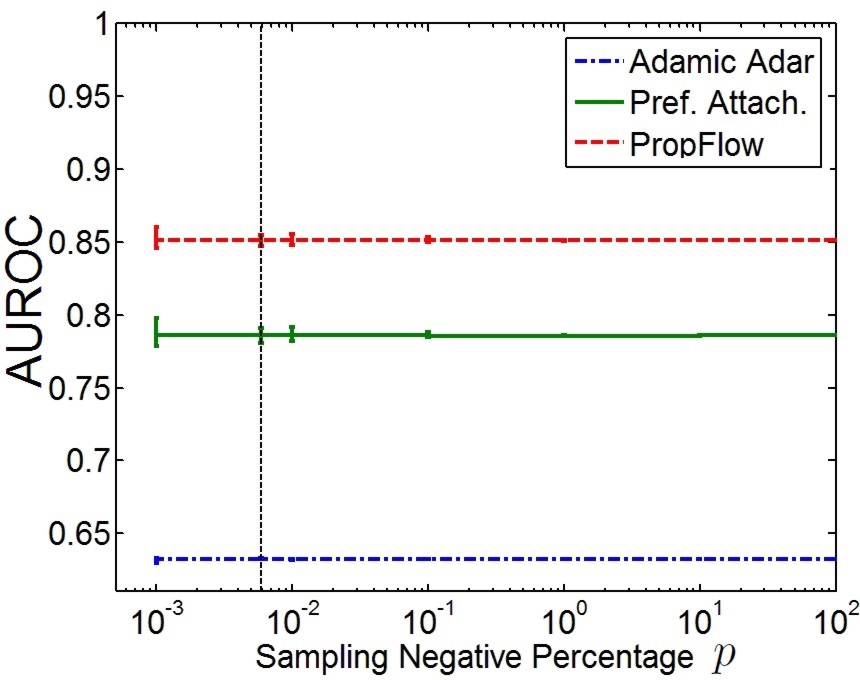}
	}
	\subfloat[Facebook]{
		\includegraphics[width=0.47\linewidth]{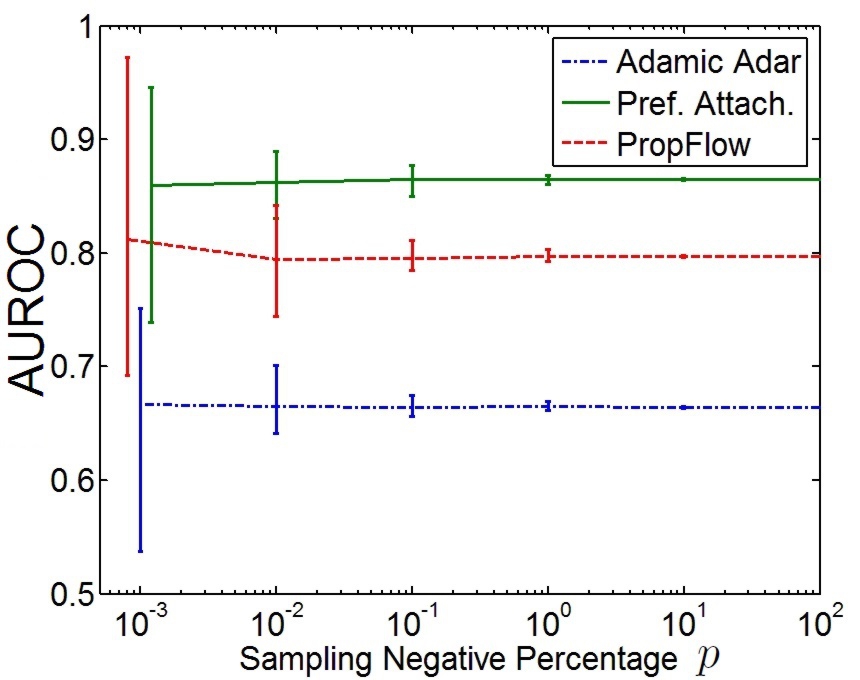}
	}\\
	\caption{The effect of sampling of negative class instances in the testing set for unsupervised classifiers. In Figure (a) the vertical line indicates class balance. We separate the error bars for different methods at the point $p=10^{-3}$ for comparison purposes. With decreasing $p$, the variance of AUROC for a link predictor increases. When the negative class percentage is between $10^{-1}$ and $10^{2}$, the mean value of AUROC for a link predictor does not change much, because the number of negative instances is still much larger than the number of positive instances. In the Condmat network the imbalance ratio is around $10^{5}$, so even when the negative class sample percentage is $10^{-1}$ the imbalance ratio in the sampled test set is around $10^{2}$. When the negative class sample percentage is low enough (i.e., $10^{-2}$), the mean value of AUROC is unstable (see results of DBLP and Enron in Figure~\ref{fig:undersample2}).}
	\label{fig:undersample1}
\end{figure}

\begin{figure}
	\centering
	\vspace{-0.4cm}
	\subfloat[DBLP]{
		\includegraphics[width=0.47\linewidth]{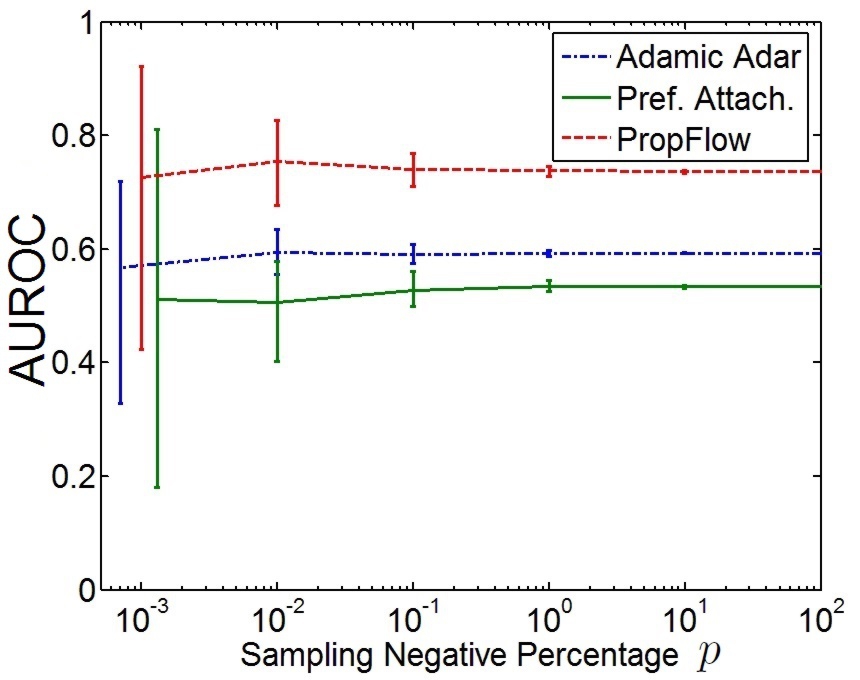}
	}
	\subfloat[Enron]{
		\includegraphics[width=0.47\linewidth]{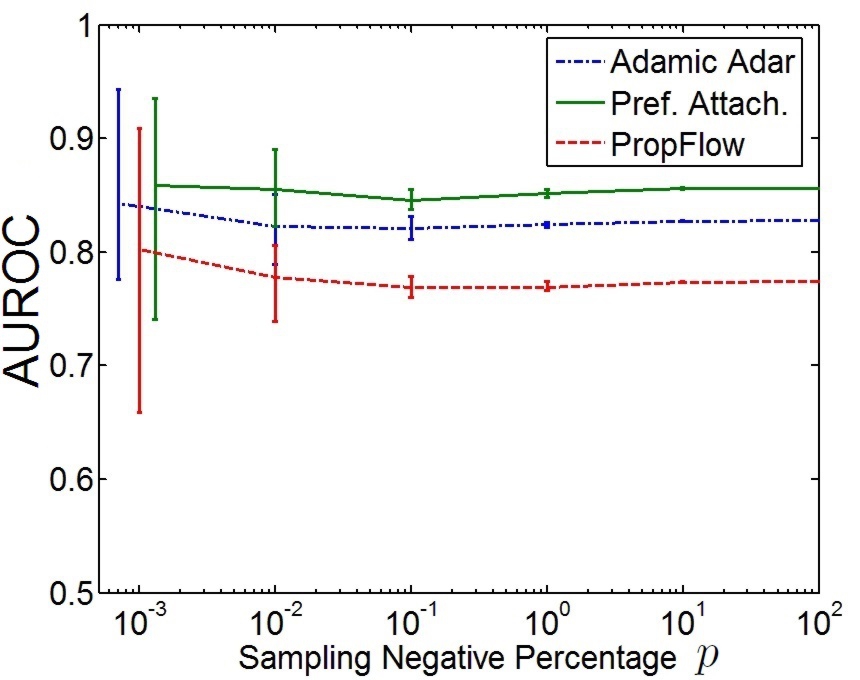}
	}
	\caption{The effect of sampling negative class instances in test sets on AUROC in supervised learning. We separate the error bars for different methods at the point $p=10^{-3}$ for comparison purposes.}
	\label{fig:undersample2}
\end{figure}

The dashed vertical line shows the AUROC for sampling that produces a balanced test set. The area deviates by more than 0.007 for PropFlow and more than 0.01 for preferential attachment, which may exceed significant variations in performance across link predictors. Further sampling causes even greater deviations. These deviations are not a weakness of the AUROC itself but are indicative of instability in score rankings within the samples. This instability does not manifest itself uniformly, and it may be greater for some predictors than for others. In Condmat in Figure~\ref{fig:undersample1}, preferential attachment exhibits greater susceptibility to the effect, while in Facebook in Figure~\ref{fig:undersample1}, PropFlow has the greatest susceptibility. The ultimate significance of the effect depends upon many factors, such as sampling percentage of negative class, properties of the predictor, and network size. From the proof of Theorem \ref{lp_stability} we can observe that the variance is also influenced by the negative instances number $N$. This is validated empirically by variations in stability across negative class sample ratios in different data sets. Condmat is stable down to 1\% sampling of negative class instances while DBLP, Enron and Facebook are stable only down to 10\% sampling of negative class instances. In sparse networks $G=(V, E)$ we can prove that the variance changes according to the order of magnitude of $|V|^{2}$.
\begin{defi}
Let a network $G = (V, E)$ be described as sparse if it maintains the property $|E| = k|V|$ for some constant $k \ll |V|$.
\end{defi}

\begin{corollary}
Given constant sampling ratio $p$, and that at most $|V|$ nodes may join the sparse network, and that the prediction ability of $\mathcal{P}$ does not change in different sparse networks, $G_{\alpha}$ and $G_{\beta}$:
\begin{align*}
\frac{\text{Var}_{\alpha}}{\text{Var}_{\beta}} \in \Theta \left( \frac{|V_{\beta}|^{2}}{|V_{\alpha}|^{2}} \right)
\end{align*}
where $\text{Var}_{\alpha}$ is the performance variance of the link predictor $\mathcal{P}$ in the sparse network $G_{\alpha}$, $\text{Var}_{\beta}$ is the performance variance of the link predictor $\mathcal{P}$ in the sparse network $G_{\beta}$, and $|V_{\alpha}|$ and $|V_{\beta}|$ are node counts in the network $G_{\alpha}$ and $G_{\beta}$.\\
\end{corollary}

\begin{proof}
As we have proved in \textbf{Theorem}~\ref{lp_stability} the variance of the performance of $\mathcal{P}$ can be written as:
\begin{align*}
\text{Var}\left(\frac{X}{Np}\right) = \frac{C(N-C)(1-p)}{N^{2}(N - 1)p}
\end{align*}
Thus the variances $\text{Var}_{\alpha}$ and $\text{Var}_{\beta}$ can be written as:
\begin{align*}
\text{Var}_{\alpha} = \frac{C_{\alpha}(N_{\alpha}-C_{\alpha})(1-p_{\alpha})}{N_{\alpha}^{2}(N_{\alpha} - 1)p_{\alpha}}\\
\text{Var}_{\beta} = \frac{C_{\beta}(N_{\beta}-C_{\beta})(1-p_{\beta})}{N_{\beta}^{2}(N_{\beta} - 1)p_{\beta}}
\end{align*}
Due to the assumption that the prediction ability of $\mathcal{P}$ does not change in $G_{\alpha}$ and $G_{\beta}$, we can write $\frac{\text{Var}_{\alpha}}{\text{Var}_{\beta}}$ as:
\begin{align*}
&\frac{\text{Var}_{\alpha}}{\text{Var}_{\beta}} 
= \frac{ \frac{C_{\alpha}(N_{\alpha}-C_{\alpha})(1-p_{\alpha})}{N_{\alpha}^{2}(N_{\alpha} - 1)p_{\alpha}} }{ \frac{C_{\beta}(N_{\beta}-C_{\beta})(1-p_{\beta})}{N_{\beta}^{2}(N_{\beta} - 1)p_{\beta}}  }
= \frac{ \frac{(1-p_{\alpha}) }{(N_{\alpha} - 1)p_{\alpha}} }{ \frac{(1-p_{\beta})}{(N_{\beta} - 1)p_{\beta}}  }
\end{align*}
Additionally we know that $p_{\alpha} = p_{\beta}$, so we can rewrite $\frac{\text{Var}_{\alpha}}{\text{Var}_{\beta}}$ as:
\begin{align*}
&\frac{\text{Var}_{\alpha}}{\text{Var}_{\beta}}
= \frac{ \frac{(1-p_{\alpha}) }{(N_{\alpha} - 1)p_{\alpha}} }{ \frac{(1-p_{\beta})}{(N_{\beta} - 1)p_{\beta}}  }
= \frac{ N_{\beta} - 1 }{ N_{\alpha} - 1 }
\end{align*}
Now we prove that $N_{\alpha} \in \Theta(|V_{\alpha}|^{2})$, $N_{\beta} \in \Theta(|V_{\beta}|^{2})$:
The number of all possible links in network $G$ is $\frac{|V|^{2} - |V|}{2}$, so for a sparse network the missing links, $|E|^{C}$, is $\frac{|V|^{2} - |V|}{2} - k|V| \in \Theta(|V|^{2})$. Let $V^{'}$ nodes and $E^{'}$ edges join the network in the future. Since the evolved network $G^{'}$ is still a sparse network and $V^{'} \leq |V|$, we know that $|V| + |V^{'}| \leq 2|V| \in \Theta(|V|)$ and $|E| + |E^{'}| \in \Theta(|V|)$. The negatives are given as $|(E \bigcup E^{'})^{C}| \in \Theta(|V|^{2})$.
Trivially we have
\begin{align*}
N_{\alpha} - 1 \in \Theta(|V_{\alpha}|^{2})\\
N_{\beta} - 1 \in \Theta(|V_{\beta}|^{2})
\end{align*}
And we know that $\frac{\text{Var}_{\alpha}}{\text{Var}_{\beta}} \in \Theta(\frac{|V_{\beta}|^{2}}{|V_{\alpha}|^{2}})$
\end{proof}

We prove that theoretically the variance changes approximately with the order of magnitude of $|V|^{2}$, and this is illustrated in both Figure~\ref{fig:undersample1} and Figure~\ref{fig:undersample2}. Figure~\ref{fig:undersample2} shows that our conclusions regarding the impact of negative class sampling with unsupervised predictors also hold for supervised predictors. In Figure~\ref{fig:undersample1} the Condmat data set has a larger number of nodes than the Facebook data set, and unsurprisingly we observe that for the same predictor with the same negative class sample ratio the variance in Facebook is much larger than in Condmat. For supervised learning, the size of DBLP is smaller than Enron and the variance for Enron is much smaller than for DBLP.


This theoretical demonstration and the empirical results show grave danger in relying on results of sampled test sets in the link prediction domain. One of the most common strategies is to undersample negatives in the test set so that it is balanced. In link prediction, class balance ratios, often easily exceeding thousands to one, are likely to leave resampled test sets that do not admit sufficiently stable evaluation for meaningful results.

\subsection{The Real Testing Distribution}
\label{sec:geo_sub_problem}
We must understand what performance we report when we undersample link prediction test sets. Undersampling is presumably part of an attempt at combating unmanageable test set sizes and describing the performance of the predictor on the network as a whole. This type of report is common, and issues of stability aside, it is theoretically valid. We question, however, whether the results that it produces actually convey useful information. Figure \ref{fig:neighborhoods} compares AUROC overall to the AUROC achievable in the distinct sub-problem created by dividing the task by geodesic distance.

We first consider the results for the preferential attachment predictor. The general conclusion across all data sets is that the apparent achievable performance is dramatically higher in the complete sets of potential edges than the performance in the sets restricted by distance. The extreme importance of geodesic distance in determining link formation correlates highly with any successful prediction method. The high-distance regions contain very few positives and effectively append a set of trivially recognizable negatives to the end. This increases the probability of a randomly selected positive appearing above a randomly selected negative, the AUROC. This phenomenon is described as the locality of link formation in growing networks \cite{leskovec:2008,wittie:2010,liu:2012,papado:2012}. In Figure~\ref{fig:locality}, we study the distribution of geodesic hops induced by each new links for four data sets. The number of new links decays exponentially with increasing hop distance between nodes.

\begin{figure}
	\centering
	\vspace{-0.4cm}
	\subfloat[\textbf{Condmat}]{
		\includegraphics[width=0.4\linewidth]{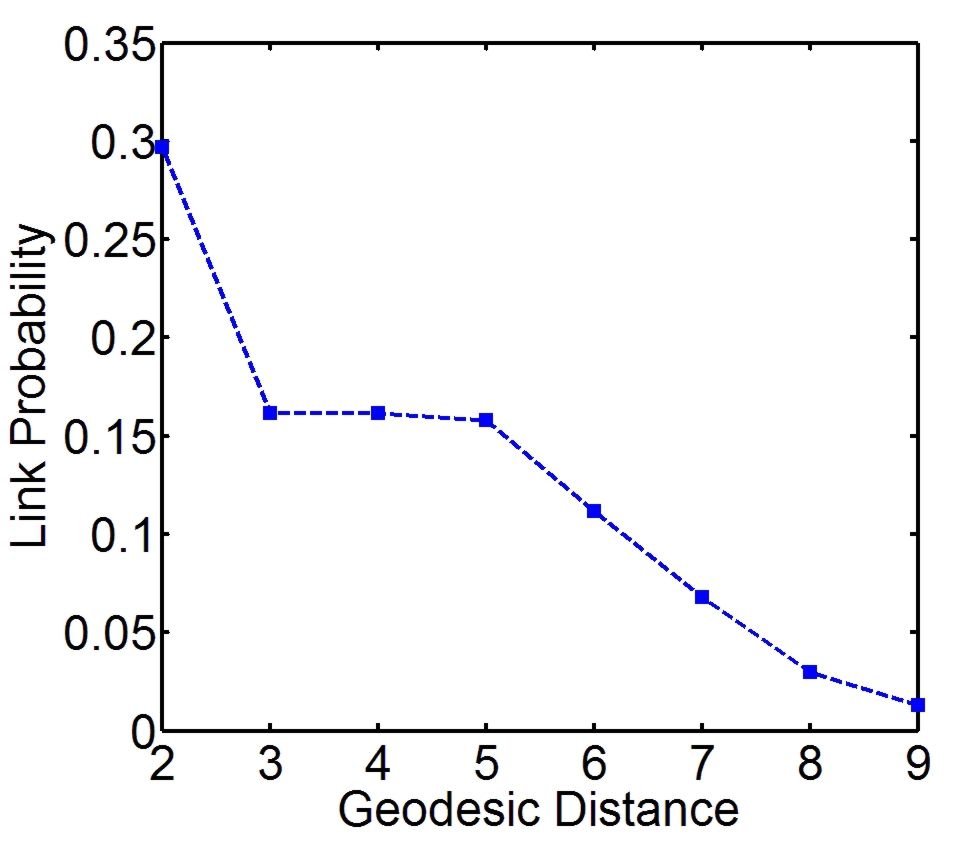}
	}
	\subfloat[\textbf{DBLP}]{
		\includegraphics[width=0.4\linewidth]{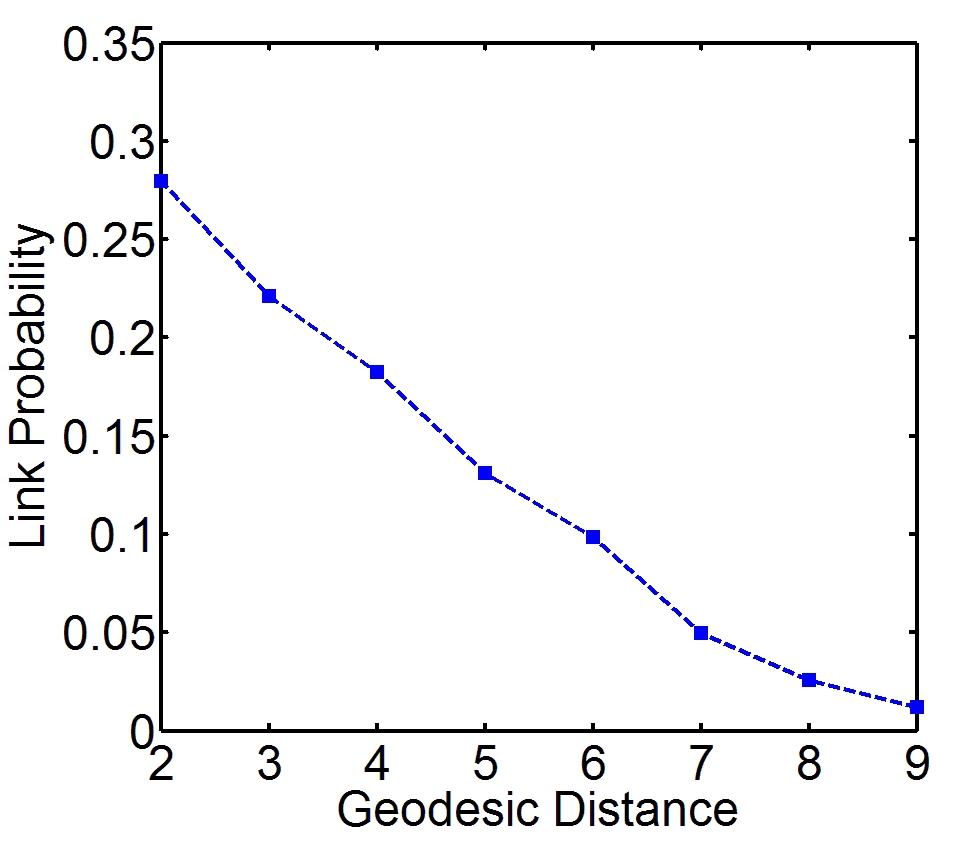}
	}\\
	\subfloat[\textbf{Enron}]{
		\includegraphics[width=0.4\linewidth]{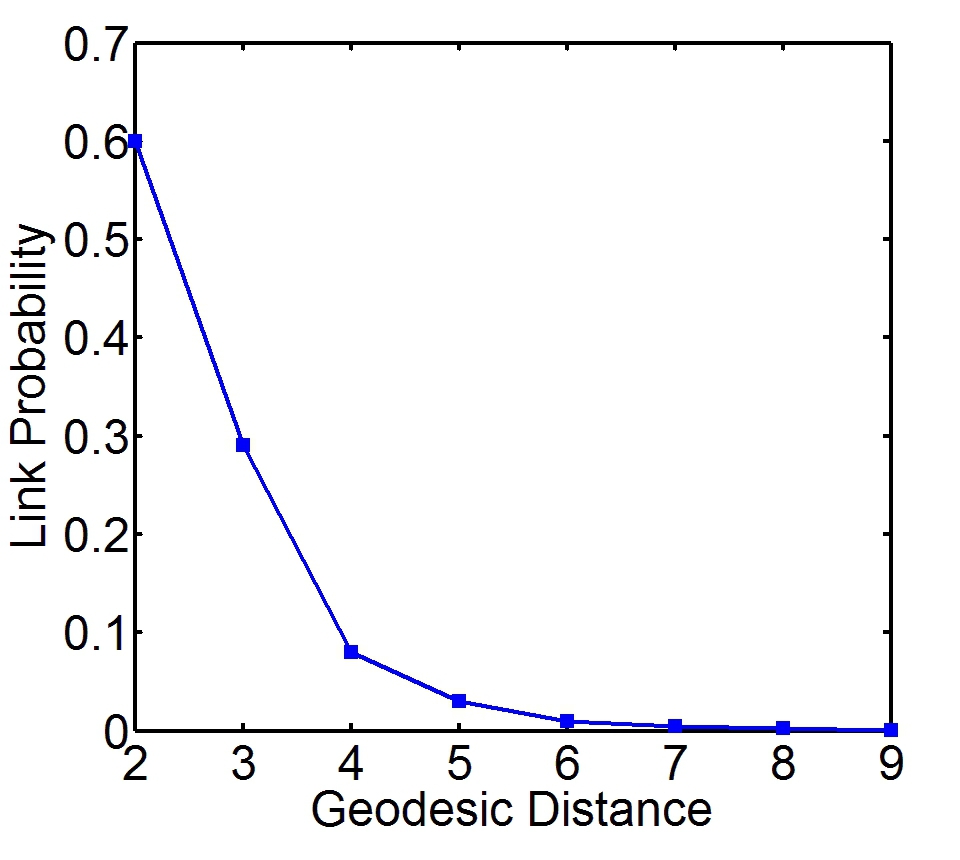}
	}
	\subfloat[\textbf{Facebook}]{
		\includegraphics[width=0.4\linewidth]{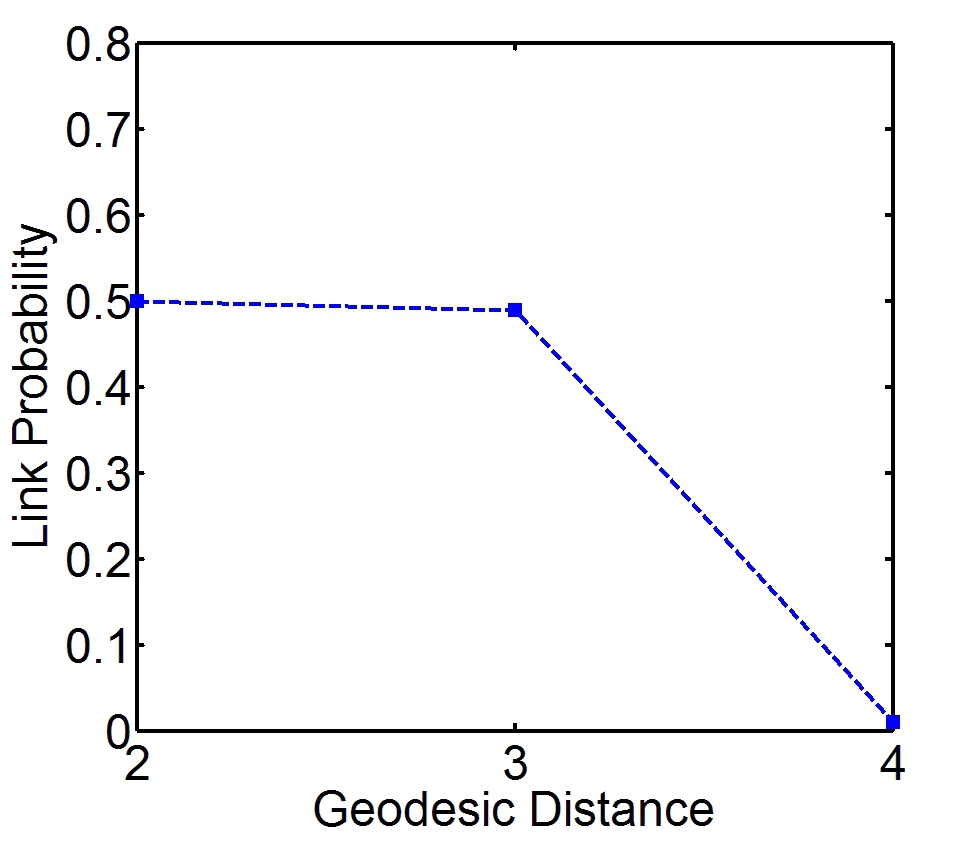}
	}
	\caption{Probability of edges $E_{\ell}$ created between nodes $\ell$ geodesic hops away.}
	\label{fig:locality}
\end{figure}

Beyond the statistics presented in Figure~\ref{fig:neighborhoods}, we compare the AUROCs of two surrogate scenarios. In the first scenario, we simulate the $\ell=2$ sub-problem, denoted as $P_{sub}$. In the second scenario, we simulate the complete link prediction problem $\ell \leq \infty$, denoted as $P_{full}$. There are $p_{s}$ positive instances and $n_{s}$ negative instances in $P_{sub}$, and there are $p_{f}$ positive instances and $n_{f}$ negative instances in $P_{full}$. We designate a parameter $\alpha$ to control the performance of the predictor $\mathcal{P}$. In the simulation the positive instances are randomly allocated among the top $\alpha$ slots in the ranking list. Additionally to simplify the simulation, we assume that the prediction method $\mathcal{P}$ has the same performance in these two problems. Because $P_{sub}$ is a sub-problem of $P_{full}$, in order to simulate $P_{full}$ more precisely, we require further details as follows:
\begin{itemize}
	\item Among $p_{f}$ positive instances, $p_{s}$ of them are randomly allocated within the top $\alpha (p_{s}+n_{s}) \beta$ slots in the ranking list, where the parameter $\beta$ is introduced to simulate the impact of non-trivially recognizable negatives on the ranking of $p_{s}$ positives in sub-problem $P_{s}$.
	\item Then $p_{f} - p_{s}$ positive instances are randomly allocated within the top $\alpha (p_{f} + n_{f})$ slots in the ranking list.
\end{itemize}

The parameter $\alpha$ is designed to simulate the performance of the predictor, while the parameter $\beta$ is designed to simulate the impact of appending negatives. Table~\ref{tab:surrogate} shows the results of this comparison. In order to comprehensively measure the statistical significance of differences between $P_{f}$ and $P_{s}$, we compare the AUROCs of $P_{f}$ to those of $P_{s}$ by 100,000 simulations under different settings of $\alpha$ and $\beta$. The numbers of $p_{s}$, $n_{s}$, $p_{f}$ and $n_{f}$ are taken from the DBLP data set.

In Table~\ref{tab:surrogate} the predictability $\alpha$ values correspond to a high AUROC, 0.9, and to a worst-case AUROC, 0.5. When the impact of appending negatives is small (i.e. $\beta=10$), the AUROC of $P_{f}$ is most dramatically greater than the AUROC of $P_{s}$, with p-value less than 0.0001. Even if the impact of appending negatives is large (i.e. $\beta=50$), the AUROC of $P_{f}$ is much larger than that of $P_{s}$, with p-value at most 0.048. The above observation is not significantly influenced by the performance of the predictor $\mathcal{P}$.
\begin{table}

\caption{{Surrogate Comparisons}. The $\alpha$ parameter controls the performance of link predictors, and lower values of $\alpha$ indicate higher performance of link predictors. The $\beta$ parameter simulates the impact of non-trivially recognizable negatives on the ranking of $p_{s}$ positives in sub-problem $P_{s}$. In the simulation we demonstrate that based on the AUROC the performance on the full link prediction problem ($\ell \leq \infty$) is significantly better ($H_{0}$ rejected at 2 sigmas) than the sub-problem ($\ell \leq 2$). This observation made in surrogate scenarios is validated by the real-world experiments in Figure~\ref{fig:neighborhoods}.}
\centering
\resizebox{4.2in}{!}{
\begin{tabular}{|l|*{5}{c|}}\cline{1-6}
\backslashbox{$\alpha$}{$\beta$}
&\makebox[2em]{10}&\makebox[2em]{20}&\makebox[2em]{30}
&\makebox[2em]{40}&\makebox[2em]{50}\\ \cline{1-6} \cline{1-6}
0.2 & $8$ $\text{sigmas}$ & $6$ $\text{sigmas}$ & $5$ $\text{sigmas}$ & $3$ $\text{sigmas}$ & $2$ $\text{sigmas}$ \\\cline{1-6}
0.3 & $8$ $\text{sigmas}$ & $6$ $\text{sigmas}$ & $5$ $\text{sigmas}$ & $3$ $\text{sigmas}$ & $2$ $\text{sigmas}$ \\\cline{1-6}
0.4 & $9$ $\text{sigmas}$ & $7$ $\text{sigmas}$ & $5$ $\text{sigmas}$ & $4$ $\text{sigmas}$ & $2$ $\text{sigmas}$ \\\cline{1-6}
0.5 & $9$ $\text{sigmas}$ & $7$ $\text{sigmas}$ & $5$ $\text{sigmas}$ & $4$ $\text{sigmas}$ & $2$ $\text{sigmas}$ \\\cline{1-6}
0.6 & $9$ $\text{sigmas}$ & $7$ $\text{sigmas}$ & $6$ $\text{sigmas}$ & $4$ $\text{sigmas}$ & $2$ $\text{sigmas}$ \\\cline{1-6}
0.7 & $9$ $\text{sigmas}$ & $7$ $\text{sigmas}$ & $6$ $\text{sigmas}$ & $4$ $\text{sigmas}$ & $2$ $\text{sigmas}$ \\\cline{1-6}
0.8 & $9$ $\text{sigmas}$ & $7$ $\text{sigmas}$ & $6$ $\text{sigmas}$ & $4$ $\text{sigmas}$ & $2$ $\text{sigmas}$ \\\cline{1-6}
0.9 & $9$ $\text{sigmas}$ & $7$ $\text{sigmas}$ & $6$ $\text{sigmas}$ & $4$ $\text{sigmas}$ & $3$ $\text{sigmas}$ \\\cline{1-6}
\end{tabular}
}
\\
{
\begin{flushleft}
\end{flushleft}
}
\label{tab:surrogate}
\end{table}

In Figure~\ref{fig:neighborhoods} we can observe that different prediction methods have different behaviors for varying $\ell$ sub-problems. In the DBLP data set, preferential attachment is unstable across geodesic distances while PropFlow exhibits monotonic behavior. This is because the preferential attachment method is inherently ``non-local'' in its judgment of link formation likelihood. Additionally, as discussed earlier, the difference between AUROC in $\ell \leq \infty$ and AUROC in distinct sub-problems (i.e. $\ell = 2$) is diminished when the size of the network decreases.

\begin{figure}
	\centering
	\vspace{-0.4cm}
	\subfloat[PA-\textbf{Condmat}]{
		\includegraphics[width=0.3\linewidth]{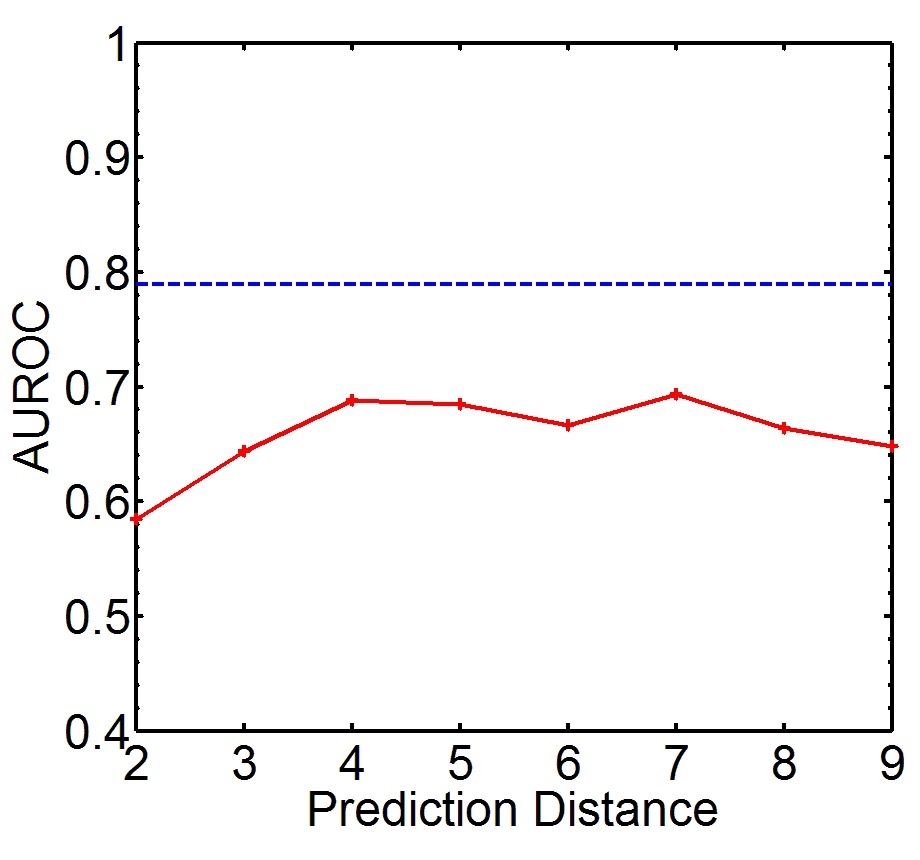}
	}
	\subfloat[PF-\textbf{Condmat}]{
		\includegraphics[width=0.3\linewidth]{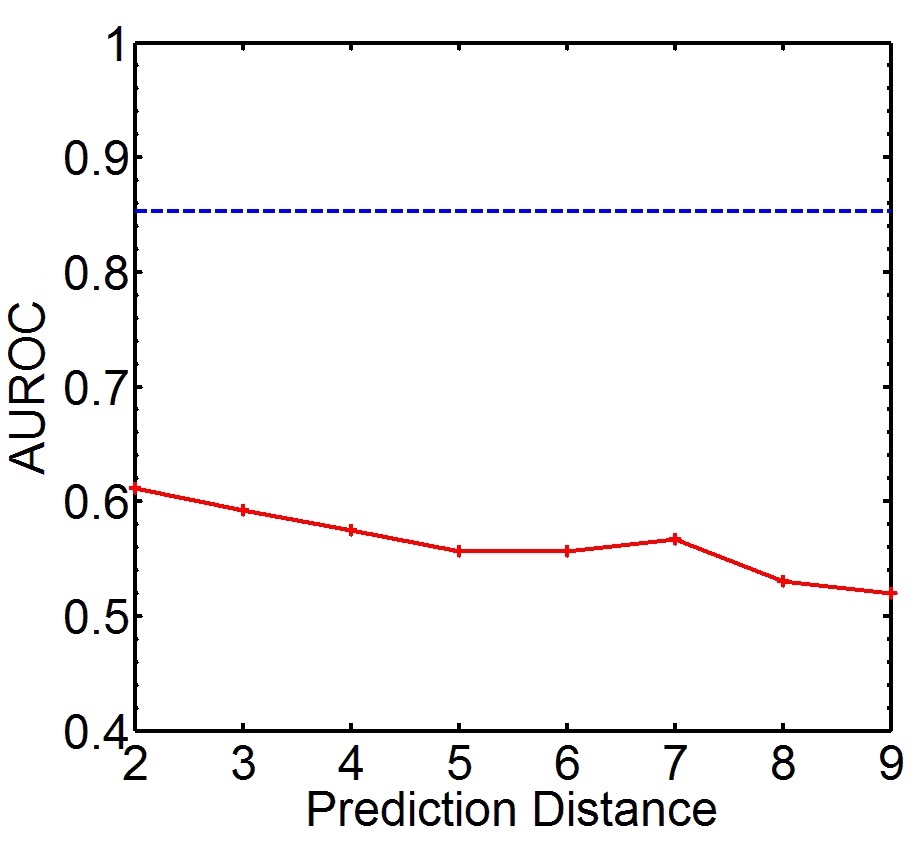}
	}\\
	\subfloat[PA-\textbf{DBLP}]{
		\includegraphics[width=0.3\linewidth]{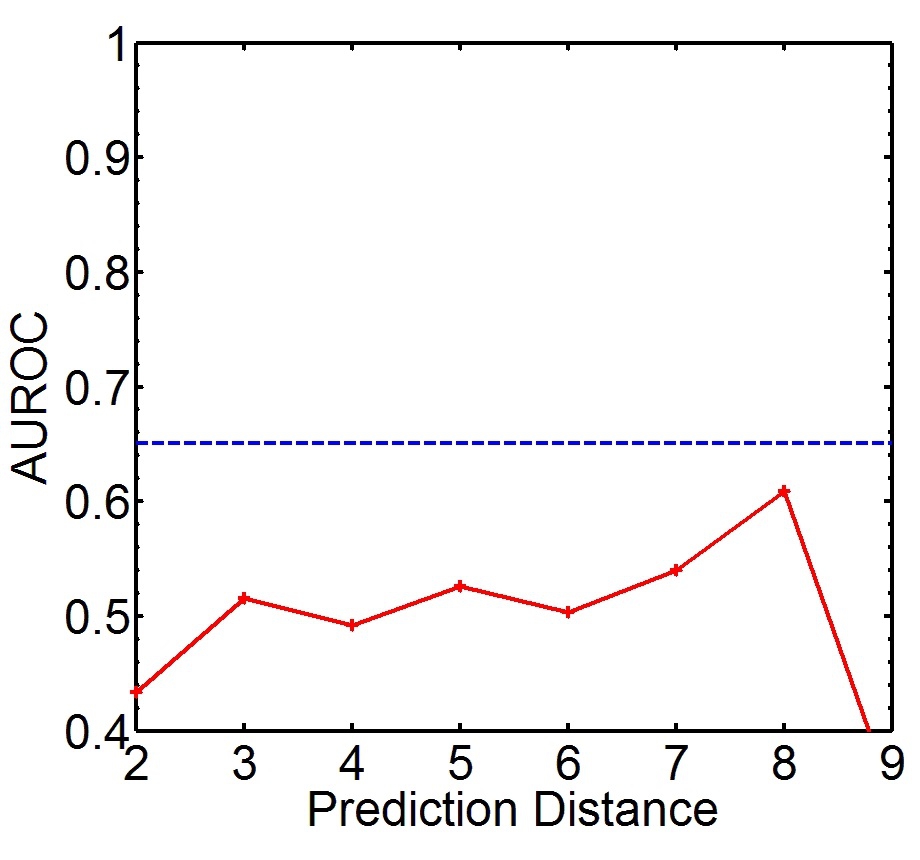}
	}
	\subfloat[PF-\textbf{DBLP}]{
		\includegraphics[width=0.3\linewidth]{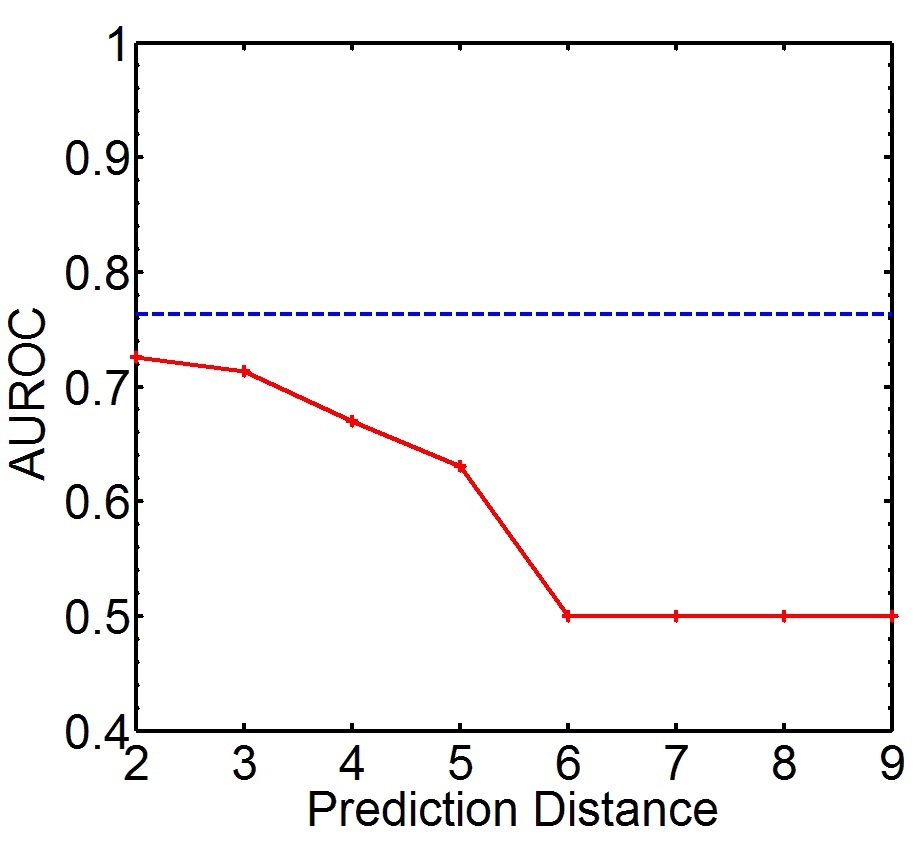}
	}\\
	\subfloat[PA-\textbf{ Enron}]{
		\includegraphics[width=0.3\linewidth]{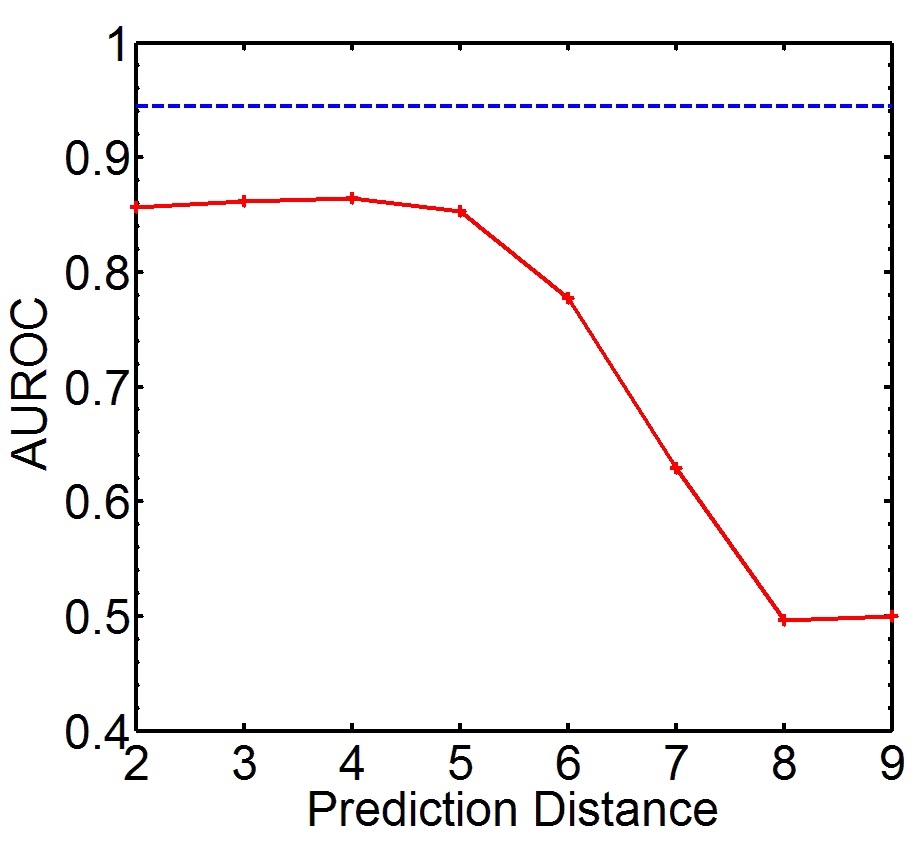}
	}
	\subfloat[PF-\textbf{Enron}]{
		\includegraphics[width=0.3\linewidth]{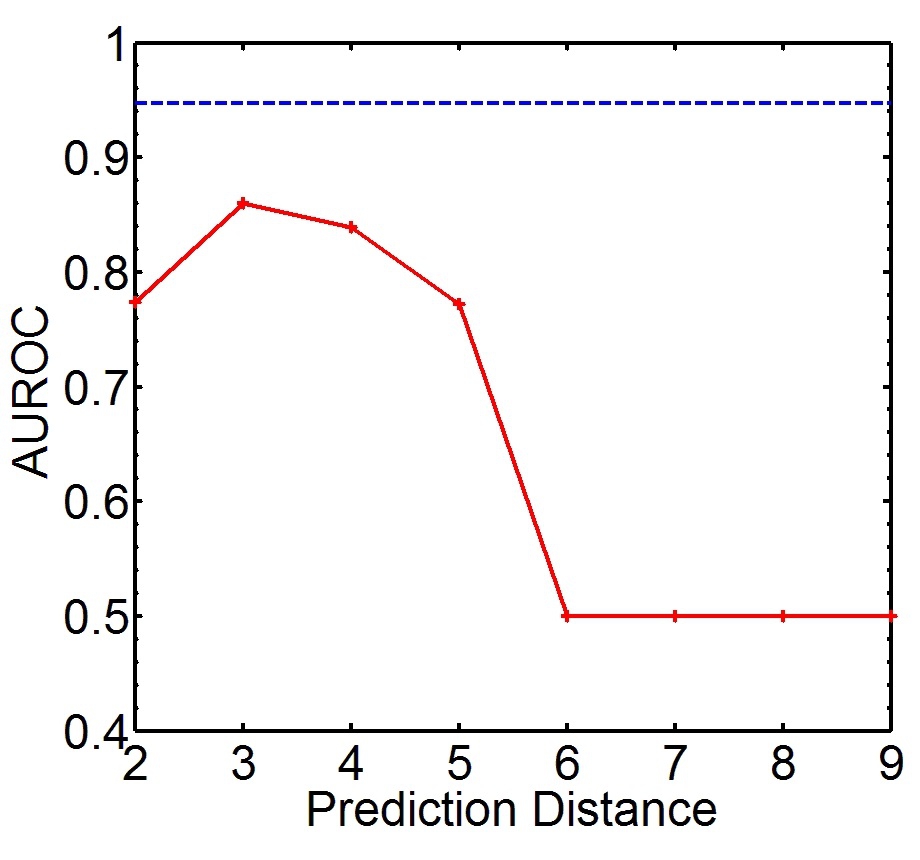}
	}\\
	\subfloat[PA-\textbf{Facebook}]{
		\includegraphics[width=0.3\linewidth]{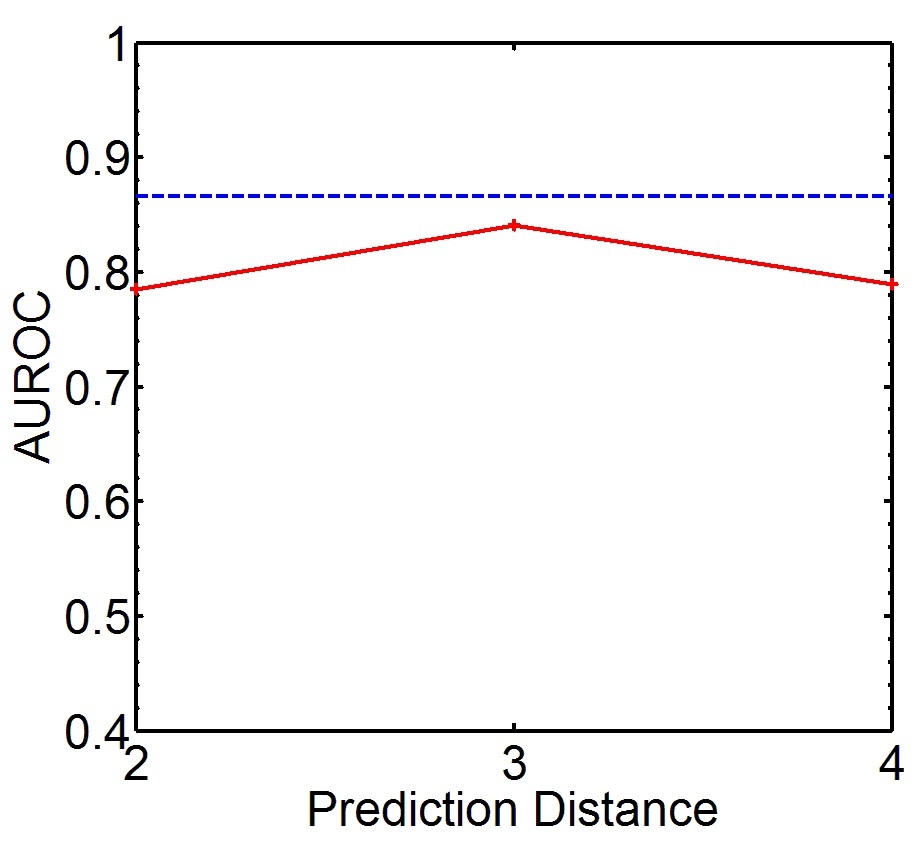}
	}
	\subfloat[PF-\textbf{Facebook}]{
		\includegraphics[width=0.3\linewidth]{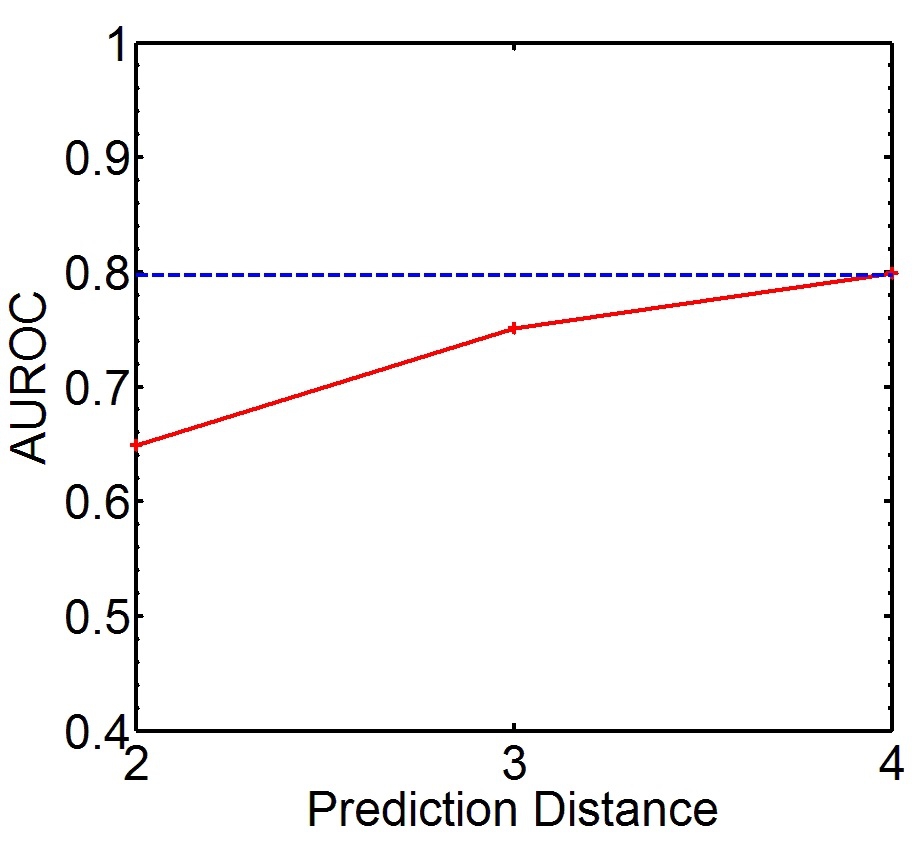}
	}
	\caption{The ROC curve performance of the Preferential Attachment and PropFlow link predictors over each neighborhood. The results of Adamic/Adar are not included because the Adamic/Adar method only has descriptive power for node pairs within two hops. The horizontal line represents the performance apparent by considering all potential links. We compare AUROC overall to the AUROC achievable in the distinct sub-problem created by dividing the task by geodesic distance $\ell$. The red line represents the AUROC for each sub-problem.}
	\label{fig:neighborhoods}
\end{figure}

The effect is exaggerated for PropFlow and for other ranking methods that inherently scale according to distance, such as rooted PageRank and Katz. In such cases, the ROC curve for the amalgamated data approximates concatenation of the individual ordered outputs, which inherently places the distances with higher imbalance ratios at the end where they inflate the overall area. Figure \ref{fig:neighborhoods} shows the effect for the PropFlow prediction method on the right.

For PropFlow, the apparent achievable performance in Condmat is 36.2\% higher for the overall score ordering than for the highest of the individual orderings! This result also has important implications from a practical perspective. In the Condmat network, PropFlow appears to have a higher AUROC than preferential attachment ($\ell \leq \infty$), but the \emph{only} distance at which it outperforms preferential attachment is the 2-hop distance. Preferential attachment is a superior choice for the other distances in cases where the other distances matter. These details are hidden from view by ROC space. They also illustrate that the performance indicated by overall ROC is not meaningful with respect to deployment expectations and that it conflates performance across neighborhoods with a bias toward rankings that inherently reflect distance.

Consider the data distribution of the link prediction problem used in this paper. In Condmat there are 148.2 million negatives and 29,898 positives. The ratio of negatives to positives is 4,955 to 1. There are 1196 positives and 214,616 negatives in $\ell=2$. To achieve a 1 to 1 ratio with random edge sampling, statistical expectation is for 43.3 2-hop negatives to remain. The 2-hop neighborhood contains 30\% of all positives, so clearly it presents the highest baseline precision. That border is the most important to capture well in classification, because improvements in $\ell=2$ discrimination are worth much more than improvements at higher distances. 16\% of all positives are in the $\ell=3$ sub-problem, so the same argument applies with it versus higher distances.

\begin{figure}
	\centering
	\subfloat[PA-\textbf{Condmat}]{
		\includegraphics[width=0.3\linewidth]{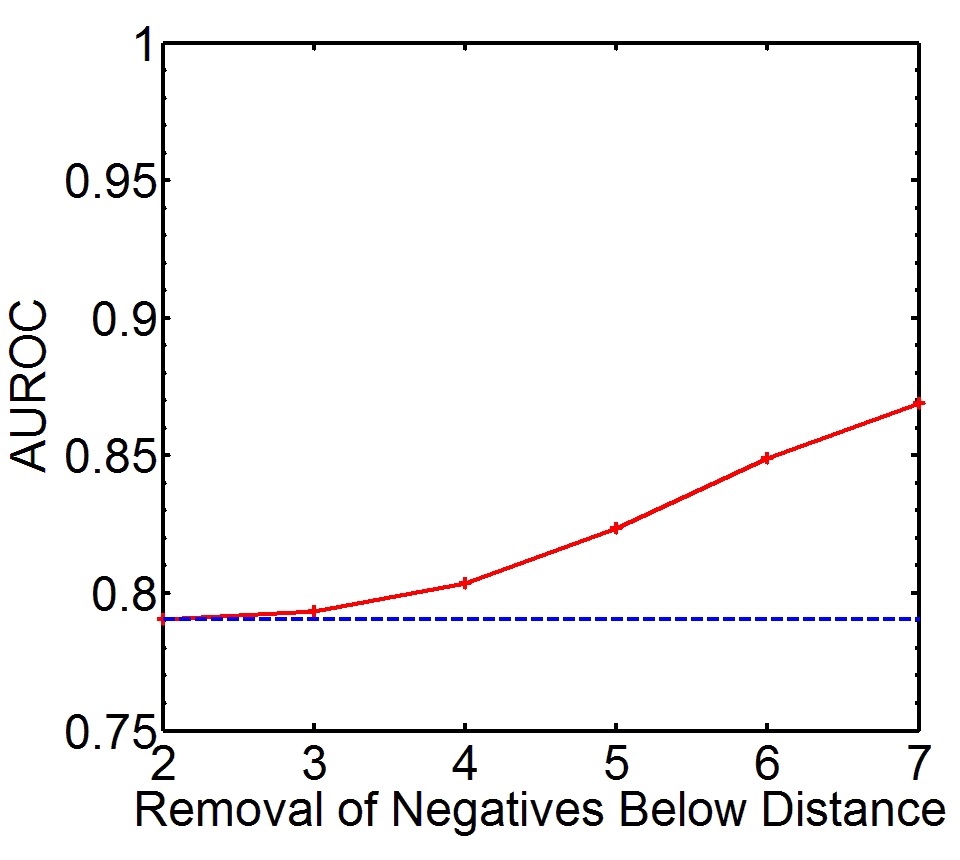}
	}
	\subfloat[PF-\textbf{Condmat}]{
		\includegraphics[width=0.3\linewidth]{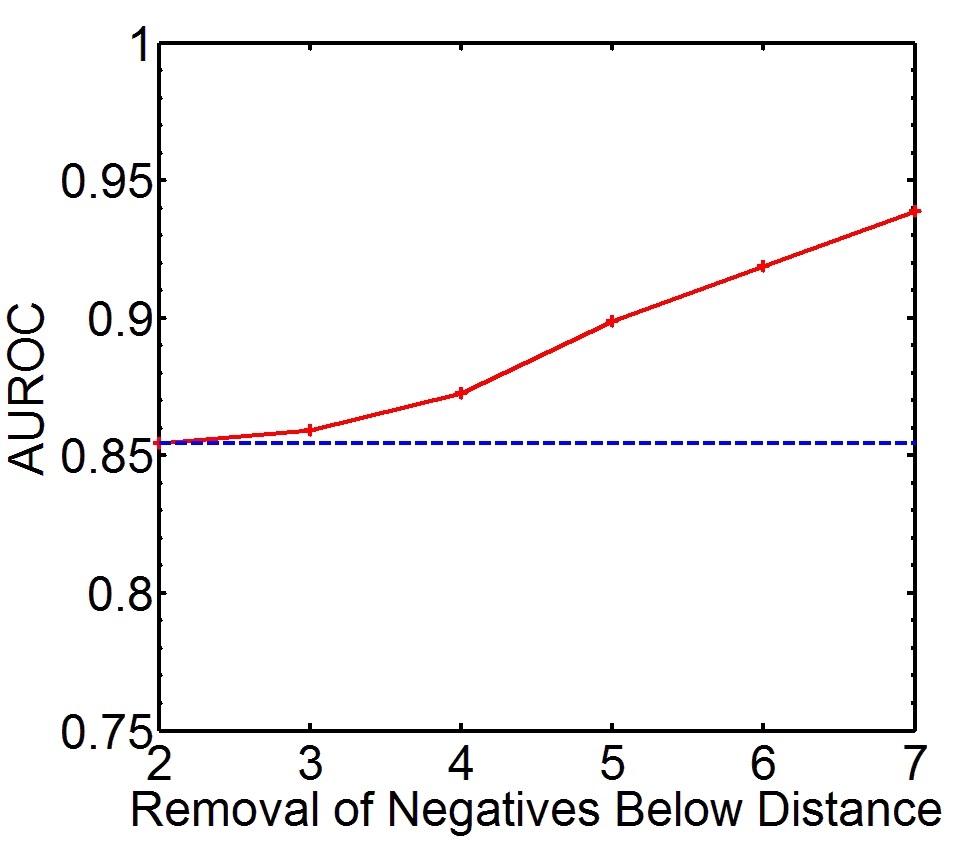}
	}\\
	\subfloat[PA-\textbf{DBLP}]{
		\includegraphics[width=0.3\linewidth]{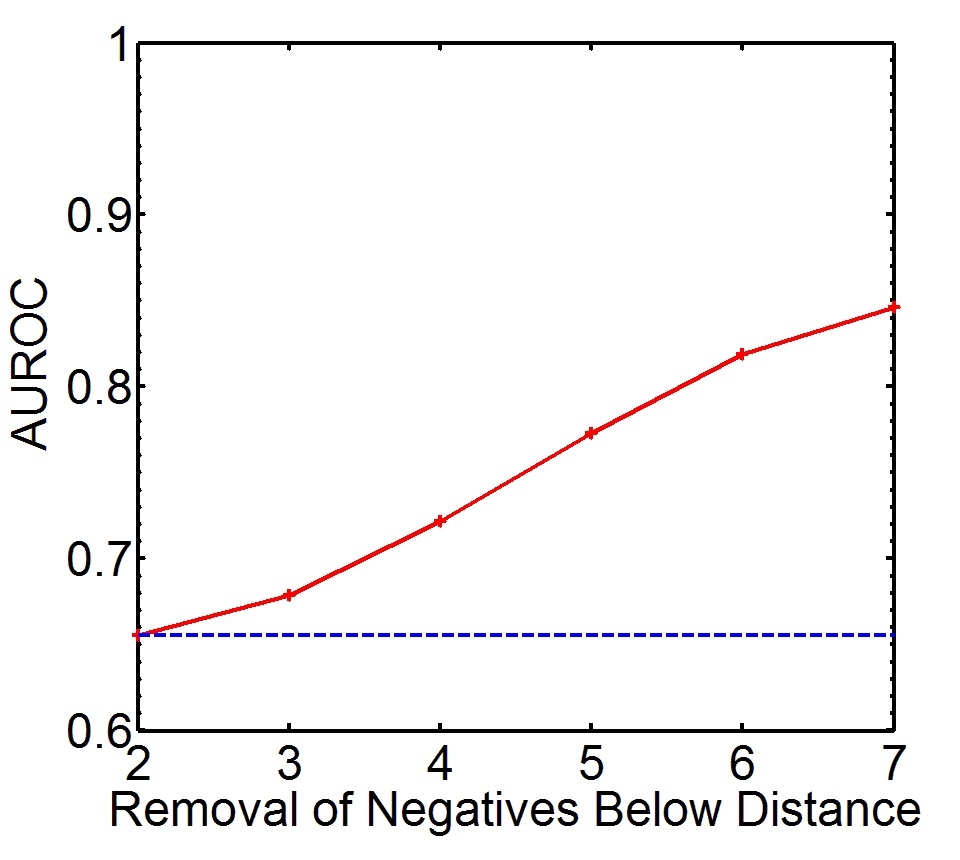}
	}
	\subfloat[PF-\textbf{DBLP}]{
		\includegraphics[width=0.3\linewidth]{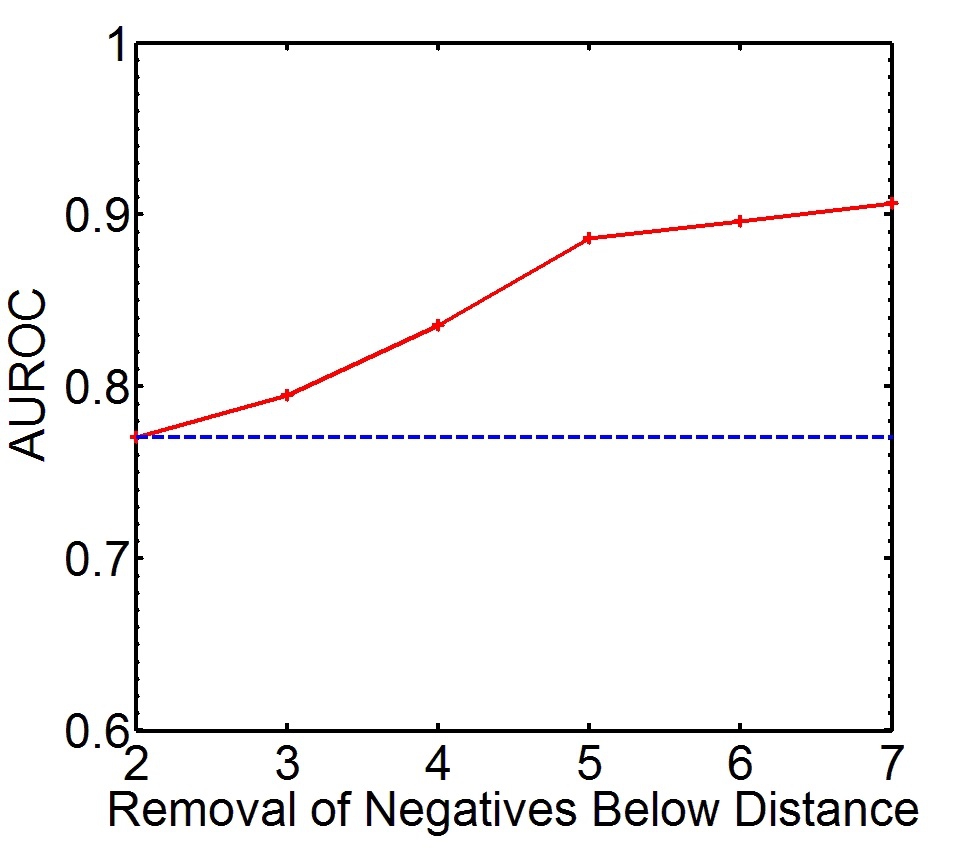}
	}\\
	\subfloat[PA-\textbf{Enron}]{
		\includegraphics[width=0.3\linewidth]{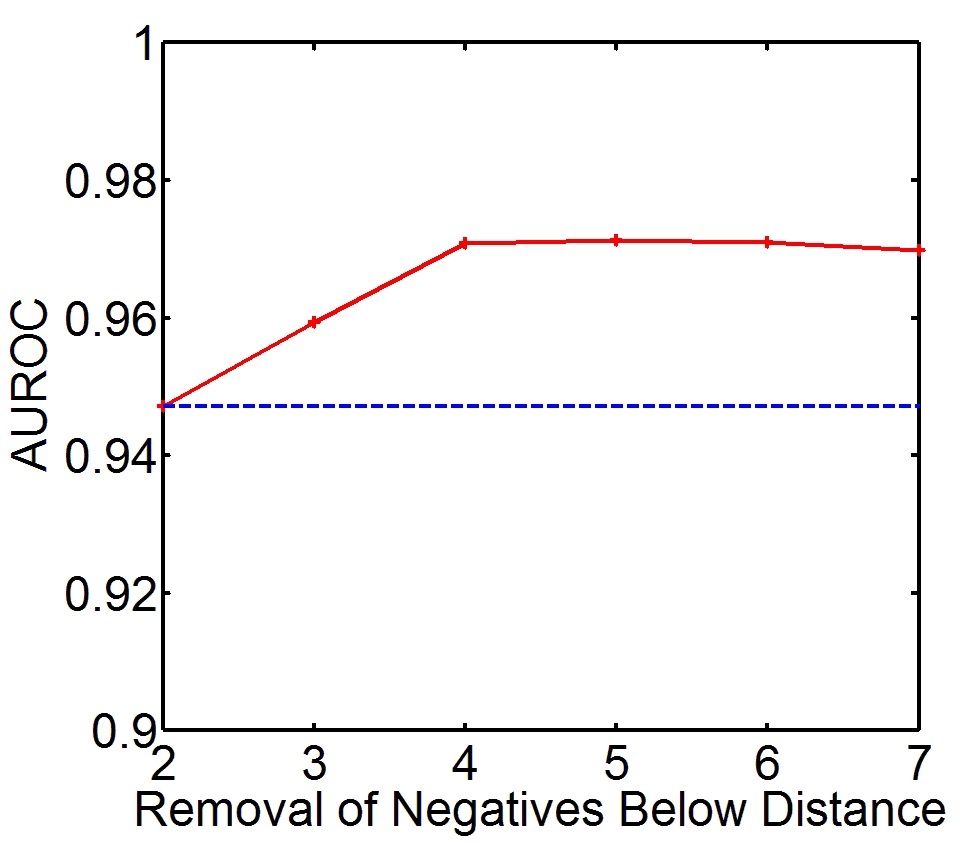}
	}
	\subfloat[PF-\textbf{Enron}]{
		\includegraphics[width=0.3\linewidth]{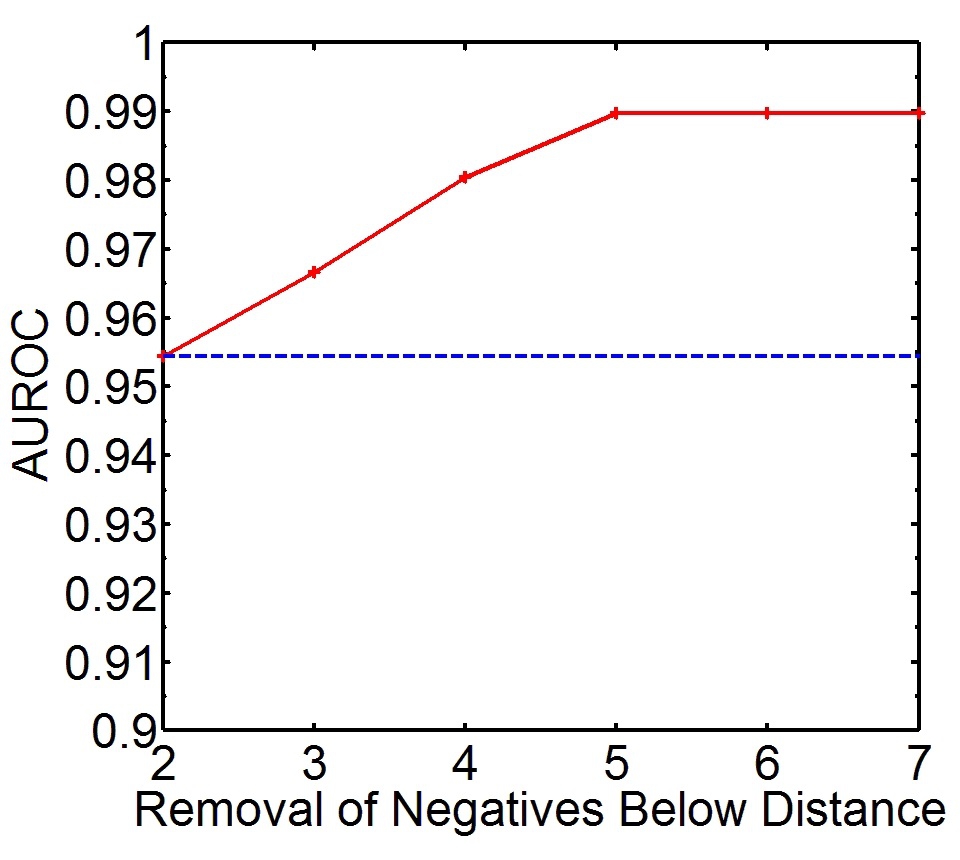}
	}\\
	\subfloat[PA-\textbf{Facebook}]{
		\includegraphics[width=0.3\linewidth]{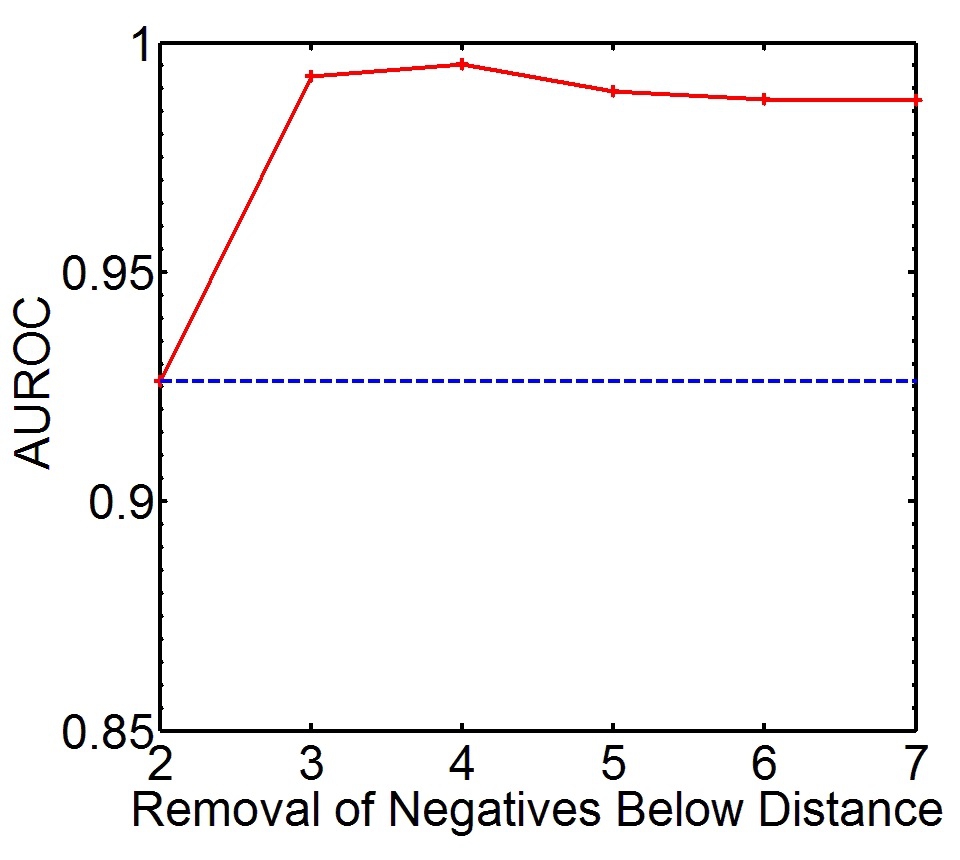}
	}
	\subfloat[PF-\textbf{Facebook}]{
		\includegraphics[width=0.3\linewidth]{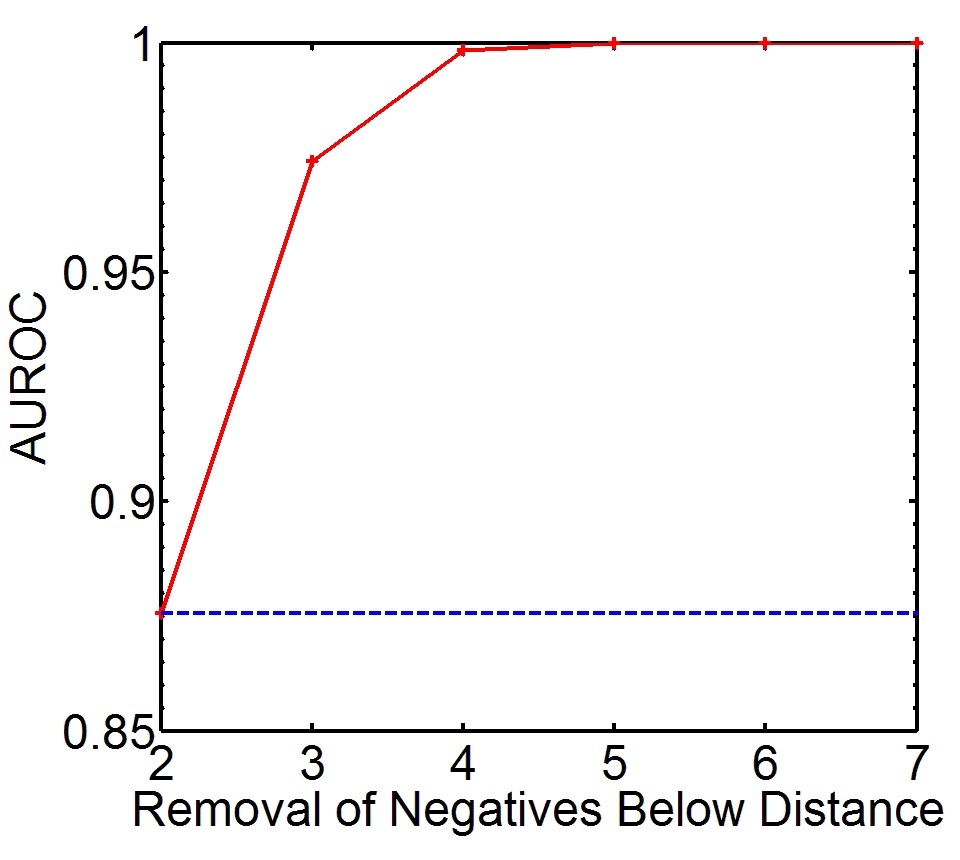}
	}	
	\caption{The effect of removing negative instances from increasingly distant potential links. The horizontal line represents the base AUROC on the unsampled test set. This figure illustrates the minimal effect of selectively filtering all negatives from low-distance neighborhoods. As we observe with the increase of $\ell$, the AUROC increases when all negative instances below $\ell$ are removed.}
	\label{fig:moreproof}
\end{figure}

The real data distribution on which we report performance when we perform this sampling is the relatively easy boundary described by the highly disparate $\ell=2$ positives and high-distance negatives. Figure \ref{fig:moreproof} further substantiates this point by illustrating the minimal effect of selectively filtering all negatives from low-distance neighborhoods. We know that performance in the $\ell=2$ is most significant because of the favorable class ratio and that improvements in defining this critical boundary offer the greatest rewards. Simultaneously, as the figure shows, in Condmat we can entirely remove all $\ell=2$ test negatives with only a 0.2\% effect on the AUROC for two predictors from entirely different families. We must remove all $\ell \leq 4$ test negatives before the alteration becomes conspicuous, yet these are the significant boundary instances.

Since we also know that data distributions do not affect ROC curves, we can extend this observation even when no sampling is involved: considering the entire set of potential links in ROC space evaluates prediction performance of low-distance positives versus high-distance negatives. We must instead describe performance within the distribution of $\ell=2$ positives and negatives and select predictors that optimize this boundary.

\subsection{Case Study on Kaggle Sampling}
\label{sec_sub_kaggle}
Of the described sampling methods, only uniform random selection from the complete set of potential edges preserves the testing distribution. Though questionable for meaningful evaluation of deployment potential, it is at least an attempt at unbiased evaluation. One recently employed alternative takes another approach to sampling, aggressively undersampling negatives from over-represented distances and preserving a much higher proportion of low-distance instances. The Kaggle link prediction competition \cite{narayanan:2011} undersampled the testing set by manipulating the amount of sampling from each neighborhood to maintain approximate balance within the neighborhoods. The distribution of distances exhibited by the 8960 test edges is shown in Figure \ref{fig:kaggle}.

Consider the results of Figure \ref{fig:kaggle} against the results of fair random sampling in the Condmat network. Unless Kaggle has an incredibly small effective diameter, it is impossible to obtain this type of distribution. It requires a sampling approach that includes low-distance edges from the testing network with artificially high probability. While this selective sampling approach might seem to better highlight some notion of average boundary performance across neighborhoods, it is instead meaningless because it creates a testing distribution that would never exist in deployment. The Kaggle competition disclosed that the test set was balanced. In a deployment scenario, it is impossible to provide to a prediction method a balance of positives and negatives from each distance, because that would require knowledge of the target of the prediction task itself.

More significantly, the Kaggle approach is unfair and incomparable because the original distribution is not preserved, and there is no reason to argue for one arbitrary manipulation of distance prevalence over another. Simultaneously, the AUROC will vary greatly according to each distributional shift. It is even possible to predictably manipulate the test set to achieve an arbitrary AUROC through such a sampling approach. Any results obtained via such a testing paradigm are inextricably tied to the sampling decisions themselves, and the sampling requires the very results we are predicting in a deployment scenario. As a result, Kaggle AUROCs may not be indicative of real achievable performance or even of a proper ranking of models in a realistic prediction task.

\begin{figure}
	\centering
	\subfloat[Condmat]{
		\includegraphics[width=0.43\linewidth]{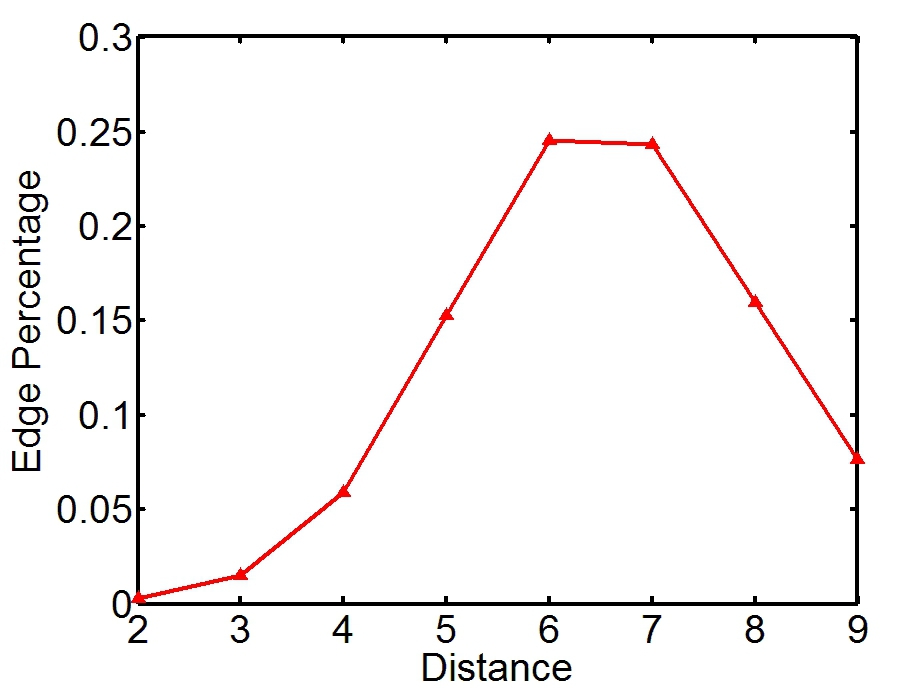}
	}
	\subfloat[Kaggle]{
		\includegraphics[width=0.43\linewidth]{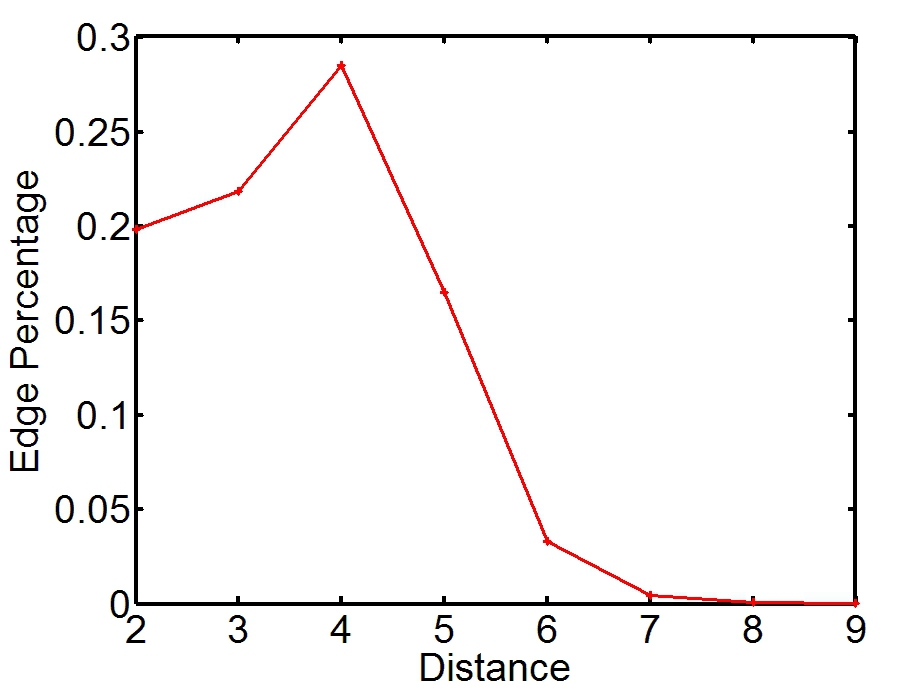}
	}
	\caption{Distribution of distances in test sets.}
	\label{fig:kaggle}
\end{figure}

To empirically explore the differences between a fair random sampling strategy and the Kaggle sampling strategy, we provide the distributions of distances in different data sets using both sampling strategies. Additionally we also compare the AUROC achievable in two sampling strategies. In Figure~\ref{fig:kaggle1} we compare the differences of distance distributions when using fair random sampling and Kaggle sampling, while in Table~\ref{performance_metric} we provide ROC performances (AUROC) of prediction methods under these two sampling strategies.

In Table~\ref{performance_metric} we can see that the apparent achievable performance using the Kaggle sampling strategy is remarkably higher, up to 25\%, than the performance achievable by fair random sampling. The performance discrepancy between fair random sampling and Kaggle sampling depends upon several factors, such as the prediction method, network size, and the geodesic distribution of positive instances. Here we will explore the observations made in the Table~\ref{performance_metric}.
\begin{itemize}
	\item \textbf{Geodesic Distribution}.
	Increasing $\ell$ increases the difficulty of the prediction sub-problem due to increasing imbalance. The Kaggle sampling maintains approximate balance within each neighborhood, greatly reducing the difficulty of the link prediction task. This is the cause for apparent performance improvements when using the Kaggle sampling strategy. 
	\item \textbf{Preferential Attachment}.\\
	Preferential attachment ignores the impact of geodesic distance between nodes, so it fails to penalize distant potential links appropriately. Kaggle sampling removes many high-distance negative instances, and most positive instances within high-distance neighborhoods have high preferential attachment scores. As a result, preferential attachment benefits particularly from Kaggle sampling.
	\item \textbf{PropFlow}.\\
	PropFlow considers the influence of geodesic distance on the formation of links. In Condmat or DBLP, when there are more positive instances spanning high distances, the Kaggle strategy unfairly penalizes path-based measures such as PropFlow. Contrarily when positive instances reside in low-distance neighborhoods, such as Enron and Facebook, PropFlow fares better.
	\item \textbf{Adamic Adar}.\\
	The Adamic/Adar method only has descriptive power within the $\ell=2$ sub-problem. In data sets where high-distance links are sampled more often, the apparent performance of Adamic/Adar is strikingly and unfairly impacted.
\end{itemize}

\begin{figure}
	\centering
	\subfloat[Condmat]{
		\includegraphics[width=0.4\linewidth]{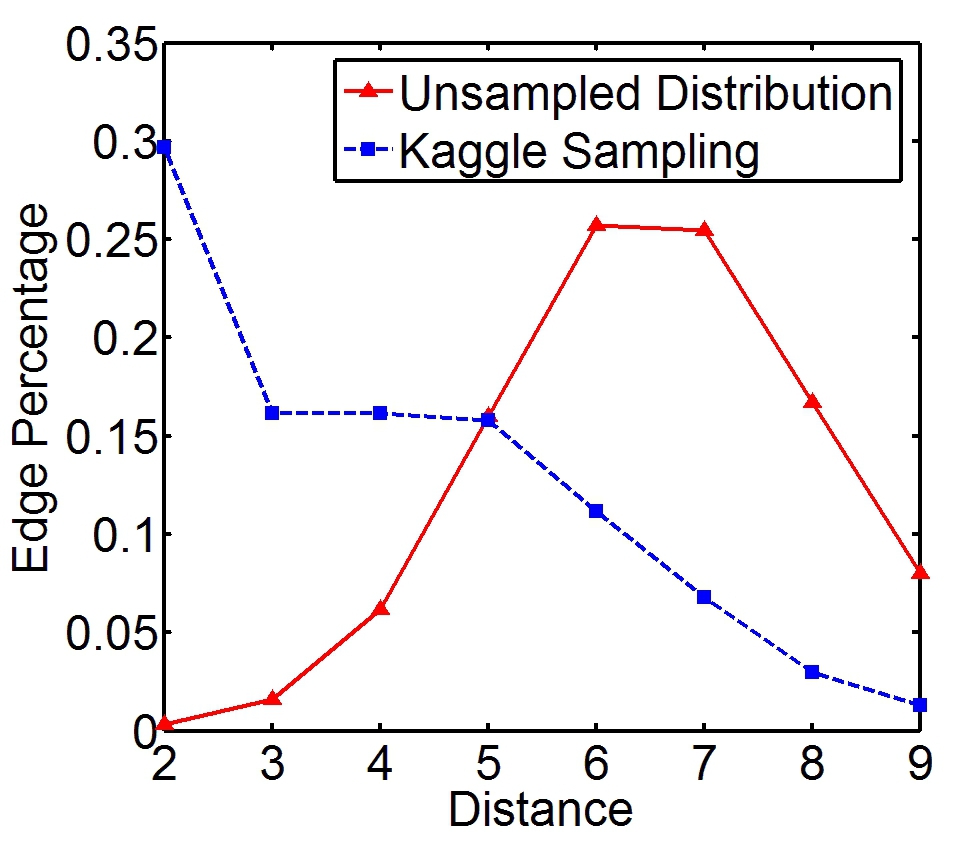}
	}
	\subfloat[DBLP]{
		\includegraphics[width=0.4\linewidth]{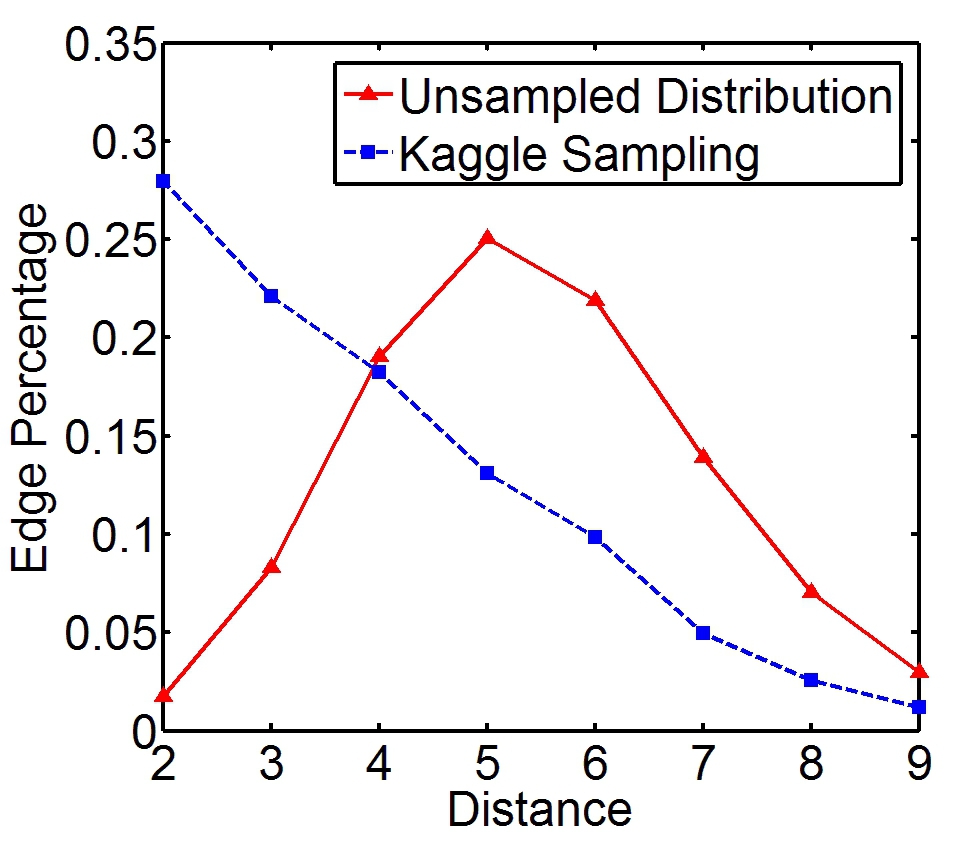}
	}\\
	\subfloat[Enron]{
		\includegraphics[width=0.4\linewidth]{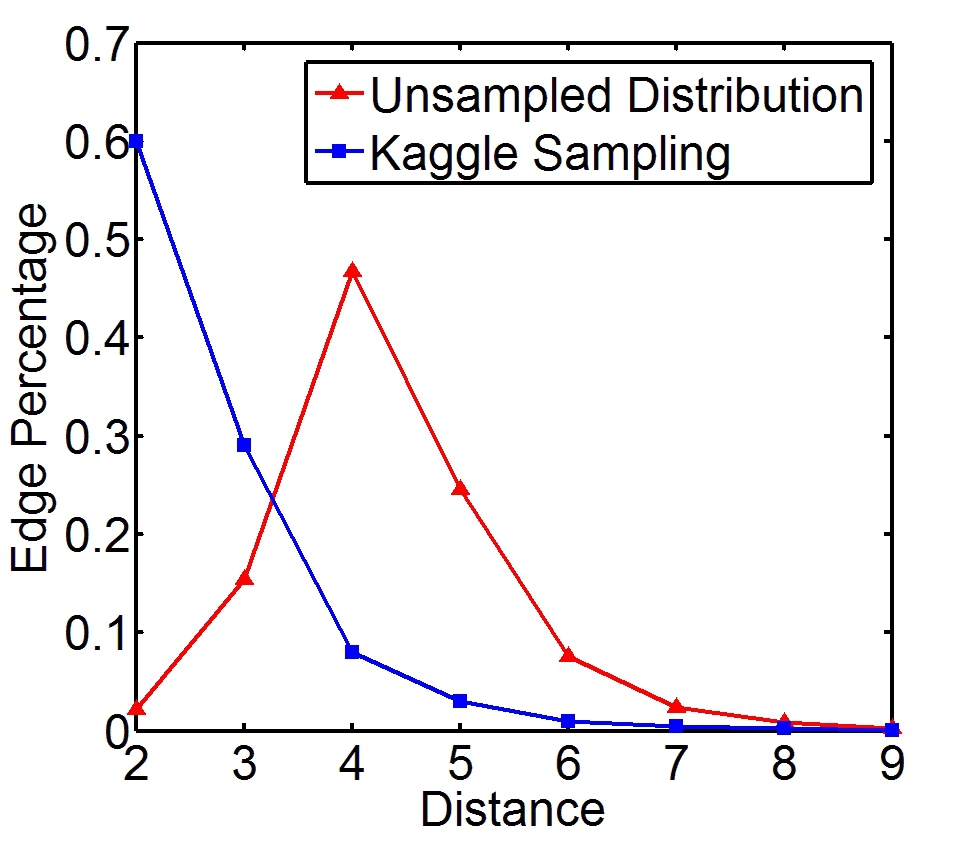}
	}
	\subfloat[Facebook]{
		\includegraphics[width=0.4\linewidth]{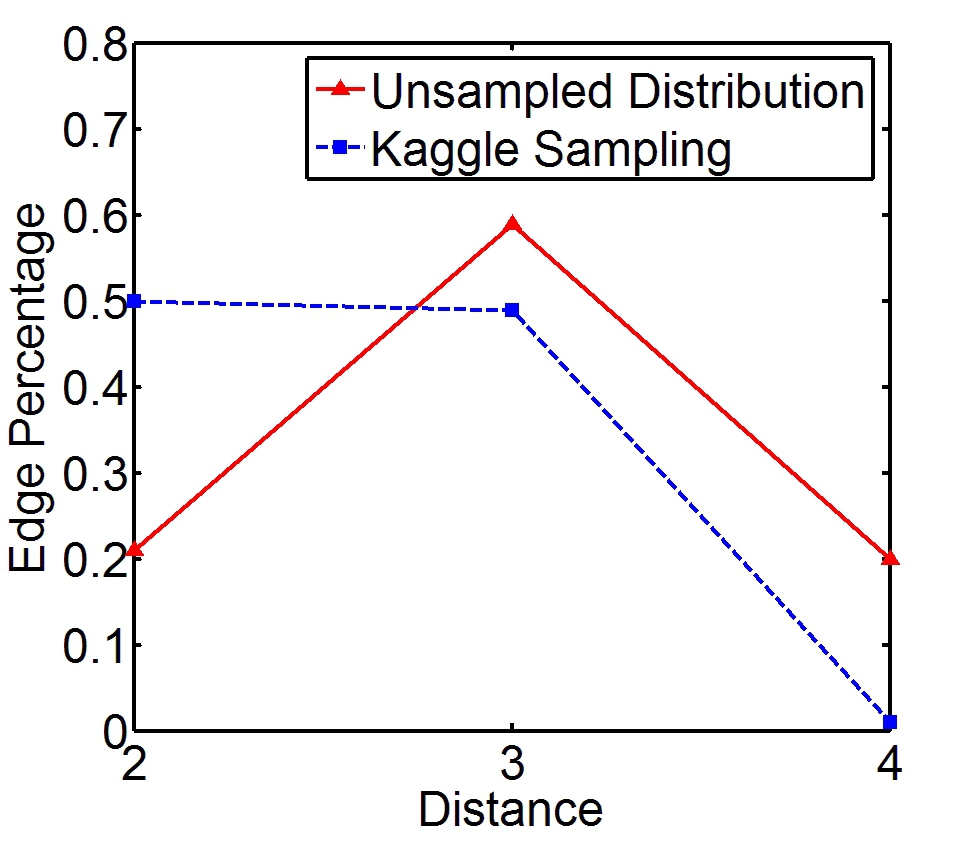}
	}
	\caption{Comparison between fair random sampling and Kaggle sampling.}
	\label{fig:kaggle1}
\end{figure}

\begin{table}
\renewcommand{\arraystretch}{1.3}
\caption{ROC Area (AUROC)} 
\label{performance_metric}
\centering 
\begin{tabular}{|c|p{0.8cm}|p{0.8cm}|p{0.8cm}|p{0.8cm}|p{0.8cm}|p{0.8cm}|p{0.8cm}|p{0.8cm}|} 
\cline{1-9}
{} & \multicolumn{2}{c|}{Condmat} & \multicolumn{2}{c|}{DBLP} & \multicolumn{2}{c|}{Enron} & \multicolumn{2}{c|}{Facebook}\\
\cline{2-9}
{Predictor} & {Fair} & {Kaggle} & {Fair} & {Kaggle} & {Fair} & {Kaggle} & {Fair} & {Kaggle}\\
\cline{1-9}
Pref. Attach. & 0.79 & \textbf{ 0.88} & 0.65 & \textbf{ 0.86} & 0.94 & \textbf{ 0.97} & 0.86 & \textbf{ 0.98}  \\ \cline{1-9}
PropFlow & 0.85 & \textbf{ 0.86} & 0.76 & \textbf{ 0.89} & 0.95 & \textbf{ 0.99} & 0.80 & \textbf{ 0.99} \\ \cline{1-9}
Adamic Adar & 0.63 & \textbf{ 0.63} & 0.62 & \textbf{ 0.63} & 0.79 & \textbf{ 0.80} & 0.66 & \textbf{ 0.75} \\ \cline{1-9}
\end{tabular}
\end{table}

\section{New Nodes}
\label{sec:newnodes}
There are two fundamentally different ways to generate test sets in link prediction. The first is to create a set of potential links by examining the predictor network and selecting all pairs for which no edge exists. Positives are those among the set that subsequently appear in the testing network, and negatives are all others. The second is to use the testing network to generate the set of potential links. Positives are those that exist in the testing network but not in the training network, and negatives are those that could exist in the testing network but do not. The subtle difference lies in whether or not the prediction method is faced with or penalized for links that involve \textit{nodes} that do not appear during training time.

The choice we should make depends on how the problem is posed. If we are faced with the problem of returning a most confident set of predictions, then new nodes in the testing network are irrelevant. Although we could predict that an existing node will connect to an unobserved node, we cannot possibly predict what node the unobserved node will be.

If we are faced with the problem of answering queries, then the ability to handle new nodes is an important aspect of performance. On one hand, we could offer a constant response, either positive or negative, to all queries regarding unfamiliar nodes. The response to offer and its effect on performance depend on the typical factors of cost and expected class distribution. On the other hand, some prediction methods may support natural extensions to provide a lesser amount of information in such cases. For instance, preferential attachment could be adapted to assume a degree of 1 for unknown nodes. Path-based predictors would have no basis to cope with this scenario whatsoever. In supervised classification, any such features become missing values and the algorithm must support such values.

Evaluating with potential links drawn from the testing network is problematic for decomposing the problem by distance since the distance must be computed from single-source shortest paths based on the pretend removal of the link that appears only in the testing network. Since distance is such a crucial player in determining link likelihood in most networks, this would nonetheless be an early step in making a determination about link formation likelihood in any case, so its computation for creating divided test sets is probably unavoidable. Given the extra complexity introduced by using potential link extraction within the testing network, we opt for determining link pairs for testing based on training data unless there is a compelling reason why this is unsatisfactory. This decision only has the potential to exclude links that are already known to be impossible to anticipate from training data, so it necessarily has the same effect across any set of predictors.

\section{Top $K$ Predictive Rate}
\label{sec:topk}
Though we caution trusting results that come only in terms of fixed thresholds metrics, some of these fixed thresholds metrics have significant real-world applications. A robust threshold curve metric exhibits the trade-off between sensitivity and specificity. A desirable property of a good fixed-threshold metric is that higher score implies an increase both in \textit{sensitivity} and \textit{specificity}. In this section we discuss the \textit{top $K$ predictive rate} \cite{liben-nowell:2007}, which we shall write as $TPR_{K}$, and explore its evaluation behavior in the link prediction problem. Top $K$ equivalent evaluation metrics have been discussed previously in the work of \cite{omadadhain:2005a,huang:2005,wang:2007}, and this measure is well-known as R-precision from information retrieval. We provide a proof of one property of $TPR_{K}$ that is important in link prediction. Based on this proof, we explore the restrictions of $TPR_{K}$ in evaluating the link prediction performance.

We denote the set of \textit{true positives} as $\mathbf{TP}$, the set of \textit{true negatives} as $\mathbf{TN}$, the set of \textit{false positives} as $\mathbf{FP}$, the set of \textit{false negatives} as $\mathbf{FN}$, the set of all positive instances as $\mathbf{P}$, and the set of all negative instances as $\mathbf{N}$.
\begin{defi}
$TPR_{K}$ is the percentage of correctly classified positive samples among the top $K$ instances in the ranking by a specified link predictor $\mathcal{P}$.
\end{defi}
This metric has the following property:
\begin{theorem}
\label{topk_linear}
When $K = |\mathbf{P}|$ in the link prediction problem, \textit{sensitivity} and \textit{specificity} are linearly dependent on $TPR_{K}$.
\end{theorem}

\begin{proof}
By definition, we know:
\begin{align*}
TPR_{K} = \frac{|\mathbf{TP}|}{K}, \text{sensitivity} = \frac{|\mathbf{TP}|}{|\mathbf{P}|}
\end{align*}
So the equivalence $K = |\mathbf{P}|$ allows us to trivially conclude that $TPR_{K}$ and $\text{sensitivity}$ are identical. We can write \textit{specificity} as:
\begin{align*}
\text{specificity} = \frac{|\mathbf{TN}|}{|\mathbf{N}|}
\end{align*}
When $K = |\mathbf{P}|$, by definition $|\mathbf{FP}| = K - |\mathbf{TP}|$, because we predict that all top $K$ are positive instances, and we can conclude that:
\begin{align*}
&|\mathbf{P}| - |\mathbf{TP}| = |\mathbf{FP}| \\
&\rightarrow (1 - TPR_{K}) |\mathbf{P}| = |\mathbf{FP}| \\
&\rightarrow |\mathbf{TN}| + (1 - TPR_{K}) |\mathbf{P}| = |\mathbf{TN}| + |\mathbf{FP}| \\
&\rightarrow \frac{|\mathbf{TN}|}{|\mathbf{N}|} + (1 - TPR_{K}) |\mathbf{P}| \frac{1}{|\mathbf{N}|} = 1 \\
&\rightarrow \text{specificity} = 1- (1 - TPR_{K}) \frac{ |\mathbf{P}|}{|\mathbf{N}|}
\end{align*}
From this, we see that \textit{specificity} increases monotonically with the increase of $TPR_{K}$, and is linearly dependent on $TPR_{K}$.
\end{proof}

\begin{figure}
	\centering
	\subfloat[Condmat]{
		\includegraphics[width=0.4\linewidth]{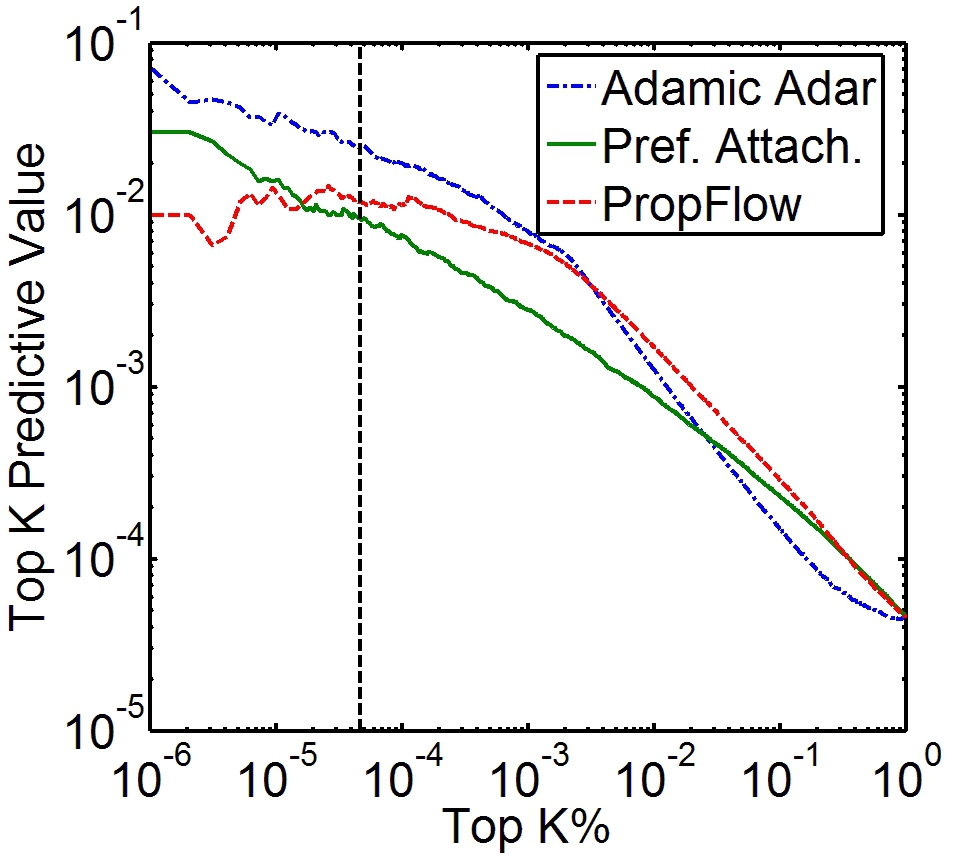}
	}
	\subfloat[DBLP]{
		\includegraphics[width=0.4\linewidth]{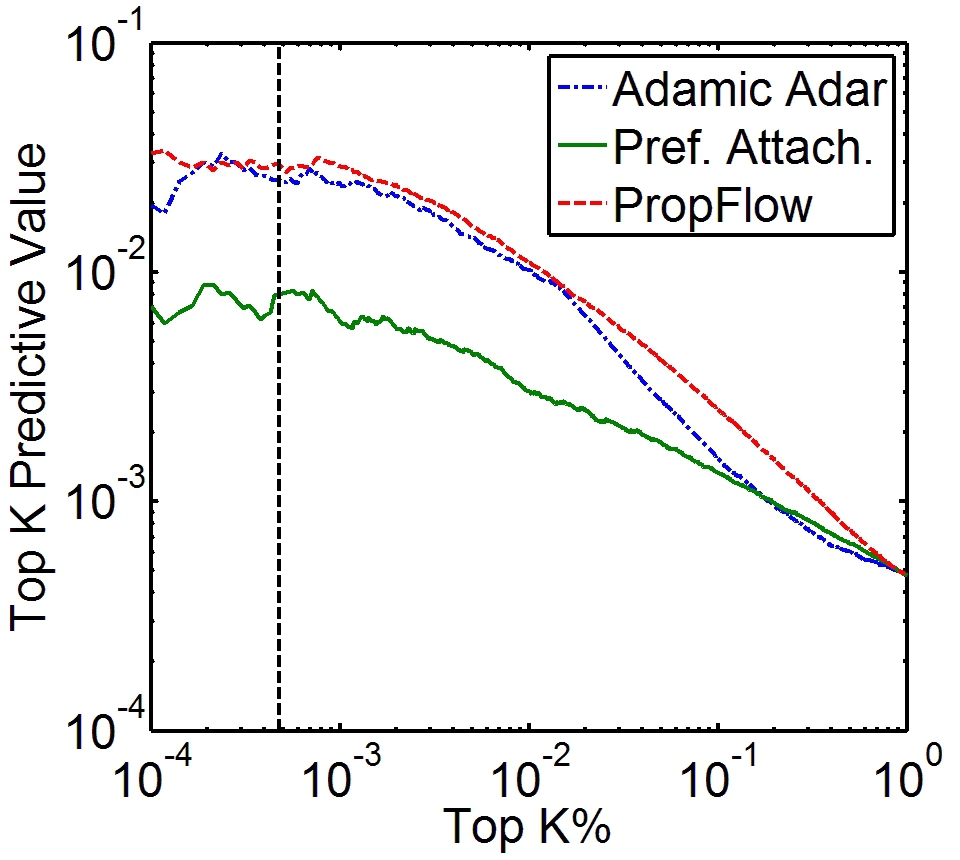}
	} \\
	\subfloat[Enron]{
		\includegraphics[width=0.4\linewidth]{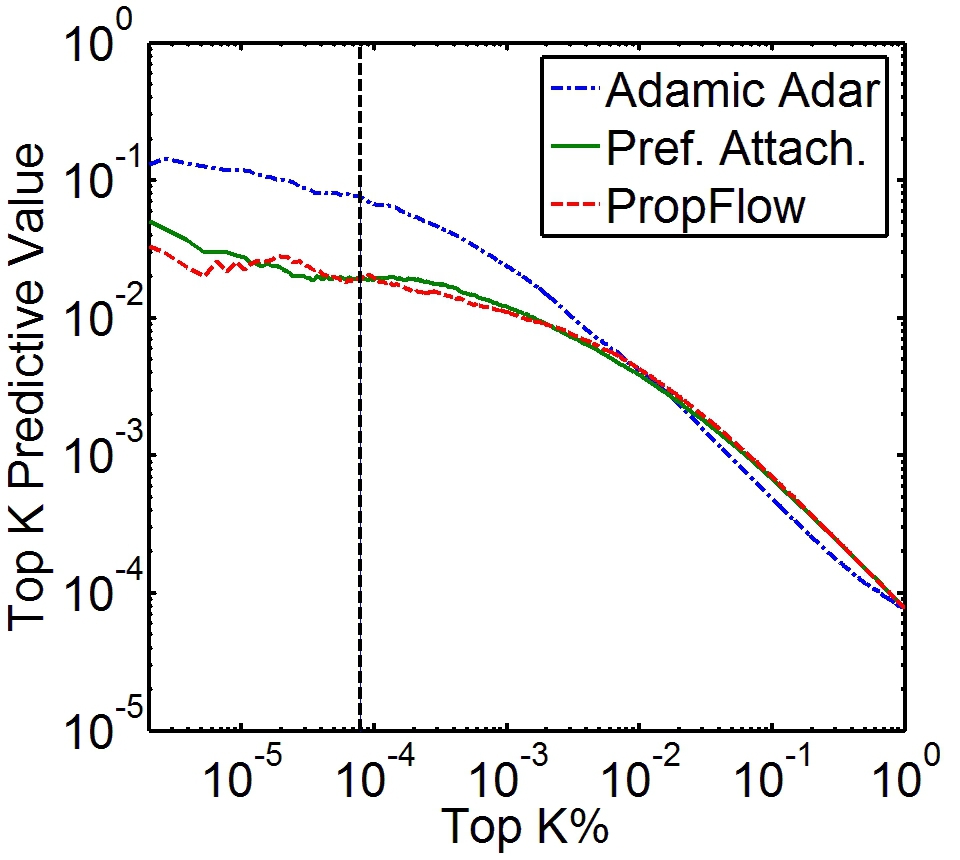}
	}
	\subfloat[Facebook]{
		\includegraphics[width=0.4\linewidth]{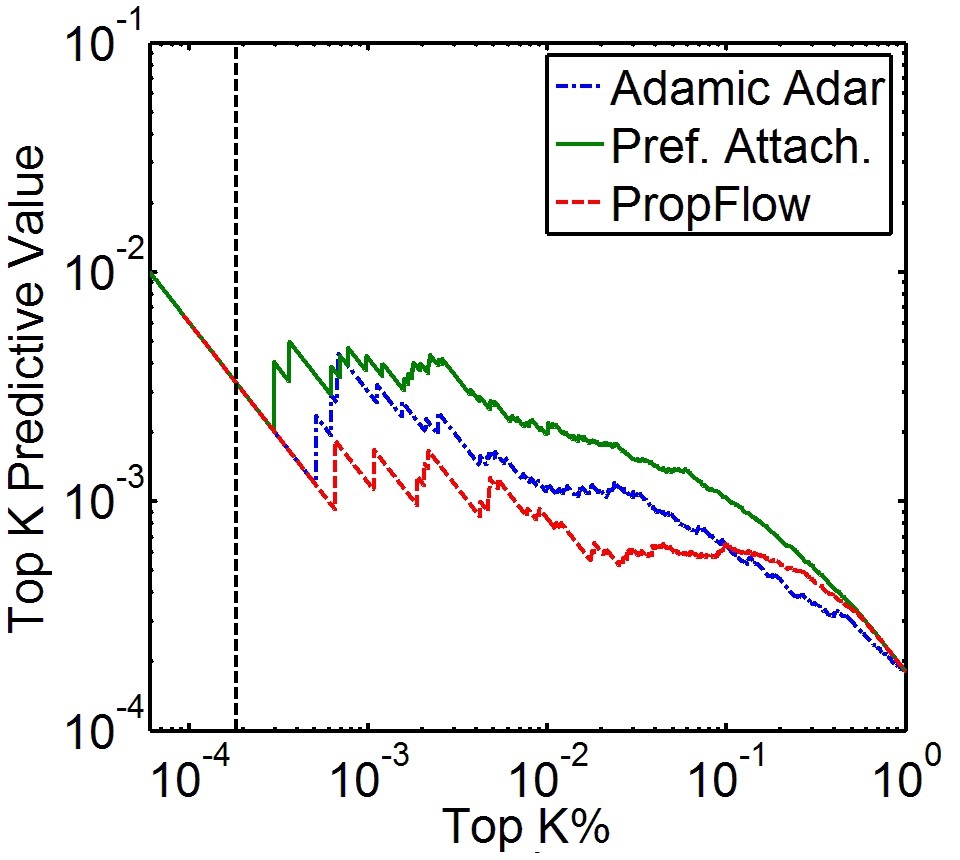}
	}	
	\caption{Top K Predictive Rate. In each figure the vertical line indicates that $K\% = \frac{|\mathbf{P}|}{|\mathbf{P}|+|\mathbf{N}|}$.}
	\label{fig:topk_metric}
\end{figure}

In Figure~\ref{fig:topk_metric} we provide the $TPR_{K}$ performance of predictors on Condmat, DBLP, Enron, and Facebook. 
The vertical line indicates the performance of a naive algorithm that draws samples as edges uniformly at random from all instances. Although the top $K$ predictive rate can provide a good performance estimation of link prediction methods when $K$ is appropriately selected, we still cannot recommended it as a primary measurement. If $|\mathbf{P}|$ is known, then $K$ may be set to that, but in real applications it is often impossible to know or even approximate the number of positive instances in advance, so $K$ is not specifiable. From the Figure~\ref{fig:topk_metric} we can see that different values of $K$ lead to different evaluations and even rankings of link prediction methods. Figure~\ref{fig:topk_metric} \textbf{(c)} shows that a small difference in $K$ will lead to a ranking reversal of the preferential attachment and PropFlow predictors. $TPR_{K}$ is a good metric for the link prediction task when the value of $K$ is appropriately selected, but evaluation results are too sensitive to use arbitrary $K$.

\section{The Case for Precision-Recall Curves}
\label{sec:prcurve}

ROC curves (and AUROC) are appropriate for typical data imbalance scenarios because they optimize toward a useful result and because the appearance of the curve provides a reasonable visual indicator of expected performance. One may achieve an AUROC of 0.99 in scenarios where data set sizes are relatively small ($10^3$ to $10^6$) and imbalance ratios are relatively modest (2 to 20). Corresponding precisions are near 1. For complete link prediction in sparse networks, when every potential new edge is classified, the imbalance ratio is \emph{lower} bounded by the number of vertices in the network \cite{lichtenwalter:2010}. ROC curves and areas can be deceptive in this situation. In a network with millions of vertices, even with an exceptional AUROC of 0.99, one could suffer small fractions as a maximal precision. Performance put in these terms is often considered unacceptable to researchers. In most domains, examining several million false positives to find each true positive \emph{is} the classification problem. Even putting aside more concrete theoretical criticisms of ROC curves and areas \cite{hand:2009}, in link prediction tasks they fail to honestly convey the difficulty of the problem and reasonable performance expectations for deployment. We argue for the use of precision-recall curves and AUPR in link prediction contexts.

\subsection{Geodesic Effect on Link Prediction Evaluation}
\label{sec:geodesic}
Precision-recall (PR) curves provide a more discriminative view of classification performance in extremely imbalanced contexts such as link prediction \cite{davis:2006}. Like ROC curves, PR curves are threshold curves. Each point corresponds to a different score threshold with a different precision and recall value. In PR curves, the x-axis is recall and the y-axis is precision. We will now revisit a problematic scenario that arose with AUROCs and demonstrate that AUPRs present a less deceptive view of predictor performance data. Notably, and compatible with our recommendations against sampling, PR curve construction procedures will require that negatives are not subsampled from the test set. This is not computationally problematic in the consideration of distance-restricted predictions.

\begin{figure}
	\centering
	\subfloat[PA-\textbf{Condmat}]{
		\includegraphics[width=0.23\linewidth]{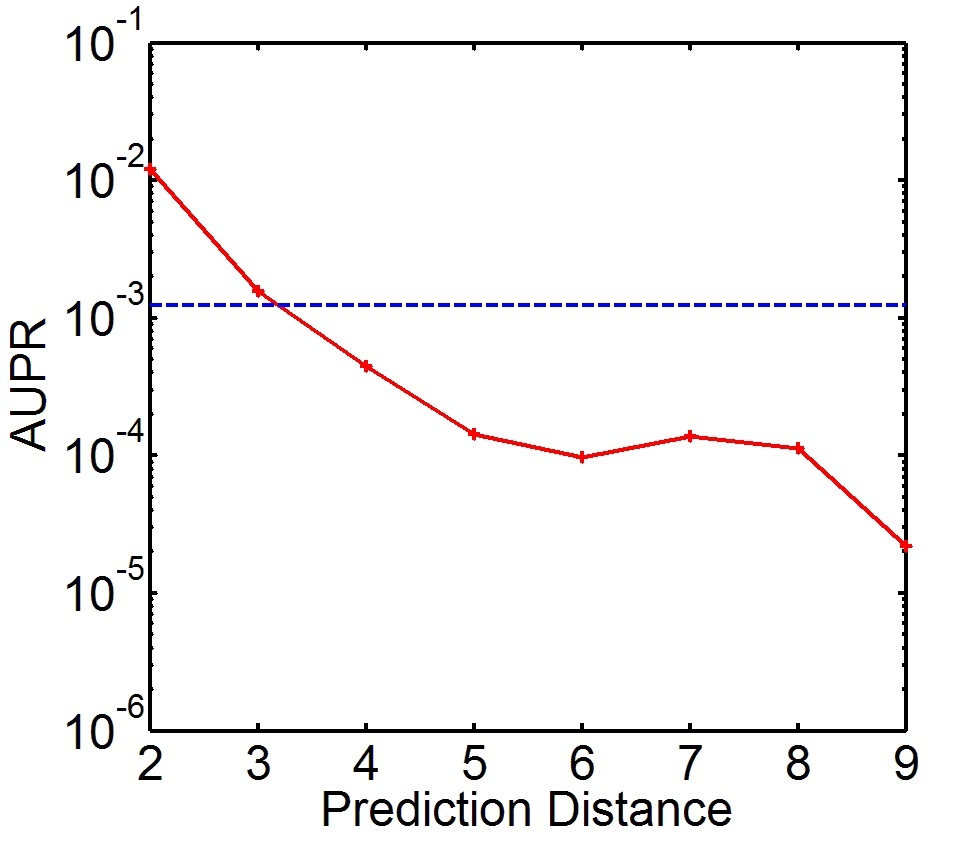}
	}
	\subfloat[PF-\textbf{Condmat}]{
		\includegraphics[width=0.23\linewidth]{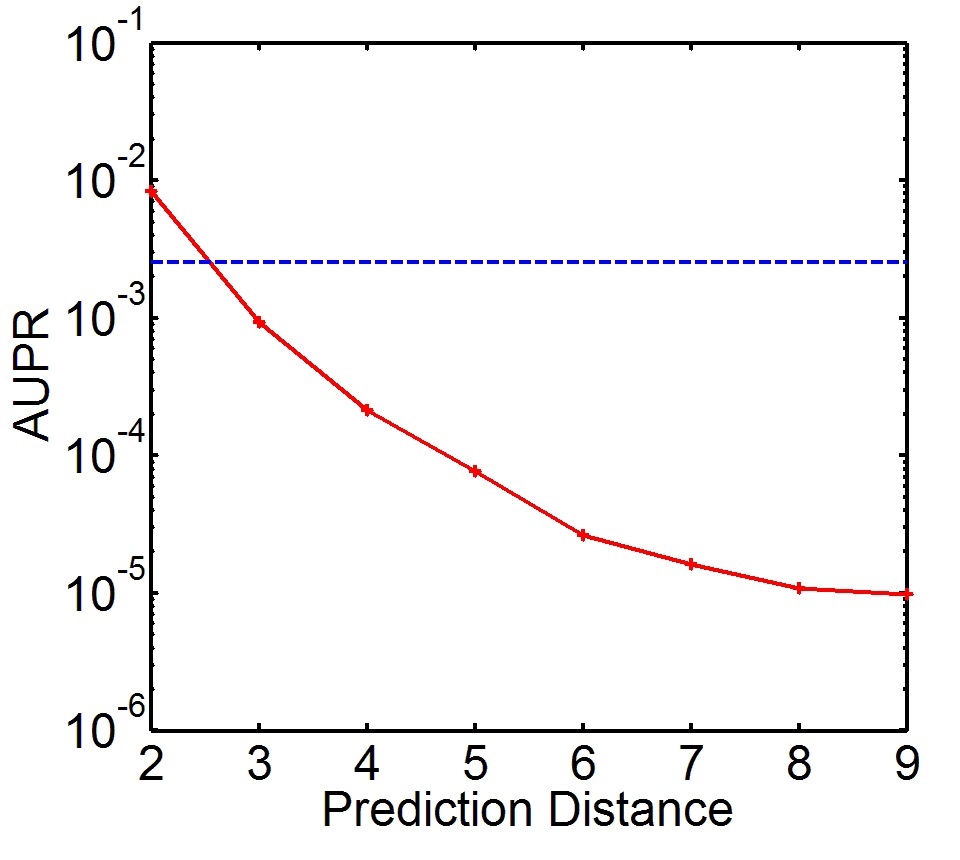}
	}
	\subfloat[PA-\textbf{DBLP}]{
		\includegraphics[width=0.23\linewidth]{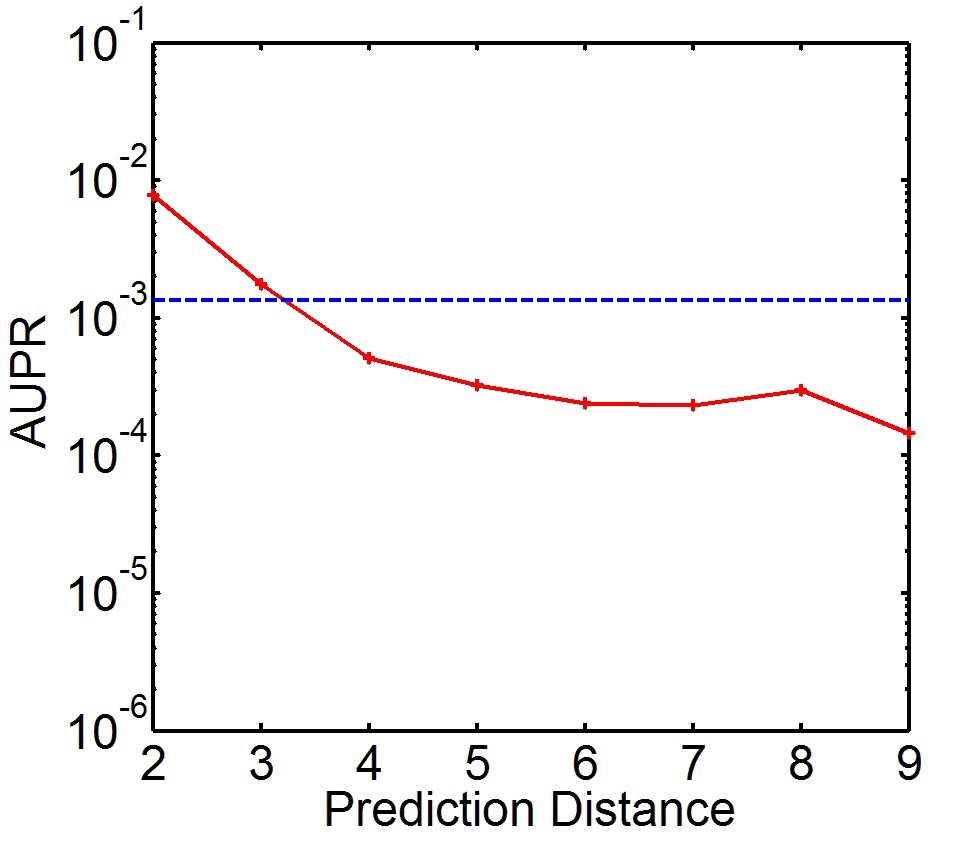}
	}
	\subfloat[PF-\textbf{DBLP}]{
		\includegraphics[width=0.23\linewidth]{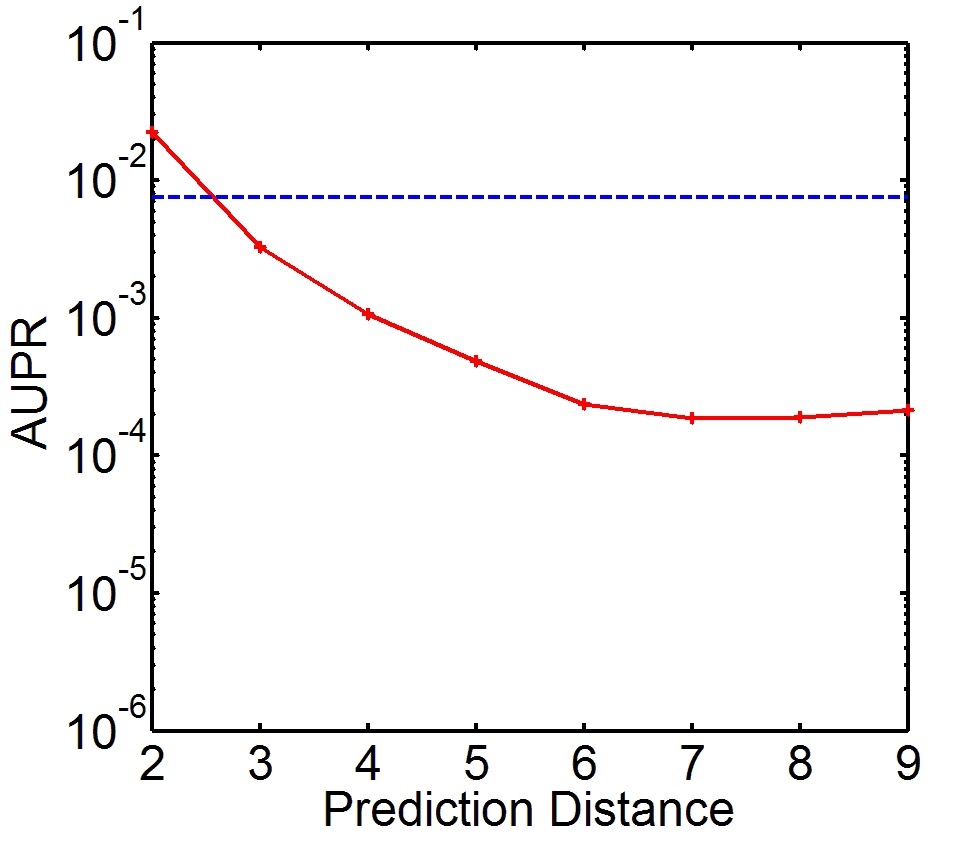}
	}\\
	\subfloat[PA-\textbf{Enron}]{
		\includegraphics[width=0.23\linewidth]{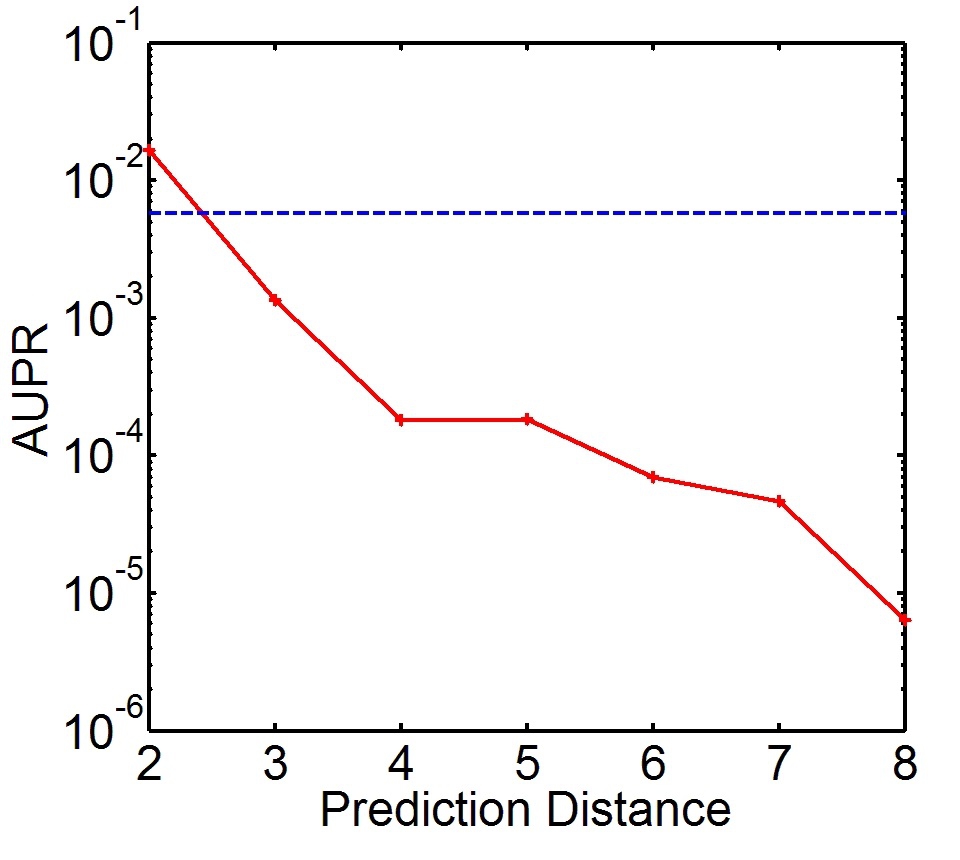}
	}
	\subfloat[PF-\textbf{Enron}]{
		\includegraphics[width=0.23\linewidth]{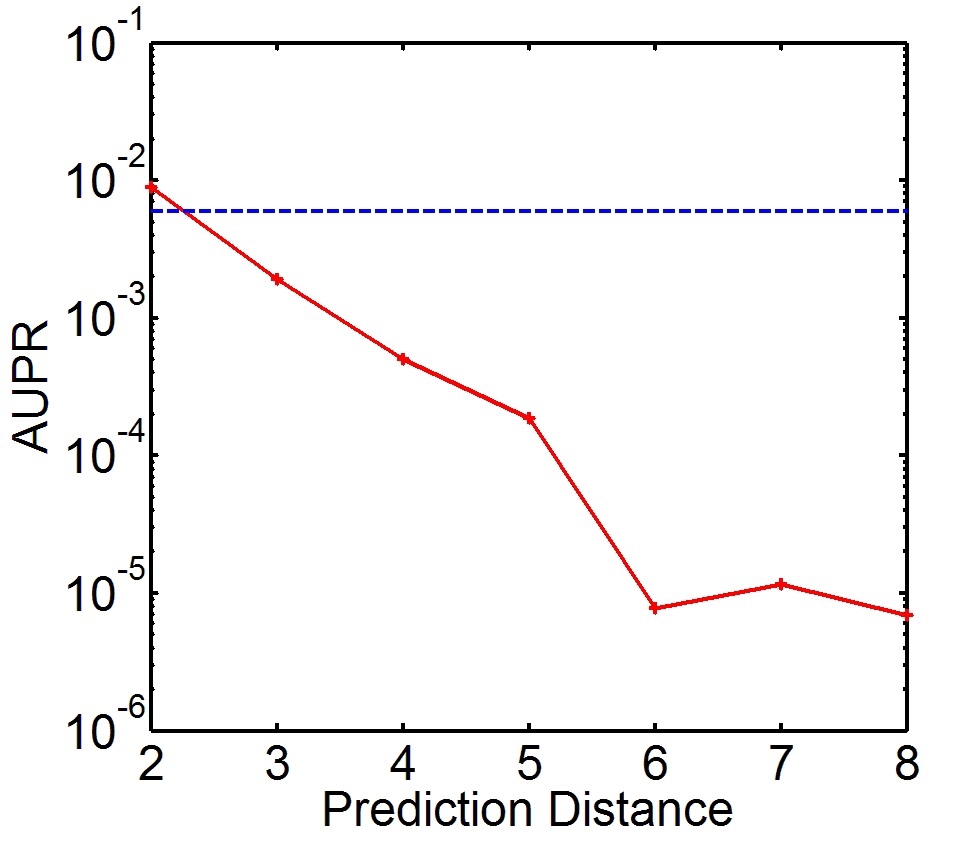}
	}
	\subfloat[PA-\textbf{Facebook}]{
		\includegraphics[width=0.23\linewidth]{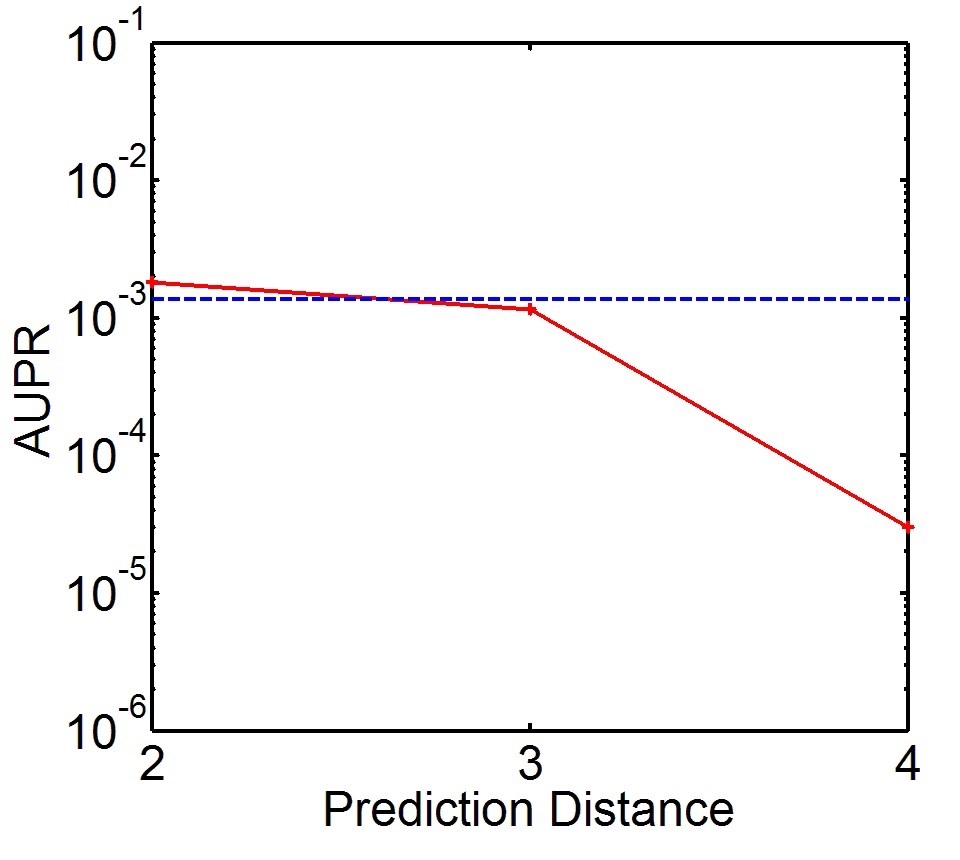}
	}
	\subfloat[PF-\textbf{Facebook}]{
		\includegraphics[width=0.23\linewidth]{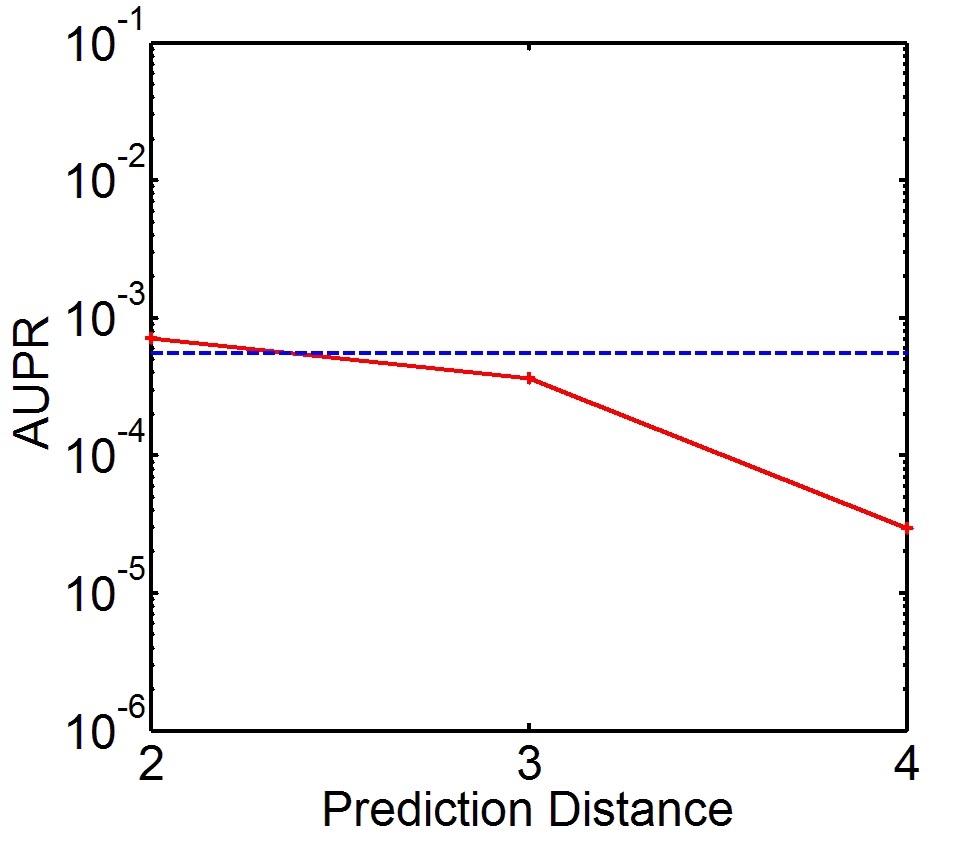}
	}	
	\caption{The precision-recall curve performance of the preferential attachment link predictor and PropFlow link predictor over each neighborhood. The horizontal line represents the performance apparent by considering all potential links as a corpus.}
	\label{fig:neighborhoods2}
		\vspace{-0.4cm}
\end{figure}
Figure \ref{fig:neighborhoods2} shows that the AUPR is higher for $\ell=2$ than it is for $\ell \leq \infty$. This also validates our proposition that increasing $\ell$ increases the difficulty of the prediction sub-problem due to increasing imbalance, which is made in Section~\ref{sec_sub_kaggle}. In the underlying curves, this is exhibited as much higher precisions throughout but especially for low to moderate values of recall. Performance by distance exhibits expected monotonic decline due to increasing baseline difficulty excluding the instabilities in very high distances. Compare this to Figure \ref{fig:neighborhoods} where the AUROC for all potential links was much greater than for any neighborhood individually, and the apparent performance was greatest in the 7-hop distance data set. We can also observe in Figure~\ref{fig:neighborhoods2} that the PR area increases almost monotonically with increasing $\ell$, which differs from the AUROC in Figure~\ref{fig:neighborhoods}.

\subsection{Temporal Effect on Link Prediction Evaluation}
\label{sec:temporal}
Time has remote yet non-negligible effects on the link prediction task. In this section we conduct two kinds of experiments. First, with a fixed training set, we compare AUROC and AUPR overall to AUROC and AUPR achievable in distinct sub-problems created by dividing the task by temporal distance when the test set is sliced into 5 subsets of the same duration. Second, with a fixed training set, we compare AUROC and AUPR overall across agglomerated sub-problems.

\begin{figure}
	\centering
	\includegraphics[width=0.88\linewidth]{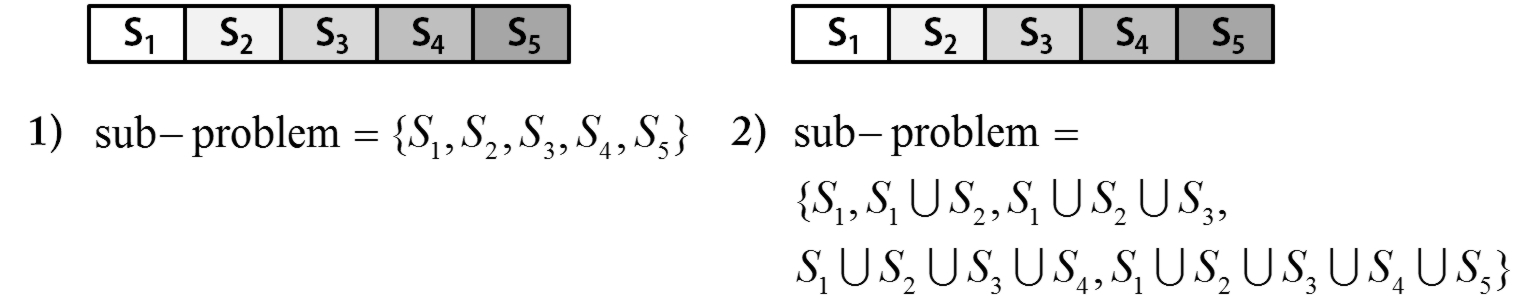}
	\caption{Experimental Configurations}
	\label{fig:experiment_setting}
\end{figure}

First we divide the testing set into 5 subsets of equal temporal duration. For instance, in DBLP there are 5 years of data in the testing set, so each year of data provides one sub-problem. Figure~\ref{fig:experiment_setting}-\textbf{1} gives an example of the experimental setting. As described in Figure~\ref{fig:experiment_setting}-\textbf{1} we denote these sub-problems as $S_{1}$, $S_{2}$, $S_{3}$, $S_{4}$, and $S_{5}$. $S_{i}$ increases in difficulty as $i$ increases, because the time series is not persistent \cite{yang:2012}, and the preponderance of a node to form new links decays exponentially \cite{leskovec:2008}. Based on this hypothesis the performance of a link predictor $\mathcal{P}$ on sub-problem $S_{i}$ should decrease with increasing $i$.

In Figure~\ref{fig:temporal_neighborhoods1} and Figure~\ref{fig:temporal_neighborhoods2} we provide ROC curves and PR curves. We revisit the deceptive nature of AUROCs and demonstrate that AUPRs present a less deceptive view of performance. Based on time series analysis, the performance of a predictor should decline in the presence of high temporal distance, but AUROCs fluctuate with increasing temporal distance. The AUPRs exhibit expected monotonic decline due to increasing baseline difficulty. This finding coincides with our results in Section~\ref{sec:sampling} and Section~\ref{sec:geodesic}, where AUROC values fluctuate with geodesic distance and AUPR values decrease monotonically.

\begin{figure}
	\centering
	\subfloat[PA-\textbf{DBLP}]{
		\includegraphics[width=0.23\linewidth]{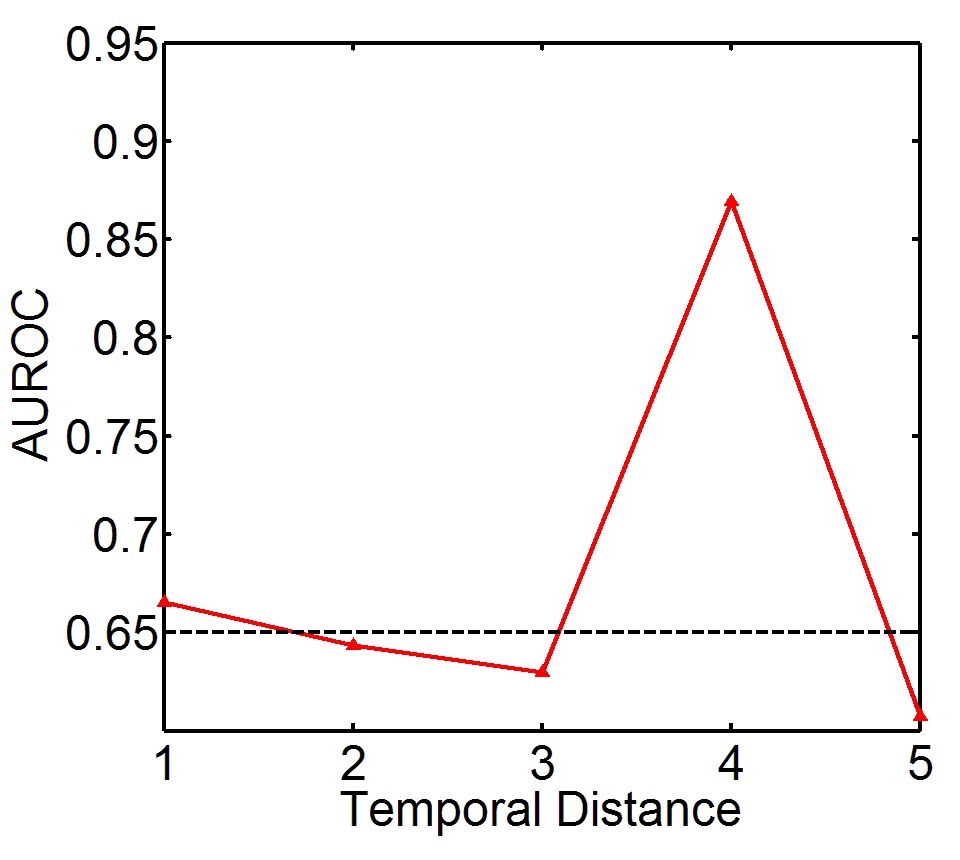}
	}
	\subfloat[PF-\textbf{DBLP}]{
		\includegraphics[width=0.23\linewidth]{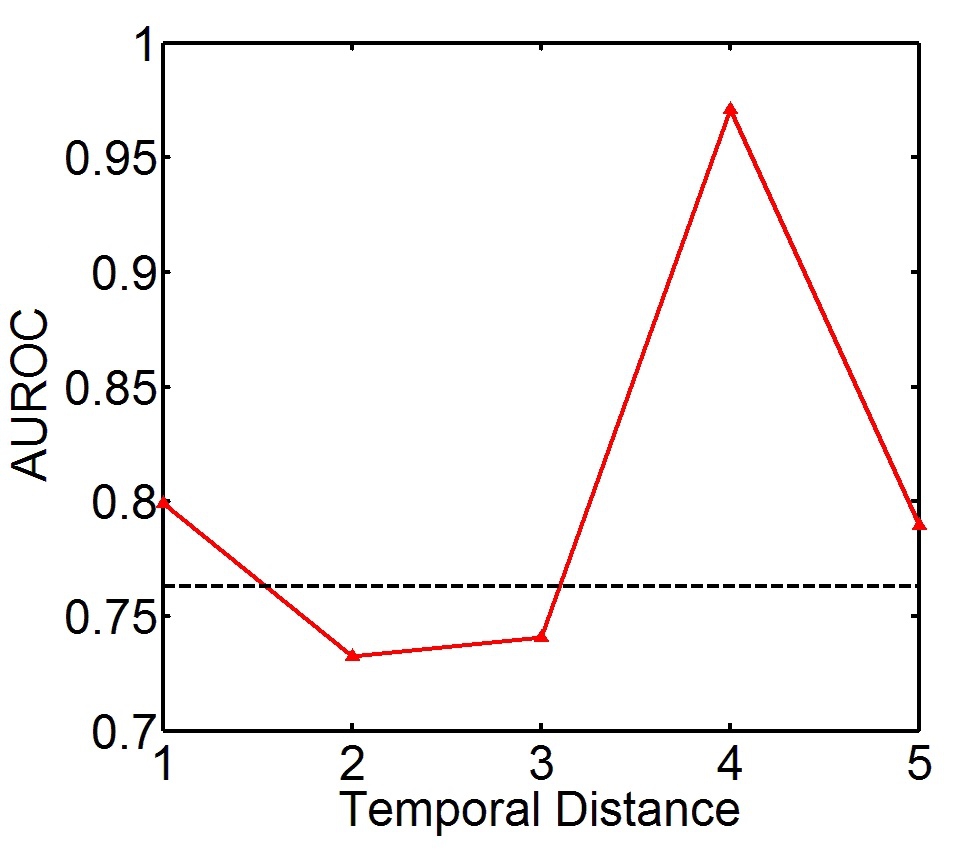}
	}
	\subfloat[PA-\textbf{Enron}]{
		\includegraphics[width=0.23\linewidth]{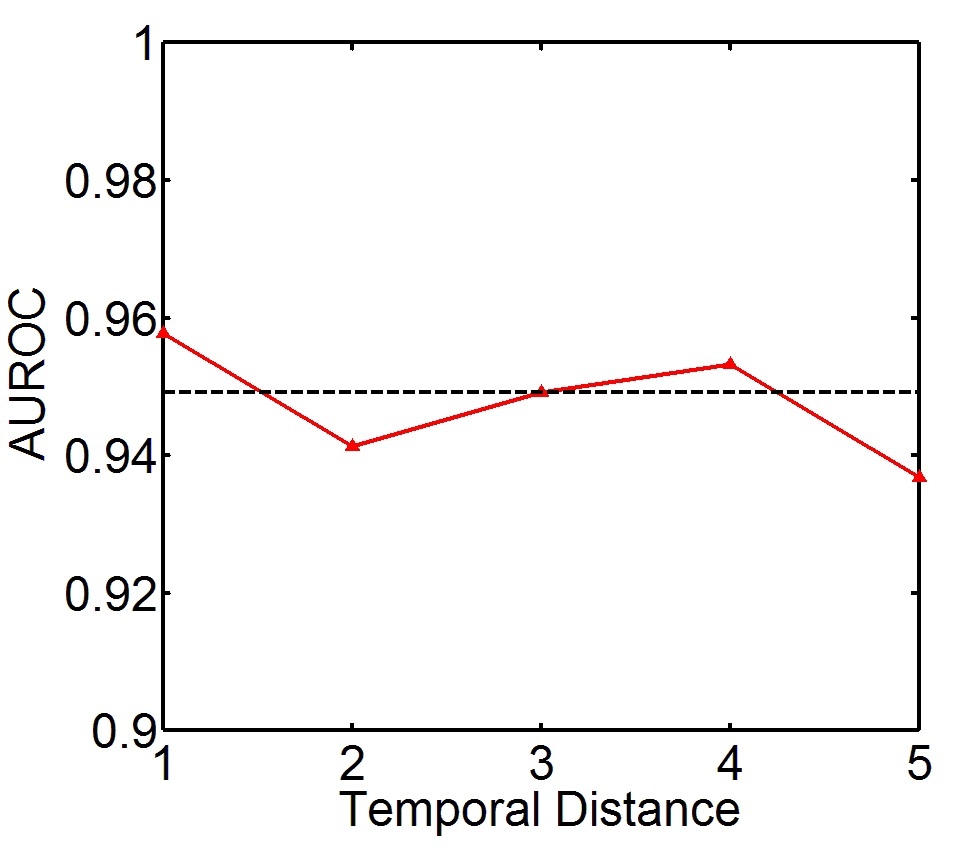}
	}
	\subfloat[PF-\textbf{Enron}]{
		\includegraphics[width=0.23\linewidth]{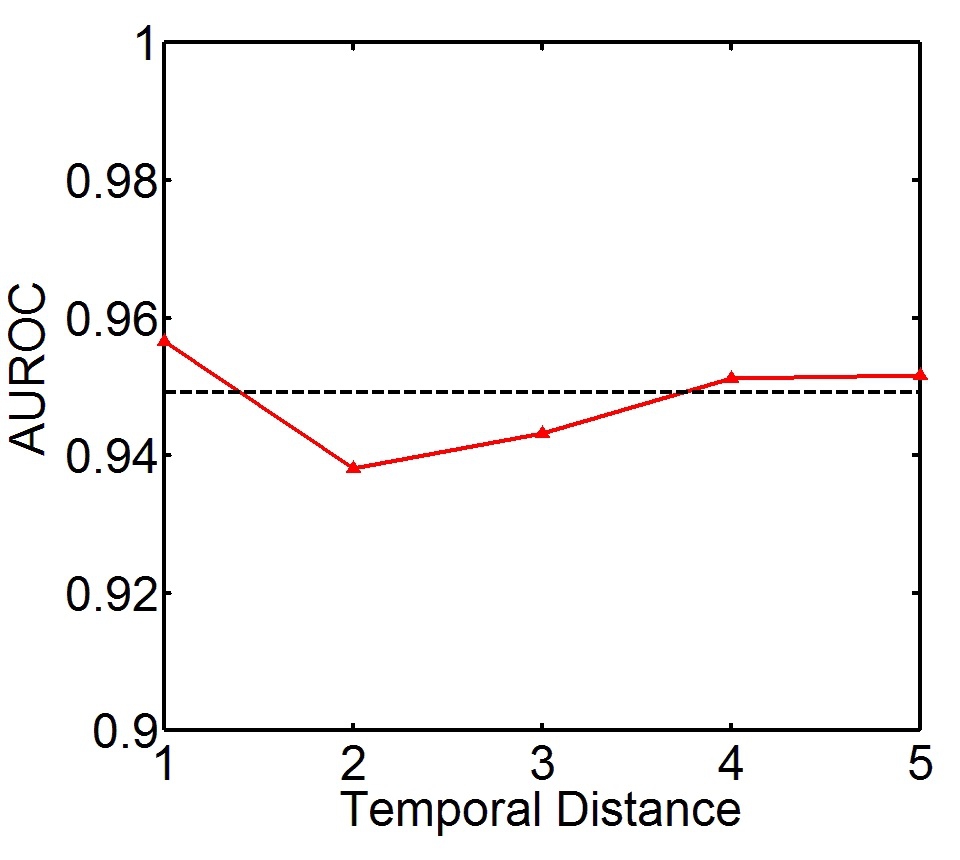}
	}
	\caption{The AUROC performance of preferential attachment and PropFlow over each temporal neighborhood. The horizontal line represents the performance apparent by considering all future potential links.}
	\label{fig:temporal_neighborhoods1}
\end{figure}

\begin{figure}
	\centering
	\subfloat[PA-\textbf{DBLP}]{
		\includegraphics[width=0.23\linewidth]{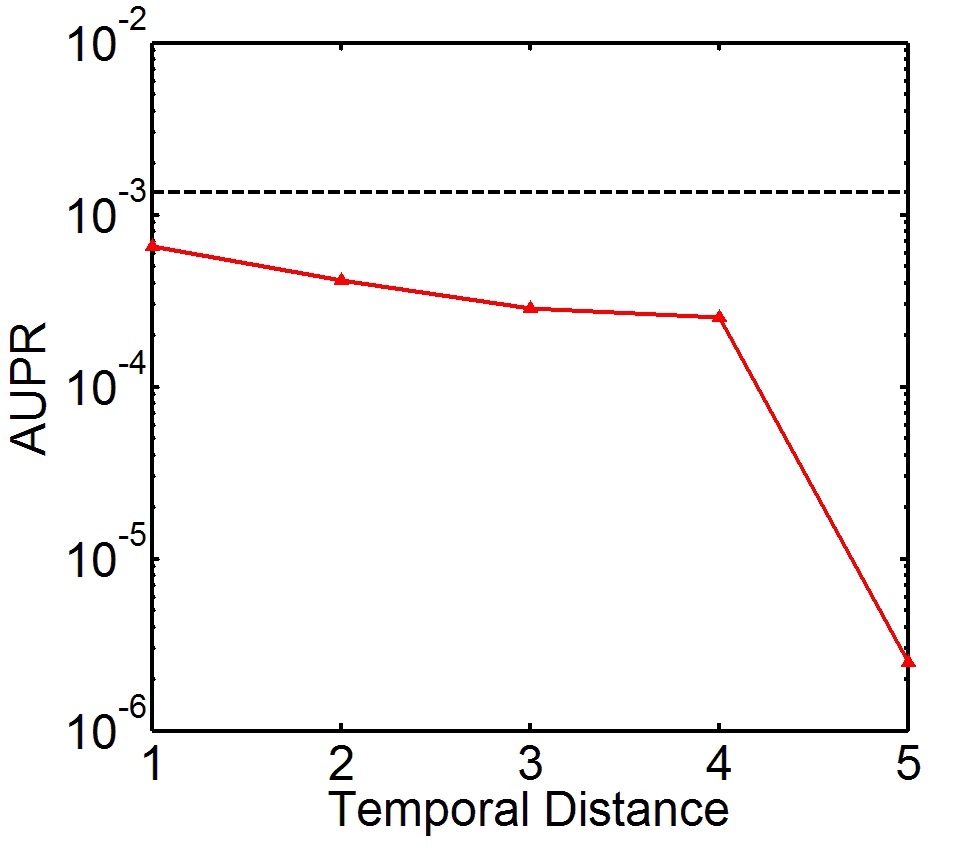}
	}
	\subfloat[PF-\textbf{DBLP}]{
		\includegraphics[width=0.23\linewidth]{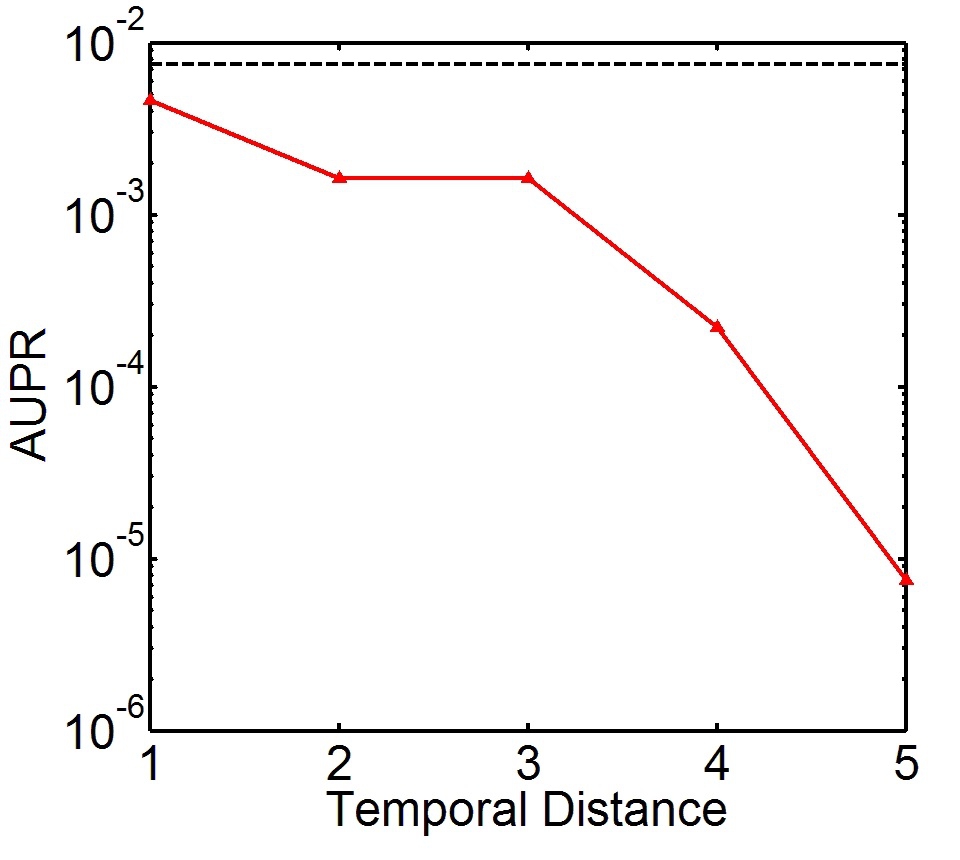}
	}
	\subfloat[PA-\textbf{Enron}]{
		\includegraphics[width=0.23\linewidth]{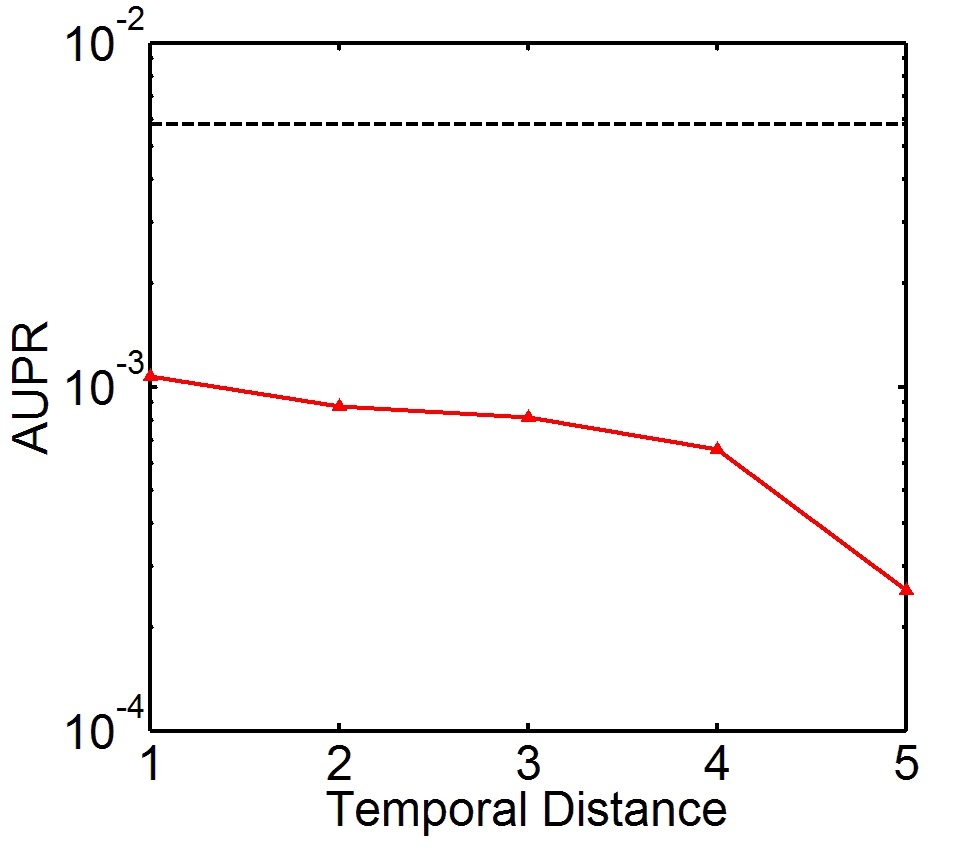}
	}
	\subfloat[PF-\textbf{Enron}]{
		\includegraphics[width=0.23\linewidth]{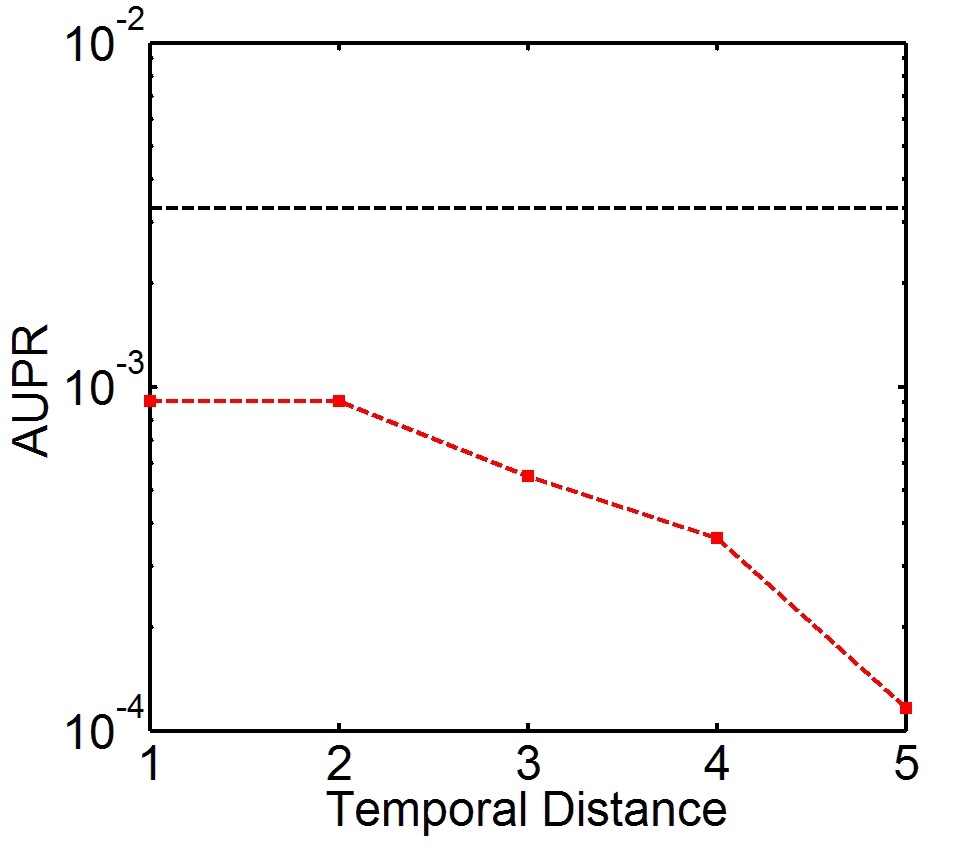}
	}
	\caption{The AUPR performance of preferential attachment and PropFlow over each temporal neighborhood. The horizontal line represents the performance apparent by considering all future potential links.}
	\label{fig:temporal_neighborhoods2}
\end{figure}

To emphasize these conclusions, we agglomerate the sub-problems over time as shown in Figure~\ref{fig:experiment_setting}-\textbf{2}. The sub-problem becomes easier when more years of data are included, because the total number of prediction candidates is fixed while the number of positive instances is increasing. In this way, the performance of a predictor should increase when the sub-problem includes more data. Figure~\ref{fig:temporal_incremental} shows distinct behaviors of AUROCs and AUPRs. The AUROC values decline or are unstable while the AUPR values for all three predictors increase monotonically.

\begin{figure}
	\centering
	\subfloat[ROC-\textbf{DBLP}]{
		\includegraphics[width=0.23\linewidth]{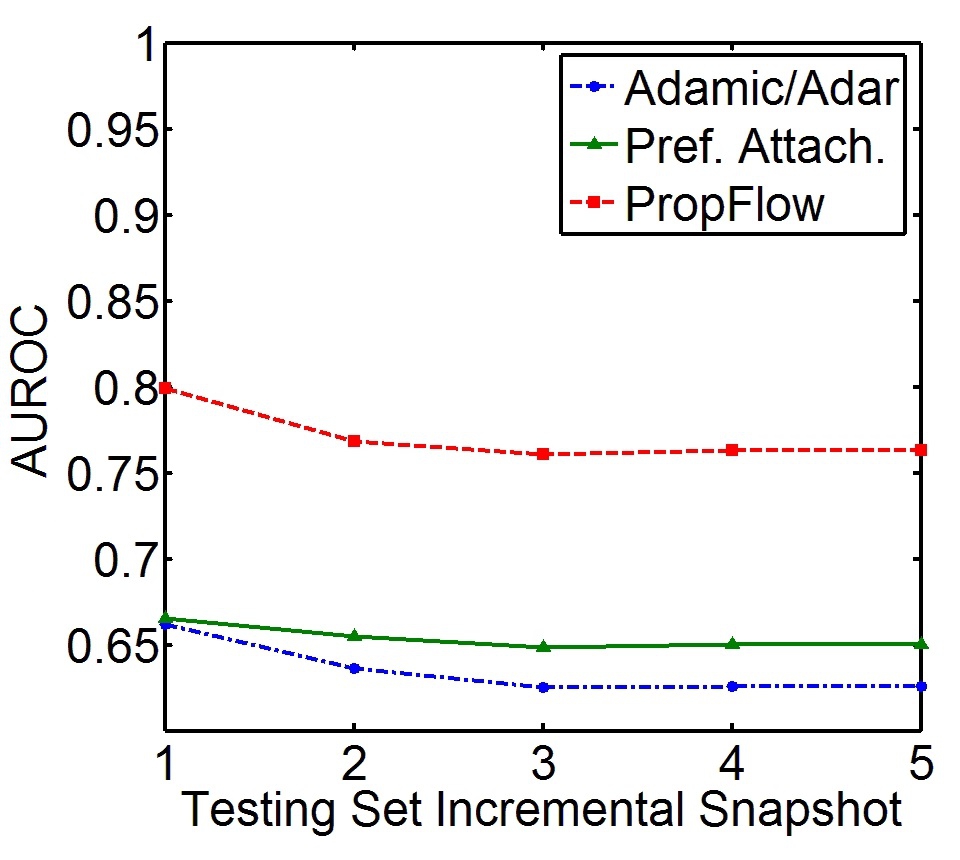}
	}
	\subfloat[PR Curve-\textbf{DBLP}]{
		\includegraphics[width=0.23\linewidth]{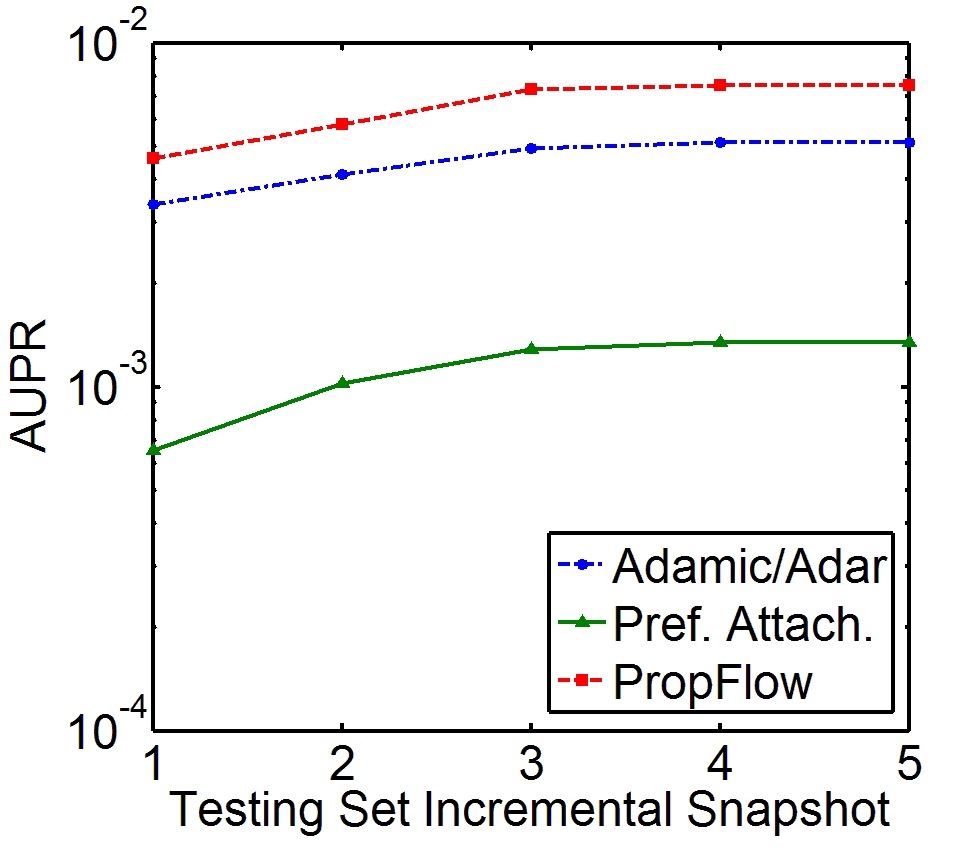}
	}
	\subfloat[ROC-\textbf{Enron}]{
		\includegraphics[width=0.23\linewidth]{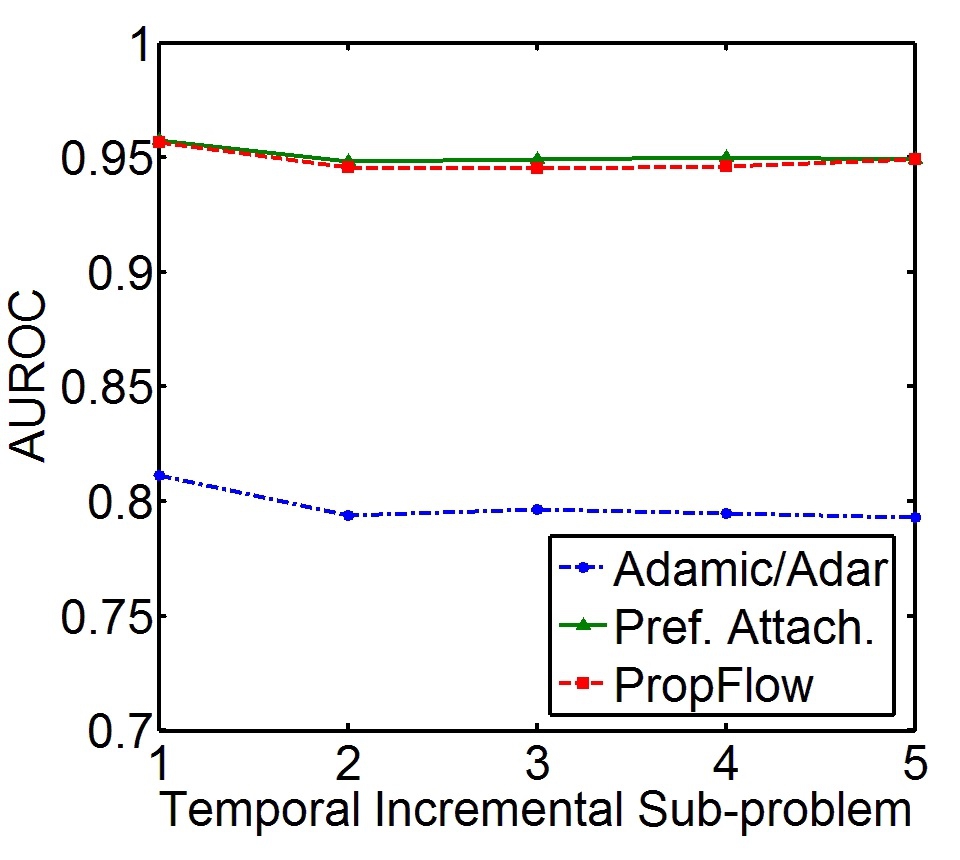}
	}
	\subfloat[PR Curve-\textbf{Enron}]{
		\includegraphics[width=0.23\linewidth]{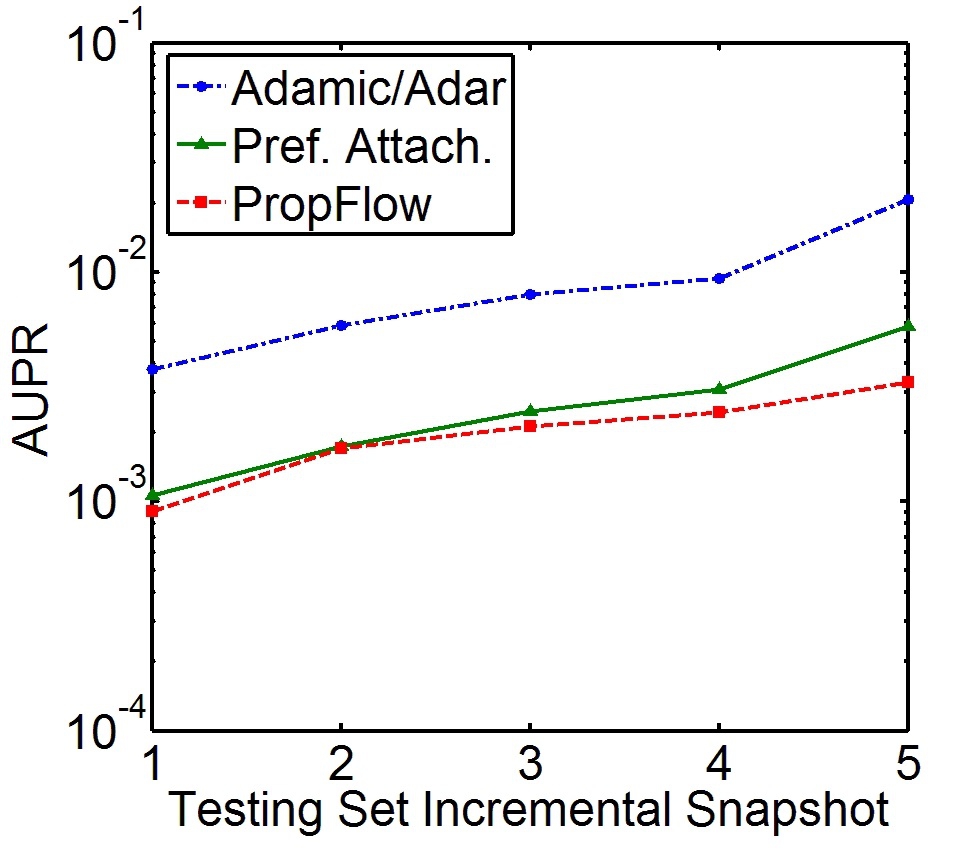}
	}
	\caption{The performance of incremental temporal neighborhood. The test set temporal distance is incremental. The sub-problem becomes easier when more years of data are included, because the total number of prediction candidates is fixed while the number of positive instances is increasing. The performance of a predictor should increase when the sub-problem includes more data.}
	\label{fig:temporal_incremental}
\end{figure}

\section{Conclusion}
\label{sec:conclusion}
To select the best predictor, we must know how to evaluate predictors. Beyond this, we must be sure that readers do not come away from papers with the question of how new methods \emph{actually} perform. It is more difficult to specify and explain link prediction evaluation strategies than with standard classification wherein it is sufficient to fully specify a data set, one of a few evaluation methods, and a given performance metric. In link prediction, there are many potential parameters often with many undesirable values. There is no question that the issues raised herein \emph{can} lead to questionable or misleading results. The theoretical and empirical demonstrations should convince the reader that they \emph{do} lead to questionable or misleading results. 

Much of this paper relies upon the premise that the class balance ratio differs, even differs wildly, across distances. There are certainly rare networks where such an expectation is tenuous, but the premise holds in every network with which the authors have worked including networks from the following families: biology, commerce, communication, collaboration, and citation.

Based on our observations and analysis, we propose the following guidelines:
\begin{enumerate}
	\item Use precision-recall curves and AUPR as an evaluation measure. In our experiments we observe that ROC curves and AUROC can be deceptive, but precision-recall curves yield better precision in evaluating the performance of link prediction (Section~\ref{sec_roc_robustness} and Section~\ref{sec:temporal}).
	\item Avoid fixed thresholds unless they are supplied by the problem domain. We identify limitations and drawbacks of fixed thresholds metrics, such as top $K$ predictive rate (Section~\ref{sec:topk}).
	\item Render prediction performance evaluation by geodesic distance. In Section~\ref{sec:geo_sub_problem}, Section~\ref{sec:prcurve}, and Section~\ref{sec:temporal} we observe that the performance of sub-problems created by dividing the link prediction task by geodesic distance is significantly different from overall link prediction performance. In cases where temporal distance is a significant component of prediction, consider also rendering performance by temporal distance.
	\item Do not undersample negatives from test sets, which will be of more manageable size due to consideration by distance. Experiments and proofs both demonstrate that undersampling negatives from test sets can lead to inaccurately measured performance and incorrect ranking of link predictors (Section~\ref{sec:sampling}).
	\item If negative undersampling \emph{is} undertaken for some reason, it must be based on a purely random sample of edges missing from the test network. It must not modify dimensions in the original distribution. Naturally \emph{any} sampling must be \emph{clearly} reported. Inappropriate methods of sampling will lead to incorrect measures of link prediction performance, a fact demonstrated by Kaggle sampling as analyzed in Section~\ref{sec_sub_kaggle}.
	\item In undirected networks, state if a method is invariant to designations of source and target. If it is not, state how the final output is produced. As we discover in Section~\ref{sec:directionality} different strategies overcoming the directionality issue lead to different judgments about performance.
	\item Always take care to use the same testing set instances regardless of the nature of the prediction algorithm.
	\item In temporal data, the final test set on which evaluation is performed should receive labels from a subsequent, unobserved snapshot of the data stream generating the network (Section~\ref{sec:temporal}).
	\item Consider whether the link prediction task set forth is to solve the recommendation problem or the query problem and construct test sets accordingly (Section~\ref{sec:newnodes}).
\end{enumerate}


\begin{acknowledgements}
Research was sponsored in part by the Army Research Laboratory (ARL) and was accomplished under Cooperative Agreement Number W911NF-09-2-0053, and in part from grant \#FA9550-12-1-0405 from the U.S. Air Force Office of Scientific Research (AFOSR) and the Defense Advanced Research Projects Agency (DARPA). The views and conclusions contained in this document are those of the authors and should not be interpreted as representing the official policies, either expressed or implied, of the ARL, AFOSR, DARPA or the U.S. Government. The U.S. Government is authorized to reproduce and distribute reprints for Government purposes notwithstanding any copyright notation hereon.
\end{acknowledgements}

\label{lastpage}
\end{document}